\documentclass[format=acmsmall, review=false, screen=true]{acmart}
\setcopyright{acmlicensed}
\acmJournal{TKDD}
\acmVolume{X}
\acmNumber{Y}
\acmArticle{Z}
\acmYear{2020}
\acmMonth{8}
\copyrightyear{2020}

\usepackage{graphicx}
\usepackage{epsfig}
\usepackage{url}
\usepackage{hyperref}
\usepackage{caption}
\usepackage[ruled, vlined, linesnumbered]{algorithm2e}
\usepackage[show]{chato-notes}
\usepackage{xspace}
\usepackage{multirow} 
\usepackage{hyperref}
\usepackage{amsmath}
\usepackage{tikz}
\usetikzlibrary{trees}
\usetikzlibrary{automata,positioning}

\hypersetup{
    bookmarks=true,         
    unicode=false,          
    pdftoolbar=true,        
    pdfmenubar=true,        
    pdffitwindow=false,     
    pdfstartview={FitH},    
    pdftitle={My title},    
    pdfauthor={Author},     
    pdfsubject={Subject},   
    pdfcreator={Creator},   
    pdfproducer={Producer}, 
    pdfnewwindow=true,      
    colorlinks=true,       
    linkcolor=black,          
    citecolor=black,        
    filecolor=black,      
    urlcolor=black           
}

\newtheorem{mydefinition}{Definition}
\newtheorem{myproposition}{Proposition}
\newtheorem{mytheorem}{Theorem}
\newtheorem{mycorollary}{Corollary}
\newtheorem{myexample}{Example}
\newtheorem{mylemma}{Lemma}
\newtheorem{problem}{Problem}
\newtheorem{fact}{Fact}
\newtheorem{myobservation}{Observation}
\newcommand{\core}{C_{k, \Delta}}

\newcommand{\baseline}{\textsf{Na\"ive-span-cores}}
\newcommand{\cores}{\textsf{Span}-\textsf{cores}}
\newcommand{\innermosts}{\textsf{Maximal-span-cores}}
\newcommand{\baselineinnermosts}{\textsf{Na\"ive-maximal-span-cores}}
\newcommand{\innermost}{\textsf{innermost-core}}

\newcommand{\spancore}{span-core\xspace}
\newcommand{\Spancore}{Span-core\xspace}
\newcommand{\spancores}{\spancore{s}\xspace}

\newcommand{\temporalcs}{\textsc{Temporal Community Search}\xspace}
\newcommand{\singletemporalcs}{\textsc{Single Temporal Community Search}\xspace}
\newcommand{\seqseg}{\textsc{Sequence Segmentation}\xspace}
\newcommand{\maxspancoreprob}{\textsc{Maximal \Spancore Mining}\xspace}
\newcommand{\greedy}{\textsf{Greedy-minimum-community-search}\xspace}
\newcommand{\cs}{\textsf{Temporal-community-search}\xspace}
\newcommand{\mcs}{\textsf{Efficient-temporal-community-search}\xspace}

\newcommand{\deepwalk}{\textsf{DeepWalk}\xspace}
\newcommand{\LINE}{\textsf{LINE}\xspace}
\newcommand{\nodevec}{\textsf{node2vec}\xspace}
\newcommand{\tcs}{\textsf{TCS}\xspace}

\newcommand{\coresset}{\mathbf{C}}
\newcommand{\coressetdelta}{\mathbf{C}_\Delta}

\newcommand{\imcores}{\mathbf{C}_M}
\newcommand{\Qdeltahighest}{C^*_{Q, \Delta}}
\newcommand{\Qdeltascore}{v^*_{Q, \Delta}}
\DeclareMathOperator*{\visited}{visited}

\newcommand{\bigO}{\mathcal{O}}
\newcommand{\tdeg}{\mbox{\ensuremath{d}}}
\newcommand{\score}{\ensuremath{score}}
\newcommand{\scorep}{\ensuremath{score^+}}
\newcommand{\scorem}{\ensuremath{score^-}}
\newcommand{\neigh}{\ensuremath{neigh_\Delta}}

\newcommand{\spara}[1]{\smallskip\noindent{\bf #1}}

\newcommand{\squishlist}{
 \begin{list}{$\bullet$}
  {  \setlength{\itemsep}{0pt}
     \setlength{\parsep}{3pt}
     \setlength{\topsep}{3pt}
     \setlength{\partopsep}{0pt}
     \setlength{\leftmargin}{2em}
     \setlength{\labelwidth}{1.5em}
     \setlength{\labelsep}{0.5em}
} }
\newcommand{\squishlisttight}{
 \begin{list}{$\bullet$}
  { \setlength{\itemsep}{0pt}
    \setlength{\parsep}{0pt}
    \setlength{\topsep}{0pt}
    \setlength{\partopsep}{0pt}
    \setlength{\leftmargin}{2em}
    \setlength{\labelwidth}{1.5em}
    \setlength{\labelsep}{0.5em}
} }

\newcommand{\squishdesc}{
 \begin{list}{}
  {  \setlength{\itemsep}{0pt}
     \setlength{\parsep}{3pt}
     \setlength{\topsep}{3pt}
     \setlength{\partopsep}{0pt}
     \setlength{\leftmargin}{1em}
     \setlength{\labelwidth}{1.5em}
     \setlength{\labelsep}{0.5em}
} }

\newcommand{\squishend}{
  \end{list}
}

\begin{document}
\title[Span-core Decomposition for Temporal Networks: Algorithms and Applications]{Span-core Decomposition for Temporal Networks: Algorithms and Applications}

\author{Edoardo Galimberti}
\affiliation{
\institution{ISI Foundation}
\city{Turin}
\country{Italy}}
\affiliation{
\institution{University of Turin}
\city{Turin}
\country{Italy}}
\email{edoardo.galimberti@isi.it}

\author{Martino Ciaperoni}
\affiliation{
\institution{Aalto University}
\city{Aalto}
\country{Finland}}
\email{martino.ciaperoni@aalto.fi}

\author{Alain Barrat}
\affiliation{
\institution{Aix Marseille Univ, Universit\'e de Toulon, CNRS, CPT}
\city{Marseille}
\country{France}}
\affiliation{
\institution{ISI Foundation}
\city{Turin}
\country{Italy}}
\email{alain.barrat@cpt.univ-mrs.fr}

\author{Francesco Bonchi}
\affiliation{
\institution{ISI Foundation}
\city{Turin}
\country{Italy}}
\affiliation{
\institution{Eurecat}
\city{Barcelona}
\country{Spain}}
\email{francesco.bonchi@isi.it}

\author{Ciro Cattuto}
\affiliation{
\institution{University of Turin and ISI Foundation}
\city{Turin}
\country{Italy}}
\email{ciro.cattuto@unito.it}

\author{Francesco Gullo}
\affiliation{
\institution{UniCredit, R\&D Department}
\city{Rome}
\country{Italy}}
\email{gullof@acm.org}

\begin{abstract}
When analyzing temporal networks, a fundamental task is the identification of dense structures (i.e., groups of vertices that exhibit a
large number of links), together with their temporal span (i.e., the period of time for which the high density holds).
In this paper we tackle this task by introducing a notion of temporal core decomposition where each core is associated with
two quantities, its coreness, which quantifies how densely it is connected, and its span, which is a temporal interval:
we call such cores \emph{\spancores}.

For a temporal network defined on a discrete temporal domain $T$, the
total number of time intervals included in $T$ is quadratic in $|T|$, so that
the total number of span-cores is potentially quadratic in $|T|$ as well.
Our first main contribution is an algorithm that, by exploiting containment properties among span-cores, computes all the span-cores efficiently.
Then, we focus on the problem of finding only the \emph{maximal span-cores}, i.e., span-cores that are not dominated by any other span-core
by both their coreness property and their span.
We devise a very efficient algorithm that exploits theoretical findings on the maximality condition to directly extract the maximal ones
without computing all span-cores.

Finally, as a third contribution, we introduce the problem of \emph{temporal community search}, where a set of query vertices is given as input, and the goal is to find a set of densely-connected subgraphs containing the query vertices and covering the whole underlying temporal domain $T$.
We derive a connection between this problem and the problem of finding (maximal) \spancores.
Based on this connection, we show how temporal community search can be solved in polynomial-time via dynamic programming, and how the maximal \spancores can be profitably exploited to significantly speed-up the basic algorithm.

We provide an extensive experimentation on several real-world temporal networks of widely different origins and characteristics.
Our results confirm the efficiency and scalability of the proposed methods.
Moreover, we showcase the practical relevance of our techniques in a number of applications on temporal networks, 
describing face-to-face contacts between individuals in schools. Our experiments highlight the relevance of the 
notion of (maximal) span-core in analyzing social dynamics, detecting/correcting anomalies in the data, and graph-embedding-based 
network classification.
\end{abstract}

\ccsdesc[500]{Mathematics of computing~Graph algorithms}
\ccsdesc[300]{Mathematics of computing~Graph theory}
\ccsdesc[500]{Information systems~Temporal data}
\ccsdesc[100]{Information systems~Data mining}
\ccsdesc[300]{Theory of computation~Dynamic programming}

\keywords{Temporal networks, Core decomposition, Maximal cores, Community search, Face-to-face interaction networks}

\maketitle
\sloppy

\section{Introduction}
\label{sec:intro}

A temporal network\footnote{We use ``network'' and ``graph'' interchangeably throughout the paper.}
is a representation of entities (vertices), their relations (links), and how these relations are established/broken over time.
Notice that here we will consider discrete times, i.e., the temporal networks can be represented as a time-ordered series
of snapshots (instantaneous graphs).
Extracting dense structures (i.e., groups of vertices exhibiting a large number of links with each other), together with their temporal span (i.e., the period of time for which the high density is observed) is a key mining primitive to characterize such temporal networks and extract
relevant structures.
This type of pattern enables fine-grain analysis of the network dynamics and can be a building block towards more complex tasks and applications, such as finding temporally recurring subgraphs or anomalously dense ones.
For instance, they can help in studying contact networks among individuals to quantify the transmission opportunities of respiratory infections
in a population and uncover situations where the risk of transmission is higher, with the goal of designing mitigation strategies~\cite{Gemmetto2014}. Anomalously dense temporal patterns among entities in a co-occurrence graph (e.g., extracted from the Twitter stream) have also been used to identify events and buzzing stories in real time  \cite{Angel,BonchiBGS16,bonchi2019theimportance}.
Another example concerns scientific collaboration and citation networks, where these patterns can help understand the dynamics of collaboration
in successful professional teams, study the evolution of scientific topics, and detect emerging technologies \cite{erdi2013prediction}.

In this paper we adopt as a measure of density of a pattern the \emph{minimum degree} holding among the vertices in the subgraph during
the pattern's span.
The problem of extracting \emph{all} these patterns is tackled by introducing a notion of \emph{temporal core decomposition} in which each core is associated with its \emph{span}, i.e., an interval of \emph{contiguous timestamps}, for which the coreness property holds.
 We term such a notion of temporal core \emph{span-core}.


Moreover, in several application scenarios it is typically required to identify only those dense patterns that contain a given set of query vertices.
We therefore introduce the problem of \emph{temporal community search}, whose goal is to find a set of cohesive temporal
subgraphs containing the input query vertices and covering the whole temporal domain.

To the best of our knowledge, the problems of  (efficient) \emph{span-core} computation and temporal community search
have never been studied so far.

\subsection{Challenges and contributions}
As the number of possible time intervals is quadratic in the size of the input temporal domain $T$, the total number of span-cores is, in the worst case, quadratic in $T$ too.
The naive method to find all \emph{span-cores}, which would be to operate a core decomposition for each of these time intervals, would
therefore be very time-consuming.
This is a major challenge that we tackle by deriving  containment properties between span-cores
and by exploiting them to devise an algorithm for computing all the \emph{span-cores} that is significantly more
efficient than the na\"ive exhaustive method.

We then shift our attention to the problem of finding only the \emph{maximal span-cores}, defined as
the span-cores that are not dominated by any other span-core by both the coreness property and the span.
A straightforward way of approaching
this
problem is to filter out non-maximal \spancores during the execution of an algorithm for computing the whole \spancore decomposition.
However, as the maximal ones are usually much less numerous than the overall span-cores, it would be desirable to have a method that effectively exploits the maximality property and extracts maximal \spancores directly, without computing the complete decomposition.
The design of an algorithm of this kind is an interesting challenge, as it contrasts with the intrinsic conceptual properties of core decomposition, based on which a core of order $k$ can be efficiently computed from the core of order $k\!-\!1$, of which it is a subset. For this reason, at first glance, the computation of the core of the highest order would seem as hard as computing the overall core decomposition. Instead, in this work we derive a number of theoretical properties about the relationship among \spancores of different temporal intervals and, based on these findings, we show how such a challenging goal may be achieved.

Finally, we focus on the problem of \emph{community search in temporal networks}.
Community search has been extensively studied in static graphs.
It requires to find a subgraph containing a given set of query vertices and maximizing a certain density measure~\cite{fang2019survey, HuangLX17}.
Here, we propose a formulation of the community-search problem in temporal networks as follows: given a set $Q$ of query vertices, and a positive integer $h$, find a segmentation of the underlying  temporal domain in $h$ segments $\{\Delta_i\}_{i=1}^h$ and a subgraph $S_i$ for every identified segment $\Delta_i$ such that each $S_i$ contains the query vertices $Q$ and the total density of the subgraphs is maximized.
Following the bulk of the literature in community search on static networks, in our definition of temporal community search we adopt the minimum degree as a density measure.

We show that, with some manipulations, temporal community search can be reformulated  as an instance of the popular \emph{sequence segmentation} problem, which asks for partitioning a sequence of  numerical values into $h$ segments so as to minimize the sum of the penalties (according to some penalty function) on the identified segments~\cite{bellman61approximation}.
Therefore, the classical dynamic-programming algorithm for sequence segmentation by Bellman~\cite{bellman61approximation}  can be easily adapted to solve temporal community search in polynomial time.

A criticality of this approach is that a na\"ive adaptation of the Bellman's algorithm takes quadratic time in the size of the input temporal domain $T$.
As a major contribution in this regard, we prove that the set of maximal span-cores provide a sound and complete basis to still have an optimal solution to temporal community search, while at the same time leading to a significant speed-up with respect to the na\"ive method.
In fact, let $T^* \subseteq T$ be the subset of timestamps that are covered by the span of at least one maximal span-core, together with the timestamps that immediately precede or succeed any of such spans.
We show that considering $T^*$ (instead of $T$)  in the (adaptation of the) Bellman's algorithm is sufficient to optimally solve the underlying temporal-community-search problem instance.
As, typically, $|T^*| \ll |T|$, this finding guarantees a considerable improvement in efficiency (as confirmed by our experiments).

A further challenge in our temporal-community-search problem is a typical one in community-search formulations based on minimum degree,
namely, that the output subgraphs are typically large in size.
We tackle this challenge by devising a method to reduce the size of the output subgraphs without affecting optimality.
The proposed method is inspired by the one devised by Barbieri~\emph{et~al.}~\cite{BarbieriBGG15} for the problem of minimum
community search (in static graphs).

\pagebreak 

To summarize, the main contributions of this paper are as follows:
\begin{itemize}

\item
We introduce the notion of span-core decomposition and maximal span-core in temporal networks, characterizing structure and size of the search space and providing important containment properties  (Section~\ref{sec:problem}).

\item
We devise an algorithm for computing all span-cores that exploits the aforementioned containment properties and is orders of magnitude faster than a na\"{\i}ve method based on traditional core decomposition (Section~\ref{sec:spancores}).

\item
We study the problem of finding only the maximal span-cores.
We derive several theoretical findings about the relationship between maximal span-cores and exploit these findings to devise an algorithm that is more efficient than computing all span-cores and discarding the non-maximal ones (Section~\ref{sec:maximal_spancores}).

\item
We introduce the problem of temporal community search and show how it  can be solved in polynomial time via dynamic programming.
We prove an important connection between temporal community search and maximal span-cores, which allows us to devise an algorithm that is considerably more efficient than the na\"ive dynamic-programming one.
We also propose a method to achieve the critical challenge of having too large communities as output (Section~\ref{sec:community_search}).

\item
We provide a comprehensive experimentation on several real-world temporal networks, with millions of vertices, tens of millions of edges, and hundreds of  timestamps, which attests efficiency and scalability of our methods (Section~\ref{sec:experiments}).

\item
We present applications on face-to-face interaction networks that illustrate the relevance of the notions of (maximal) span-core and temporal community search in real-life analyses and applications (Section~\ref{sec:applications}).

\end{itemize}

The next section provides an overview of the related literature, while Section~\ref{sec:conclusions} discusses future work and concludes the paper.

An abridged version of this work, covering Sections~\ref{sec:spancores}~and~\ref{sec:maximal_spancores}, together with the corresponding experiments (i.e., parts of Sections~\ref{sec:experiments}~and~\ref{sec:applications}), was presented in~\cite{galimberti2018mining}.

\spara{Reproducibility.}
For the sake of reproducibility, all our code and some of the datasets used in this paper are available at \href{https://github.com/egalimberti/span_cores}{github.com/egalimberti/span\_cores}.

\section{Background and related work}
\label{sec:related}

\subsection{Core decomposition}

Given a simple (static) graph $G=(V,E)$, let $d(S,u)$ denote the degree of vertex $u \in V$ in the subgraph induced by vertex set $S \subseteq V$, i.e., $d(S,u) = |\{v \in S \mid (u,v) \in E \}|$.
The notions of $k$\emph{-core} and \emph{core decomposition} are defined as follows:

\begin{mydefinition}[$k$-core and core decomposition~\cite{MatulaB83}]\label{def:kcores}
 The $k$\emph{-core} (or core of order $k$) of $G$ is a
 \emph{maximal} set of vertices $C_k \subseteq V$ such that $\forall u \in C_k: d(C_k,u) \geq k$.
 The set of all $k$-cores $V = C_0 \supseteq C_1 \supseteq \cdots \supseteq C_{k^*}$ ($k^* = \arg\max_{k} C_k \neq \emptyset$) is the
 \emph{core decomposition} of $G$.
\end{mydefinition}

Core decomposition can be computed in linear time by iteratively removing the smallest-degree vertex and setting its core number as equal to its degree at the time of removal~\cite{batagelj2011fast}.
Among the many definitions of dense structures, core decomposition is particularly appealing as, among others, it is fast to compute,
and can speed-up/approximate dense-subgraph extraction according  to various other definitions.
%
For instance, core decomposition allows for finding cliques more efficiently~\cite{EppsteinLS10}, as a $k$-clique is contained into a $(k\!-\!1)$-core, which can be significantly smaller than the original graph.
Moreover, core decomposition is at the basis of  approximation algorithms for the densest-(at-least-$k$-)subgraph problem~\cite{KortsarzP94,AndersenC09}, and  betweenness centrality~\cite{HealyJMA06}.
Core decomposition has also been recognized as an important tool to analyze and visualize complex networks~\cite{DBLP:conf/gd/BatageljMZ99,Alvarez-HamelinDBV05}
in several domains, e.g., bioinformatics~\cite{DBLP:journals/bmcbi/BaderH03,citeulike:298147},
software engineering~\cite{DBLP:journals/tjs/ZhangZCLZ10},
and social networks~\cite{Kitsak2010,GArcia2013}.
It has been studied under various settings, such as distributed~\cite{DistributedCores1}, streaming/maintenance~\cite{StreamingCores,li2014efficient}, and disk-based~\cite{DiskCores}, and generalized to various types of static graphs,
such as uncertain~\cite{bonchi14cores}, directed~\cite{DirectedCores}, weighted~\cite{WeigthedCores,s-core}, bipartite graphs~\cite{liu2019efficient}, or including attributes on the nodes~\cite{zhang2017engagement}.
For a comprehensive survey about theory, algorithms, and applications of core decomposition we refer to~\cite{bonchi2018core,malliaros2019core}.


Two types of extension of core decomposition bear some relation to our work.
First, core decomposition in \emph{multilayer networks} -- i.e., networks that are composed of a superposition of networks -- has been studied in~\cite{GalimbertiBG17,GalimbertiBGL20}.
In the multilayer setting a core is allowed to extend on any subset of layers, thus implying that the total number of multilayer cores is exponential in the number of layers.
Although temporal networks can be viewed as a special case of multilayer networks (where each timestamp is interpreted as a layer),
there is a fundamental difference: in a temporal network the ``layers" are ordered, and
the sequentiality of timestamps represents an important structural constraint.
In other words, in the temporal setting we are interested in cores that span a temporal interval, and not simply any subset of (potentially non-contiguous) timestamps.
This aspect has two critical consequences.
First, the search space and the number of temporal cores are no longer exponential, unlike the multilayer case.
Second,
to guarantee an effective fulfilment of the constraint on temporal sequentiality, the requirements for the edges that contribute to the formation of a temporal core are stricter than the ones at the basis of the multilayer-core definition.
A more detailed technical discussion of the relationship between multilayer core decomposition and the proposed temporal core decomposition is reported in Section~\ref{subsec:spancoresVSmlcores}.

The second extension of core decomposition that shares some relation to the one proposed in this work is due to Wu~{\em et~al.}, who have proposed in~\cite{wu2015core} an alternative definition of temporal core decomposition.
A major difference between Wu~{\em et~al.}'s definition and ours is that the former \emph{does not take any kind of temporal constraint into account}.
Indeed, Wu~{\em et~al.} define the $(k,h)$-core as the largest subgraph in which every vertex has at least $k$ neighbors and there are at least $h$ temporal edges between the vertex and its neighbors,
without any restriction on when these $h$ edges occur: the sequentiality of connections is not taken into account and non-contiguous timestamps can support the same core.
In fact, the $(k,h)$-core decomposition
can be seen as a kind of weighted static core decomposition on the weighted static network resulting from the aggregation of the
temporal network.
In contrast, our temporal cores have each a clear temporal collocation and continuous spans, so that our definition
includes temporality in an explicit way and cannot be reduced to Wu~{\em et~al.}'s one.
As we will see in Section~\ref{sec:applications}, associating a temporal collocation to each core is important in applications.

\subsection{Patterns in temporal networks}
A number of works on extracting dense patterns from a temporal network focus on the well-established notion of \emph{densest subgraph}, i.e., a subgraph maximizing the average-degree density.
Jethava~and~Beerenwinkel~\cite{jethava2015finding} consider as input
a set of graphs sharing the same vertex set, which can thus also be interpreted as a temporal network.
On such an input they study the \emph{densest common subgraph} problem, i.e., the problem of finding a subgraph maximizing the minimum average degree over all graphs (timestamps), and devise a linear-programming formulation and a greedy heuristic algorithm for it.
Further (mostly theoretical) advancements to the densest-common-subgraph problem have been provided by  Reinthal~\emph{et~al.}~\cite{reinthal2016finding} and Charikar~\emph{et~al.}~\cite{charikar2018finding}.
Semertzidis~\emph{et~al.}~\cite{semertzidis2016best} instead introduce two more variants of the problem, where the goal is to maximize the average average degree and the minimum minimum degree, respectively. They show that the average-average variant easily reduces to the traditional densest-subgraph problem, and that the minimum-minimum variant can be exactly solved by a simple adaptation of the classic algorithm for core decomposition.

Complementary works focus on variants of the densest-subgraph-discovery problem. Rozenshtein~\emph{et~al.} study the problem of discovering dense temporal subgraphs whose edges occur in short time intervals considering the exact timestamp of the occurrences~\cite{rozenshtein2017finding},
and the problem of partitioning the timeline of a temporal network into non-overlapping intervals, such that the intervals span subgraphs with maximum total density~\cite{rozenshtein2018finding}.
Epasto~\emph{et~al.}~\cite{epasto2015efficient} deal with the problem of maintaining the densest subgraph in a dynamic setting.

Attention in the literature has also been devoted to densities other than the average degree.
The notion of $\Delta$-clique, as a set of vertices in which each pair is in contact at least every $\Delta$ timestamps, has been proposed in~\cite{viard2016computing, himmel2016enumerating}.
Bentert~\emph{et~al.}~\cite{bentert2018listing} introduce the $\Delta$-$k$-plex, a relaxation of $\Delta$-clique in which each vertex has an edge to all but at most $k-1$ vertices at least once every $\Delta$ consecutive timestamps.
Li~\emph{et~al.}~\cite{li2018persistent}
study the problem of finding the maximum
($\theta$,$\Delta$)-persistent $k$-core in a temporal network, i.e., the largest subgraph that is a connected $k$-core in all the subintervals of duration $\theta$ of a given temporal interval $\Delta$.


A different, but still slightly related body of literature focuses on other definitions of temporal patterns, such as frequent evolution patterns in temporal attributed graphs~\cite{BerlingerioBBG09,inokuchi2010mining, desmier2012cohesive},
link-formation rules in temporal networks~\cite{BringmannBBG10,leung2010mining}, frequency-estimation algorithms for counting temporal motifs~\cite{Kovanen:2011,liu2019sampling}, finding a small vertex set whose removal eliminates all temporal paths connecting two designated terminal vertices~\cite{zschoche2018complexity}, finding a subgraph that maximizes the sum of edge weights in a network whose topology remains fixed but edge weights evolve over time~\cite{bogdanov2011mining,ma2017fast}, and
the discovery of dynamic relationships and events~\cite{das2011dynamic}, or of correlated activity patterns~\cite{Gauvin2014}.

This work  differs from all the above ones as our notions of span-core and temporal core decomposition do not correspond (or are straightforwardly reducible) to any of those temporal patterns.

\subsection{Community search}
Given a static graph and a set of query vertices, the \emph{community search} problem aims at finding a cohesive subgraph \emph{containing the query vertices}.
Community search has attracted a great deal of attention in the last years~\cite{fang2019survey,HuangLX17}.
Sozio~and~Gionis~\cite{Sozio} are the first to introduce this problem by employing the \emph{minimum degree} as a cohesiveness measure.
Their formulation can be solved by a simple (linear-time) greedy algorithm, which resembles the traditional $2$-approximation algorithm for densest subgraph proposed in~\cite{Char00}.
More recently, Cui~\emph{et~al.}~\cite{SozioLocalSIGMOD14} devise a local-search approach to improve the efficiency of the method defined in~\cite{Sozio}, but only for the special case of a single query vertex.
The case of multiple query vertices has instead been addressed by Barbieri~\emph{et~al.}~\cite{BarbieriBGG15}, who exploit core decomposition as a preprocessing step to improve efficiency.
They also tackle the problem of \emph{minimum community search}, i.e., a variant of community search where the size of the output subgraph has to be minimized.

Community search has also been studied under different names and/or settings.
Huang~\emph{et~al.}~\cite{KtrussSIGMOD14} introduce a community-search model based on the $k$-truss notion.
Andersen~and~Lang~\cite{Andersen1} and Kloumann~and~Kleinberg~\cite{Kloumann} study \emph{seed set expansion} in social graphs, in order to find communities with small conductance or that are well-resemblant of the characteristics of the query vertices, respectively.
Other works define connectivity subgraphs based on electricity analogues~\cite{connect}, random walks~\cite{CenterpieceKDD06}, the minimum-description-length principle~\cite{akoglu2013mining}, the Wiener index~\cite{ruchansky2015minimum}
and  network efficiency~\cite{RuchanskyBGGK17}.
Recent approaches also introduce the flexibility of having query vertices belonging to different communities~\cite{bian2018multi,yan2019constrained}.
Finally, community search has been formalized for attributed graphs~\cite{huang2017attribute,fang2017attributed} and spatial graphs~\cite{fang2017spatial} as well.

In this work we study for the first time community search in temporal graphs.
Specifically, we provide a novel definition of the problem by asking for a set of subgraphs containing the given query vertices, along with their corresponding temporal intervals, such that the total minimum-degree density of the identified subgraphs is maximized and the
union of the temporal intervals spanned by those subgraphs covers the whole underlying temporal domain.
{
None of the above works deal with such a definition of temporal community search, not even the works by Rozenshtein~\emph{et~al.}~\cite{rozenshtein2018finding} and Li~\emph{et~al.}~\cite{li2018persistent} discussed in the previous subsection. 
In fact, although Rozenshtein~\emph{et~al.}~\cite{rozenshtein2018finding} and Li~\emph{et~al.}~\cite{li2018persistent} search for cohesive temporal subgraphs,
they do not accept any query vertices in input.
This is a fundamental feature, which makes those works actually solve a problem other than community search. 
Another key difference is that they focus on different notions of density.
}



\section{Temporal core decomposition: problem statement}
\label{sec:problem}

In this section we provide preliminary definitions and the needed notations,
introduce the problem of finding all \spancores and only the maximal ones, and
prove containment properties among \spancores that are at the basis of our efficient algorithms.

\subsection{Span-cores}
We are given a \emph{temporal graph} $G = (V,T,\tau)$, where $V$ is a set of vertices,  $T = [0, 1, \ldots, t_{max}] \subseteq \mathbb{N}$ is a discrete time domain, and $\tau: V  \times V \times  T\rightarrow \{0,1\}$ is a function defining for each pair of vertices $u,v \in V$ and each timestamp $t \in T$ whether edge $(u,v)$ exists in $t$. We denote $E = \{(u,v,t) \mid \tau(u,v,t) = 1 \}$ the set of all temporal edges. Given a timestamp $t \in T$, $E_t = \{(u,v) \mid \tau(u,v,t) = 1 \}$ is the set of edges existing at time $t$.
A temporal interval $\Delta = [t_s, t_e]$ is contained into another temporal interval $\Delta' = [t'_s, t'_e]$, denoted $\Delta \sqsubseteq \Delta'$, if $t'_s \leq t_s$ and $t'_e \geq t_e$.
Given an interval $\Delta \sqsubseteq T$, we denote $E_\Delta = \bigcap_{t \in \Delta} E_t$ the edges existing in 
\emph{all timestamps}\footnote{We remark that this is just one of the possible ways of defining the existence of an edge 
in a temporal domain. There are two basic semantics used in the literature: the ``AND'' semantics we employ here, where an edge 
is required to exist in all the timestamps of an interval, and an ``OR'' semantics, requiring that an edge appears in at least 
 one of the timestamps. Although both semantics can be meaningful and there is no strong a-priori argument to prefer one over the other, 
the types of application and the desired semantics of the data analysis can dictate the choice. 
In this work we are particularly interested in networks of social interactions (contacts, communications, etc.), and in 
exposing structures that are cohesive and stable, together with their duration. It seems 
then natural to consider the AND semantics, as an OR semantics would correspond to an aggregation on a temporal interval
and would not constrain the simultaneity of interactions to define a structure.
This simultaneity is crucial in applications such as the ones in Section~\ref{sec:applications}, in which we show the relevance of our work in the analysis of contact networks among individuals recorded by an RFID-based proximity-sensing infrastructure.} of $\Delta$. Given a subset $S \subseteq V$ of vertices, let $E_{\Delta}[S] = \{(u,v) \in E_{\Delta} \mid u \in S, v \in S\}$ and $G_{\Delta}[S] = (S, E_{\Delta}[S])$.
Finally, the  \emph{temporal degree} of a vertex $u$ within $G_{\Delta}[S]$ is denoted $\tdeg_\Delta(S,u) = |\{v \in S \mid (u,v) \in E_\Delta[S] \}|$.

\begin{mydefinition}[$(k,\Delta)$-core] \label{def:core}
The  $(k,\Delta)$\emph{-core} of a temporal graph $G = (V,T,\tau)$ is (when it exists) a maximal and non-empty set of vertices $ \emptyset \neq C_{k,\Delta} \subseteq V$, such that $\forall u \in C_{k,\Delta} : \tdeg_\Delta(C_{k,\Delta},u) \geq k$, where
 $\Delta \sqsubseteq T$ is a temporal interval and $k \in \mathbb{N}^+$.
\end{mydefinition}

A $(k,\Delta)$-core is thus a set of vertices implicitly defining a cohesive subgraph (where $k$ represents the cohesiveness constraint), together with its \emph{temporal span}, i.e., the interval $\Delta$ for which the subgraph satisfies the cohesiveness constraint.
In the remainder of the paper we refer to this type of temporal pattern as \emph{\spancore}.

The first problem we tackle in this work is to compute the \emph{span-core decomposition} of a temporal graph $G$, i.e., all span-cores of~$G$.

\begin{problem}[Span-core decomposition] \label{pbl:dececomposition}
Given a temporal graph $G$, find the set of all $(k,\Delta)$-cores of $G$.
\end{problem}

Unlike standard cores of simple graphs, \spancores are not all nested into each other, due to their spans.
However, they still exhibit containment properties.
Indeed, it can be observed that a $(k,\Delta)$-core is contained into any other $(k',\Delta')$-core with less restrictive degree  and span conditions, i.e., $k' \leq k$, and $\Delta' \sqsubseteq \Delta$.
This property is depicted in Figure~\ref{fig:searchspace}, and formally stated in the next proposition.

\begin{figure}[t!]
\centering
\includegraphics[width=0.7\columnwidth]{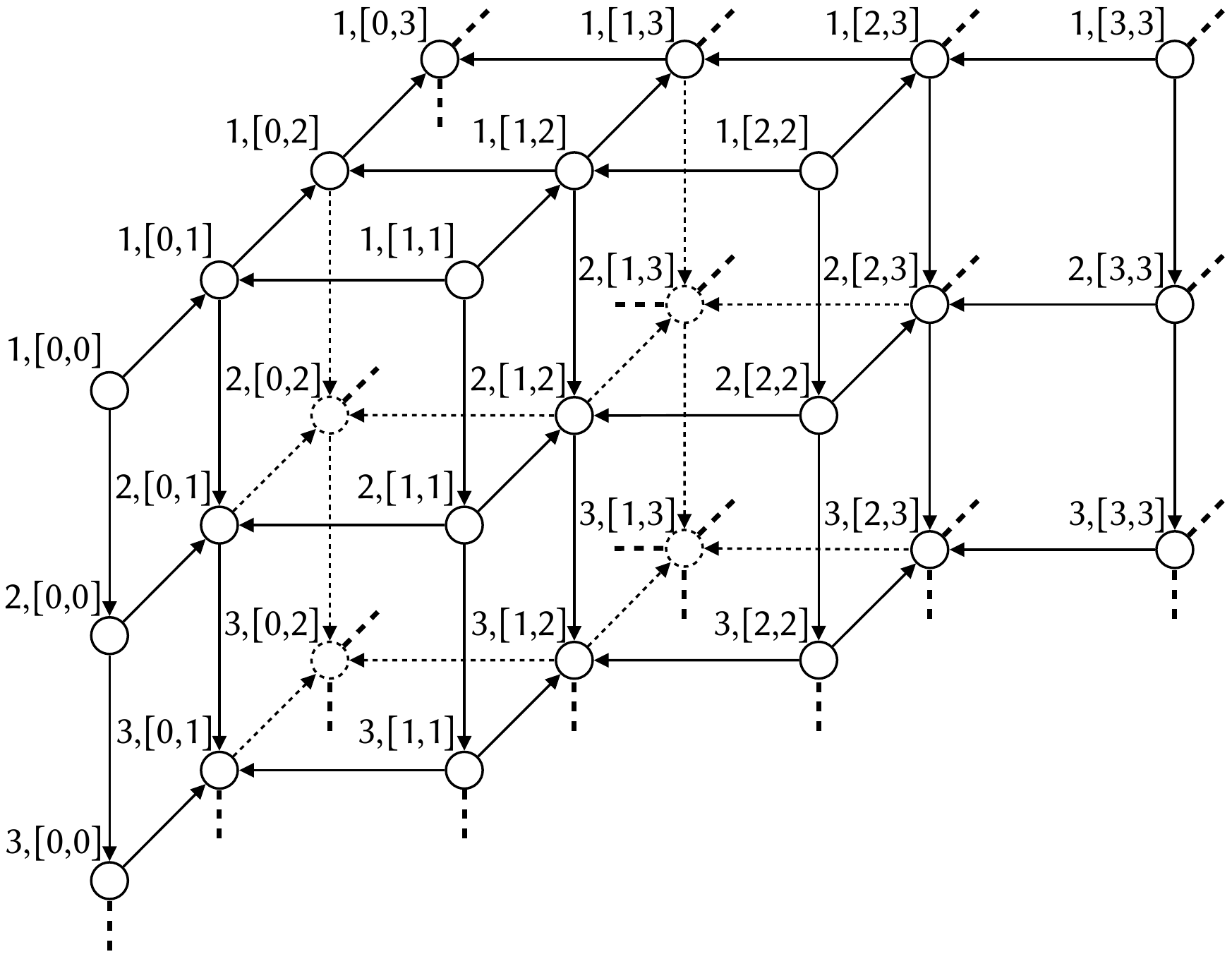}
\caption{\label{fig:searchspace} Search space: for a temporal span $\Delta = [t_s,t_e]$, the $(k,\Delta)$-core is depicted as a node labeled ``$k, [t_s, t_e]$''. An arrow  $C_1 \rightarrow C_2$ denotes $C_1 \supseteq C_2$ (the
distinction between solid and dotted arrows is for visualization sake only).}
\end{figure}

\begin{myproposition}[\Spancore containment]\label{prp:conteinment}
For any two \spancores $C_{k, \Delta}$, $C_{k',\Delta'}$ of a temporal graph $G$ it holds that
$$
k' \leq k \wedge \Delta' \sqsubseteq \Delta \ \Rightarrow \  \core \subseteq C_{k',\Delta'}.
$$
\end{myproposition}
\begin{proof}
The result can be proved by separating the two conditions in the hypothesis, i.e., by separately showing that ($i$) $k' \leq k \Rightarrow C_{k, \Delta} \subseteq C_{k', \Delta}$, and ($ii$) $\Delta' \sqsubseteq \Delta \Rightarrow C_{k,\Delta} \subseteq C_{k,\Delta'}$.
The first point holds as, keeping the span $\Delta$ fixed, the maximal set of vertices $C$ for which $\tdeg_\Delta(C,u) \geq k$ is clearly contained in the maximal  set of vertices $C'$ for which $\tdeg_\Delta(C',u) \geq k'$, if $k' \leq k$.
To prove (ii), it can be noted that $\Delta' \sqsubseteq \Delta \Rightarrow E_\Delta \subseteq E_{\Delta'}$, which implies that
$\forall u \in C_{k,\Delta} : \tdeg_\Delta(C_{k,\Delta}, u) \leq \tdeg_{\Delta'}(C_{k,\Delta}, u)$.
Therefore, all vertices within $C_{k,\Delta}$ satisfy the condition to be part of $C_{k,\Delta'}$ too.
\end{proof}

The following observation directly derives from Proposition~\ref{prp:conteinment} and states that finding all the \spancores having a fixed span $\Delta$ corresponds to computing the core decomposition of a simple graph.

\begin{myobservation}\label{observation1}
For a fixed temporal interval $\Delta \sqsubseteq T$, finding all span-cores that have $\Delta$ as their span is equivalent to computing the classic core decomposition~\cite{batagelj2011fast} of the simple graph $G_\Delta = (V, E_\Delta)$.
\end{myobservation}

\subsection{Maximal \spancores}
As the total number of temporal intervals that are contained into the whole time domain $T$ is $|T|(|T|\!+\!1)/2$, the total number of \spancores  is
potentially $\mathcal{O}(|T|^2 \times k_{max})$, where $k_{max}$ is the largest value of $k$ for which a  $(k,\Delta)$-core exists.
It is thus quadratic in $|T|$, which may be too large an output for human direct inspection.
In this regard, it may be useful to focus only on the most relevant cores, i.e., the \emph{maximal} ones, as defined next.

\begin{mydefinition}[Maximal \spancore] \label{def:maximal}
A \spancore $\core$ of a temporal graph $G$ is said \emph{maximal} if there does not exist any other \spancore $C_{k',\Delta'}$ of $G$ such that $k \leq k'$ and $\Delta \sqsubseteq \Delta'$.
\end{mydefinition}

Hence, a \spancore is recognized as maximal if it is not dominated by another \spancore both on the order $k$ and the span $\Delta$.
Differently from the \emph{innermost core} (i.e., the core of the highest order) in the classic core decomposition, which is unique,
in our temporal setting the number of maximal \spancores is $\mathcal{O}(|T|^2)$, as, in the worst case,
there may be one maximal \spancore for every temporal interval.
However, as observed in empirical temporal-network data, maximal \spancores are always much less than the overall \spancores: the difference is usually one order of magnitude or more.
The second problem we tackle in this work is to compute the maximal \spancores of a temporal graph.

\begin{problem}[\maxspancoreprob] \label{pbl:maximal}
Given a temporal graph $G$, find the set of all maximal $(k,\Delta)$-cores of $G$.
\end{problem}

Clearly, one could solve Problem~\ref{pbl:maximal} by  solving Problem~\ref{pbl:dececomposition} and  filtering out all the non-maximal \spancores.
However, an interesting yet challenging question is whether one can exploit the maximality condition to develop faster algorithms that can directly extract the maximal ones, without computing all the span-cores.
We provide a positive answer to this question in Section~\ref{sec:maximal_spancores}.

{

\subsection{Relation to multilayer core decomposition~\cite{GalimbertiBG17,GalimbertiBGL20}}\label{subsec:spancoresVSmlcores}


\emph{Multilayer graphs} are a representation paradigm of  complex systems, where multiple relations of different types occur between the same pair of entities~\cite{bonchi2015chromatic,DickisonMagnaniRossi2016,tagarelli2017ensemble}.
A multilayer graph is formally defined as a triple $G = (V, E, L)$, where $V$ is a set of vertices, $L$ is a set of layers, and $E \subseteq V \times V \times L$ is a set of edges.
Given a multilayer graph $G = (V, E, L)$ and an $|L|$-dimensional integer vector $\vec{k} = [k_{\ell}]_{\ell \in L}$,
Galimberti~\emph{et~al.}~\cite{GalimbertiBG17,GalimbertiBGL20} define the notion of \emph{multilayer} $\vec{k}$\emph{-core} of $G$ as a \emph{maximal} set $C \subseteq V$ of vertices such that, for all $\ell \in L$, the minimum degree
of a vertex in $C$ in layer $\ell$ is larger than or equal to $k_{\ell}$.
In other words, a $\vec{k}$-multilayer-core corresponds to a subgraph that satisfies  the $k_\ell$-core definition in layer $\ell$, for all $\ell \in L$.
For instance, for $|L| = 2$, a multilayer $(k_1,k_2)$-core is a subgraph that is simultaneously a $k_1$-core in the first layer and a $k_2$-core in the second layer.

Temporal graphs can be viewed as a special case of multilayer graphs where timestamps correspond to layers.
Therefore, a natural question while introducing a notion of temporal core is how it relates to the definition of a core in the multilayer setting.
A fundamental difference is that, unlike the multilayer context, in a temporal graph the ``layers" are ordered, and the consecutio of timestamps should be taken into account.
As a result, the two definitions are not comparable and have different conceptual and computational properties in the general case.
A major remark in this regard is that the multilayer cores, as defined in \cite{GalimbertiBG17,GalimbertiBGL20}, are exponential in the number $|L|$ of layers $|L|$, \emph{thus the multilayer core decomposition takes (worst-case) exponential time, while temporal core decomposition is computable in polynomial time}.

Once having ascertained such a key difference,
another meaningful investigation would be understanding whether the notion of multilayer core may still be exploited to define/compute  \spancores, even if only in limited circumstances.
In this regard, as formally shown in the next proposition and  illustrated in the example in Figure~\ref{fig:spancoresVSmlcores}, there exists a containment relationship between \spancores and multilayer cores.

\begin{myproposition}\label{prop:spancoresVSmlcores}
Let $C$ be the $(k,\Delta)$-\spancore of a temporal graph $G = (V,T,\tau)$.
Let also $\vec{k} = [k_t]_{t \in T}$ be a $|T|$-dimensional integer vector such that $\forall t \in \Delta : k_t = k$, $\forall t \notin \Delta : k_t = 0$, and let $C'$ be the $\vec{k}$-multilayer-core extracted from $G$ by interpreting it as a multilayer graph where layers correspond to timestamps.
It holds that $C \subseteq C'$.
\end{myproposition}
\begin{proof}
According to Definition~\ref{def:core}, every vertex $v \in C$, has a at least $k$ neighbors within $C \setminus \{v\}$, for every timestamp $t \in \Delta$.
This complies with the definition of $\vec{k}$-multilayer-core, $\vec{k} = [k_t]_{t \in T}$, $\forall t \in \Delta : k_t = k$, $\forall t \notin \Delta : k_t = 0$, meaning that all vertices in $C$ are necessarily part of the $\vec{k}$-multilayer-core as well.
\end{proof}

\begin{figure}[t]
	\centering
	\begin{tikzpicture}[scale=0.20]
	\node[circle, draw, thick, scale=0.7] (A) at (0,12) {\textsf{A}};
	\node[circle, draw, thick, scale=0.7] (B) at (12,12) {B};	
	\node[circle, draw, thick, scale=0.7] (C) at (2,6) {C};
	\node[circle, draw, thick, scale=0.7] (D) at (12,0) {D};
	\node[circle, draw, thick, scale=0.7] (E) at (0,0) {E};
	
	\node[scale=0.7,rotate=315] (DA) at (7,4) {2600};
	\node[scale=0.7] (DE) at (6,-1) {1600};
	
	\draw[->, thick, >=stealth, bend left=18] (A) edge (B);
	\draw[->, thick, >=stealth, bend right=18] (A) edge (B);
	\draw[->, thick, >=stealth, bend left=18] (B) edge (D);
	\draw[->, thick, >=stealth, bend right=18] (B) edge (D);
	\draw[->, thick, >=stealth] (D) edge (E);
	\draw[->, thick, >=stealth] (E) edge (A);
	\draw[->, thick, >=stealth, dashed] (E) edge (C);
	\draw[->, thick, >=stealth, dashed, bend right=10] (C) edge (B);
	\draw[->, thick, >=stealth, dashed] (D) edge (A);
\end{tikzpicture}
	\caption{\label{fig:spancoresVSmlcores}{
	Relationship between multilayer cores defined in~\cite{GalimbertiBG17,GalimbertiBGL20} and  \spancores introduced in this work (Proposition~\ref{prop:spancoresVSmlcores}).
	The figure depicts a toy temporal graph $G$, with time domain $T = \{0,1,2\}$. Solid, dashed, and dotted edges refer to timestamp $0$, $1$, and $2$, respectively.
	The $(2,[0,1])$-\spancore of $G$ corresponds to $C = \{\mbox{\textsf{A,D,E}}\}$.
	At the same time, $C' = \{\mbox{\textsf{A,B,D,E}}\}$ corresponds to the $(2,2,0)$-multilayer-core of $G$, when $G$ is interpreted as a (3-layer) multilayer graph with the first, second, and third layer corresponding to timestamps $0$, $1$, and $2$, respectively. While there is no exact correspondence, it can be observed that
\spancore $C$ is contained into multilayer core $C'$.}}
\end{figure}
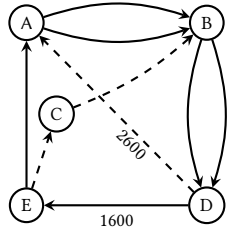

Proposition~\ref{prop:spancoresVSmlcores} suggests that, in principle, to compute \spancores, one may: ($i$) compute all  multilayer cores, ($ii$) among all multilayer cores, retrieve the ones complying with Proposition~\ref{prop:spancoresVSmlcores}, and ($iii$) post-process those multilayer cores in order to extract the actual \spancores.
However, this strategy is not feasible, as, due to the aforementioned exponential-time computation, extracting multilayer cores from a temporal graph would be affordable only for very small values of $|T|$.
In fact, Galimberti~\emph{et~al.}~\cite{GalimbertiBG17,GalimbertiBGL20} show experiments on graphs with at most 10 layers, and in the 10-layer graphs computing the multilayer core decomposition takes more than 28 hours. In the temporal setting we are interested in analyzing long-term, sometimes high-frequency, interactions, thus temporal graphs have typically much more than 10 timestamps. Indeed, all the datasets used in this paper have many more timestamps, and the algorithms from~\cite{GalimbertiBG17,GalimbertiBGL20} cannot run in a reasonable amount of time.
Moreover, even assuming to be able to compute multilayer cores on temporal networks, those cores have still to be filtered and post-processed, which makes this strategy meaningless with respect to  methods that compute \spancores directly, as the ones introduced in the next section.

}

\section{Algorithms: computing all \spancores}
\label{sec:spancores}

In this section we devise algorithms for computing a complete \spancore decomposition of a temporal graph (Problem~\ref{pbl:dececomposition}).

\spara{A na\"ive approach.} As stated in Observation 1, for a fixed temporal interval $\Delta \sqsubseteq T$, mining all span-cores $\core$ is equivalent to computing the classic core decomposition of the graph $G_\Delta = (V, E_\Delta)$.
A na\"{\i}ve strategy is thus to run a core-decomposition subroutine~\cite{batagelj2011fast}
on graph $G_\Delta$ for each temporal interval $\Delta \sqsubseteq T$. Such a method has time complexity $\bigO(\sum_{\Delta \sqsubseteq T} (|\Delta| \times |E|))$, i.e., $\bigO(|T|^2 \times |E|)$.

\spara{A more efficient algorithm.} Looking at Figure \ref{fig:searchspace} one can observe that the  na\"ive algorithm only exploits one dimension of the containment property: it starts from each point on the top level, i.e., from cores of order $1$, and goes down vertically with the classic core decomposition. Based on Proposition~\ref{prp:conteinment}, it is possible to design a more efficient algorithm that
exploits also the ``horizontal containment'' relationships.

\begin{myexample}
Consider core  $C_{1,[0,2]}$  in Figure \ref{fig:searchspace}: by Proposition~\ref{prp:conteinment} it holds that it is a subset of both $C_{1,[0,1]}$ and $C_{1,[1,2]}$. Therefore, to compute $C_{1,[0,2]}$, instead of starting from the whole $V$, one can start from $C_{1,[0,1]} \cap C_{1,[1,2]}$. Starting from a much smaller set of vertices can provide a substantial speed-up to the whole computation.
\end{myexample}

This observation, although simple,  produces a speed-up of orders of magnitude as we will empirically show in Section~\ref{sec:experiments}.
The next straightforward corollary of Proposition~\ref{prp:conteinment} states that, not only $C_{1,[0,2]} \subseteq C_{1,[0,1]} \cap C_{1,[1,2]}$, but this is the best one can get, meaning that
intersecting these two span-cores is equivalent to intersecting all span-cores structurally containing $C_{1,[0,2]}$.

\begin{mycorollary}\label{cor:cor1ofProp1}
Given a temporal graph $G = (V, T, \tau)$, and a temporal interval $\Delta = [t_s,t_e] \sqsubseteq T$, let $\Delta_+ = [\min\{t_s+1,t_e\}, t_e]$ and $\Delta_- = [t_s, \max\{t_e-1,t_s\}]$.
It holds that $$C_{1,\Delta} \ \subseteq \ (C_{1,\Delta_+} \cap C_{1,\Delta_-}) \ = \ \bigcap_{\Delta' \sqsubseteq \Delta} C_{1, \Delta'} .$$
\end{mycorollary}


\begin{myexample}
Consider again $C_{1,[0,2]}$  in Figure \ref{fig:searchspace}:  Proposition~\ref{prp:conteinment} states that it is a subset of $C_{1,[0,0]}, C_{1,[0,1]},C_{1,[1,1]},C_{1,[1,2]},C_{1,[2,2]}$.
Corollary~\ref{cor:cor1ofProp1} suggests that there is no need to intersect them all, but only $C_{1,[0,1]}$ and $C_{1,[1,2]}$: in fact, $C_{1,[0,1]} \subseteq C_{1,[0,0]} \cap C_{1,[1,1]}$ and
$C_{1,[1,2]} \subseteq C_{1,[1,1]} \cap C_{1,[2,2]}$.
\end{myexample}

\begin{algorithm}[t]
\DontPrintSemicolon
\KwIn{A temporal graph $G=(V,T,\tau)$.}
\KwOut{The set $\coresset$ of all \spancores of $G$.}
$\coresset \leftarrow \emptyset$; \ \ $Q \leftarrow \emptyset$; \ \ $\mathcal{A} \leftarrow \emptyset$\;
\ForAll{$t \in T$}
{\label{line:decomposition:init:start}
	enqueue $[t,t]$ to $Q$; \ \ $\mathcal{A}[t,t] \leftarrow V$\;\label{line:decomposition:init:end}
}
\While{$Q \neq \emptyset$}
{\label{line:decomposition:whilestart}
	dequeue $\Delta = [t_s, t_e]$ from $Q$\;
	$E_{\Delta}[\mathcal{A}[\Delta]] \leftarrow  \{(u,v) \in E_{\Delta} \mid u \in \mathcal{A}[\Delta], v \in \mathcal{A}[\Delta]\}$\;\label{line:decomposition:subgraph}

	\If{$|E_{\Delta}[\mathcal{A}[\Delta]]| > 0$}
	{
		$\coressetdelta \leftarrow $ \textsf{core-decomposition}$(\mathcal{A}[\Delta],E_{\Delta}[\mathcal{A}[\Delta]])$\;\label{line:decomposition:deltacores}				
		$\coresset \leftarrow \coresset \cup \coressetdelta$\; \label{line:decomposition:solution}
		$\Delta_1 = [\max\{t_s-1,0\}, t_e]$; \ \  $\Delta_2 = [t_s, \min\{t_e+1,t_{max}\}]$\; \label{line:decomposition:fathers}	 
		\ForAll{$\Delta' \in \{\Delta_1, \Delta_2\} \mid \Delta' \neq \Delta$}
		{
			\uIf{$\mathcal{A}[\Delta'] \neq \textsc{null}$} 
			{\label{line:decomposition:alreadyinit}	
				$\mathcal{A}[\Delta'] \leftarrow \mathcal{A}[\Delta'] \cap C_{1,\Delta}$\; \label{line:decomposition:second}	
				enqueue $\Delta'$ to $Q$\; \label{line:decomposition:enqueue}
			}	
			\Else
			{
				$\mathcal{A}[\Delta'] \leftarrow C_{1,\Delta}$\;\label{line:decomposition:first}
			}
		}
	}
}
\caption{\cores}\label{alg:decomposition}
\end{algorithm}

The main idea behind our efficient \cores\ algorithm (whose pseudocode is given as Algorithm~\ref{alg:decomposition}) is to generate temporal intervals of increasing size (starting from size one) and, for each $\Delta$ of width larger than one, to initiate the core decomposition from $(C_{1,\Delta_+} \cap C_{1,\Delta_-})$, i.e., the smallest intersection of cores containing $C_{1,\Delta}$ (Corollary~\ref{cor:cor1ofProp1}).
The intervals to be processed are added to queue $Q$, which is initialized with the intervals of size one (Lines~\ref{line:decomposition:init:start}--\ref{line:decomposition:init:end}): these are the only intervals for which no other interval can be used to reduce the set of vertices from which the core decomposition is started, thus they have to be initialized with the whole vertex set $V$.
The algorithm utilizes a map $\mathcal{A}$ that, given an interval $\Delta$, returns the set of vertices to be used as a starting set of the core decomposition on $\Delta$.
The algorithm processes all intervals stored in $Q$, until $Q$ has become empty (Lines~\ref{line:decomposition:whilestart}--\ref{line:decomposition:first}).
For every temporal interval $\Delta$ extracted from $Q$, the starting set of vertices is retrieved from $\mathcal{A}[\Delta]$  and the corresponding set of edges is identified (Line~\ref{line:decomposition:subgraph}). Unless this is empty, the  classic core-decomposition algorithm~\cite{batagelj2011fast} is invoked over $(\mathcal{A}[\Delta],E_{\Delta}[\mathcal{A}[\Delta]])$ (Line~\ref{line:decomposition:deltacores}) and its output (a set of span-cores of span $\Delta$) is added to the ultimate output set $\coresset$ (Line~\ref{line:decomposition:solution}).

Afterwards, the two intervals, denoted $\Delta_1$ and $\Delta_2$, for which $C_{1,\Delta}$ can be used to obtain the smallest intersections of cores containing them (Corollary~\ref{cor:cor1ofProp1}) are computed at Line~\ref{line:decomposition:fathers}.
For $\Delta_1$ (and analogously $\Delta_2$), we check whether $\mathcal{A}[\Delta_1]$ has already been initialized (Line~\ref{line:decomposition:alreadyinit}): this would mean that previously the other ``father'' (i.e., smallest containing core) of $C_{1,\Delta_1}$ has been computed, thus we can intersect $C_{1,\Delta}$ with $\mathcal{A}[\Delta_1]$ and enqueue $\Delta_1$ to be processed (Lines~\ref{line:decomposition:second}--\ref{line:decomposition:enqueue}). Instead, if $\mathcal{A}[\Delta_1]$ was not yet initialized, we initialize it with $C_{1,\Delta}$ (Line~\ref{line:decomposition:first}): in this case $\Delta_1$ is not enqueued because it still lacks one father to be intersected before being ready for core decomposition.
This procedural update of $Q$ ensures that both fathers of every interval in $Q$ exist and have been previously computed, thus no a-posteriori verification is needed.

\begin{myexample}
Consider again the search space in Figure~\ref{fig:searchspace}.
Algorithm~\ref{alg:decomposition} first processes the intervals $[0,0],[1,1],[2,2],$ and $[3,3]$.
Then, it intersects  $C_{1,[0,0]}$ and $C_{1,[1,1]}$ to initialize $C_{1,[0,1]}$, intersects $C_{1,[1,1]}$ and $C_{1,[2,2]}$ to initialize $C_{1,[1,2]}$, and  intersects $C_{1,[2,2]}$ and $C_{1,[3,3]}$ to initialize $C_{1,[2,3]}$. Then, it continues with the intervals of size 3: it intersects $C_{1,[0,1]}$ and $C_{1,[1,2]}$ to initialize $C_{1,[0,2]}$ and so on.
\end{myexample}
The next theorem formally shows soundness and completeness of our \cores\ algorithm.
\begin{mytheorem}\label{th:correctnessAlg2}
Algorithm~\ref{alg:decomposition} is sound and complete for Problem~\ref{pbl:dececomposition}.
\end{mytheorem}
\begin{proof}
The algorithm generates and processes a subset of temporal intervals $\mathcal{X} \subseteq \{\Delta \mid \Delta \sqsubseteq T\}$.
For every interval $\Delta \subseteq \mathcal{X}$, it computes \emph{all} \spancores $\mathbf{C}_{\Delta} = \{C_{1,\Delta}, C_{2,\Delta}, \ldots, C_{k_{\Delta},\Delta}\}$ defined on $\Delta$ by means of the \textsf{core-decomposition} subroutine on the graph $(\mathcal{A}[\Delta],E_{\Delta}[\mathcal{A}[\Delta]])$.
The set of vertices $\mathcal{A}[\Delta]$ is equivalent to $(C_{1,\Delta_+} \cap C_{1,\Delta_-})$ because of Line~\ref{line:decomposition:second} (Corollary~\ref{cor:cor1ofProp1}) and the fact that $\Delta$ is enqueued (Line~\ref{line:decomposition:enqueue}) only when both fathers have been processed and the intersection done. The correctness of doing the classic core decomposition is guaranteed by Observation ~\ref{observation1}.

As for completeness, it suffices to show that the intervals $\Delta \notin \mathcal{X}$ that have not been processed by the algorithm do not yield any \spancore.
The algorithm generates all temporal intervals size by size, starting from those of size one and then going to larger sizes. This is done by maintaining the queue $Q$. As said above, an interval $\Delta$ is enqueued as soon as both   $C_{1,\Delta_+}$ and $C_{1,\Delta_-}$ have been processed. Thus, an interval $\Delta$ is not in $\mathcal{X}$ only if either $C_{1,\Delta_+}$ or $C_{1,\Delta_-}$ does not exist. In this case $C_{1,\Delta}$ and all other $C_{k,\Delta}$ do not exist as well.
\end{proof}

\noindent \spara{Discussion.}
Algorithm~\ref{alg:decomposition} exploits the ``horizontal containment'' relationships only at the first level of the search space.
For a given $\Delta$, once the restricted starting set of vertices has been defined for $k = 1$, the traditional core decomposition is started to produce all the span-cores of span $\Delta$.
In other words, for $k > 1$ only the ``vertical containment'' is exploited.
Consider the span-core $C_{3,[1,2]}$ in Figure \ref{fig:searchspace}: we know that it is a subset of  $C_{2,[1,2]}$ (``vertical'' ) and of $C_{3,[1,1]}$ and $C_{3,[2,2]}$  (``horizontal'' ).
One could consider intersecting all these three span-cores before computing $C_{3,[1,2]}$.
We tested this alternative approach, but concluded that the overhead of computing intersections and data-structure maintenance was outweighing the benefit of starting from a smaller vertex set.


The worst-case time complexity of Algorithm~\ref{alg:decomposition} is equal to the na\"{\i}ve approach, however, in practice, it is orders of magnitude faster, as shown in Section~\ref{sec:experiments}.
\begin{figure}[t]
\begin{tabular}{ccc}
\includegraphics[width=0.275\columnwidth]{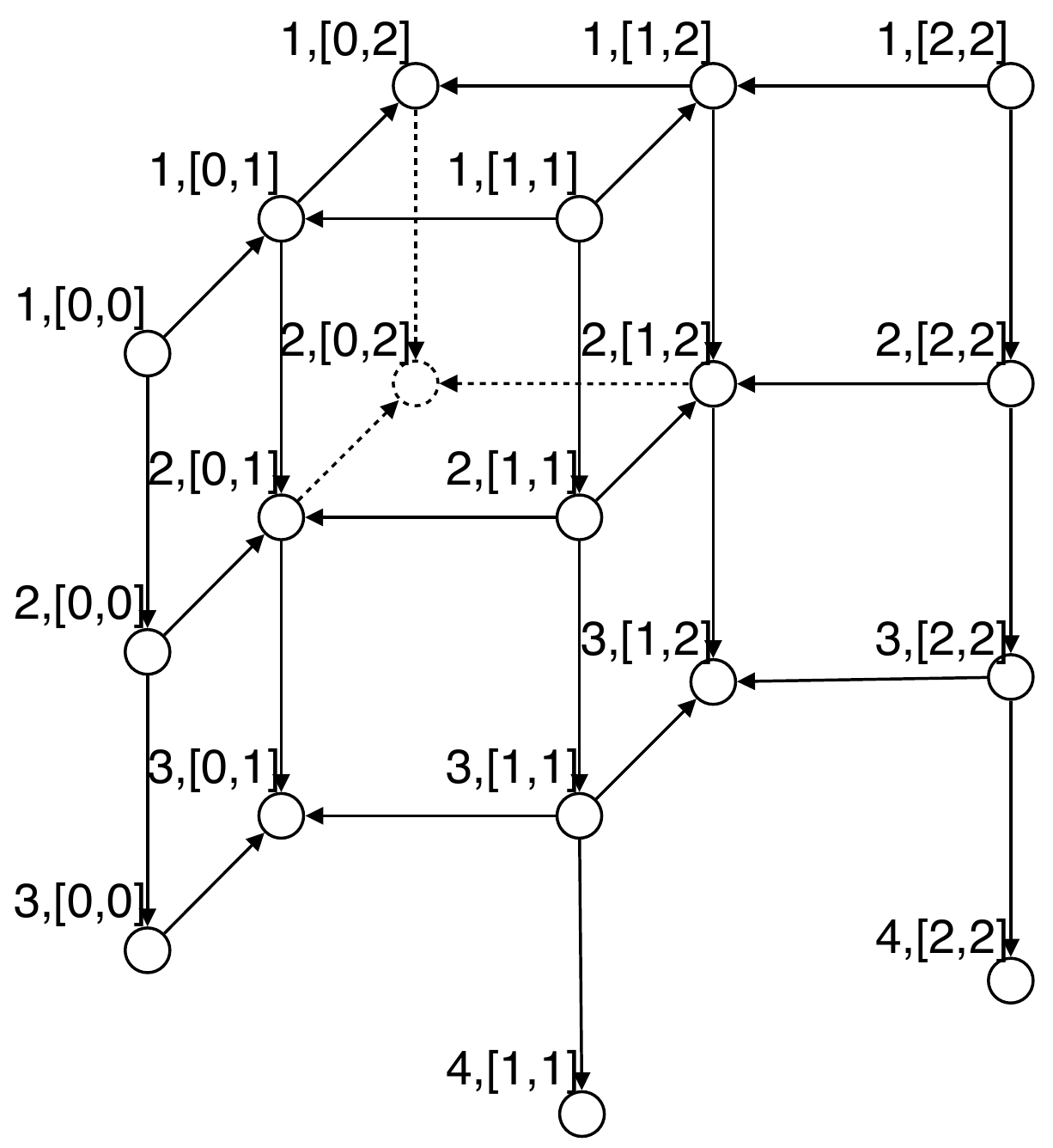} \hspace{0.5cm} & \includegraphics[width=0.275\columnwidth]{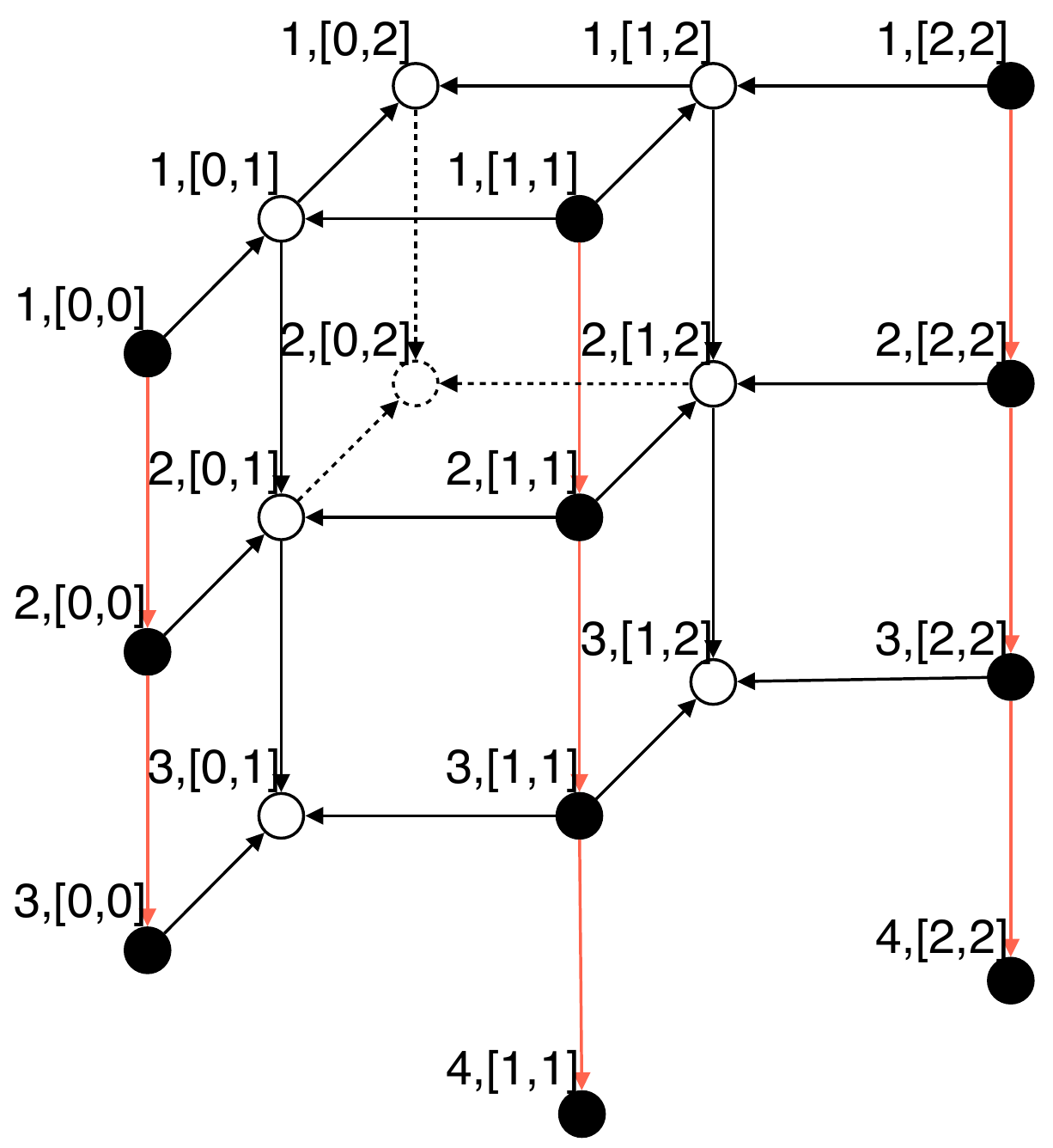} & \hspace{0.5cm} \includegraphics[width=0.275\columnwidth]{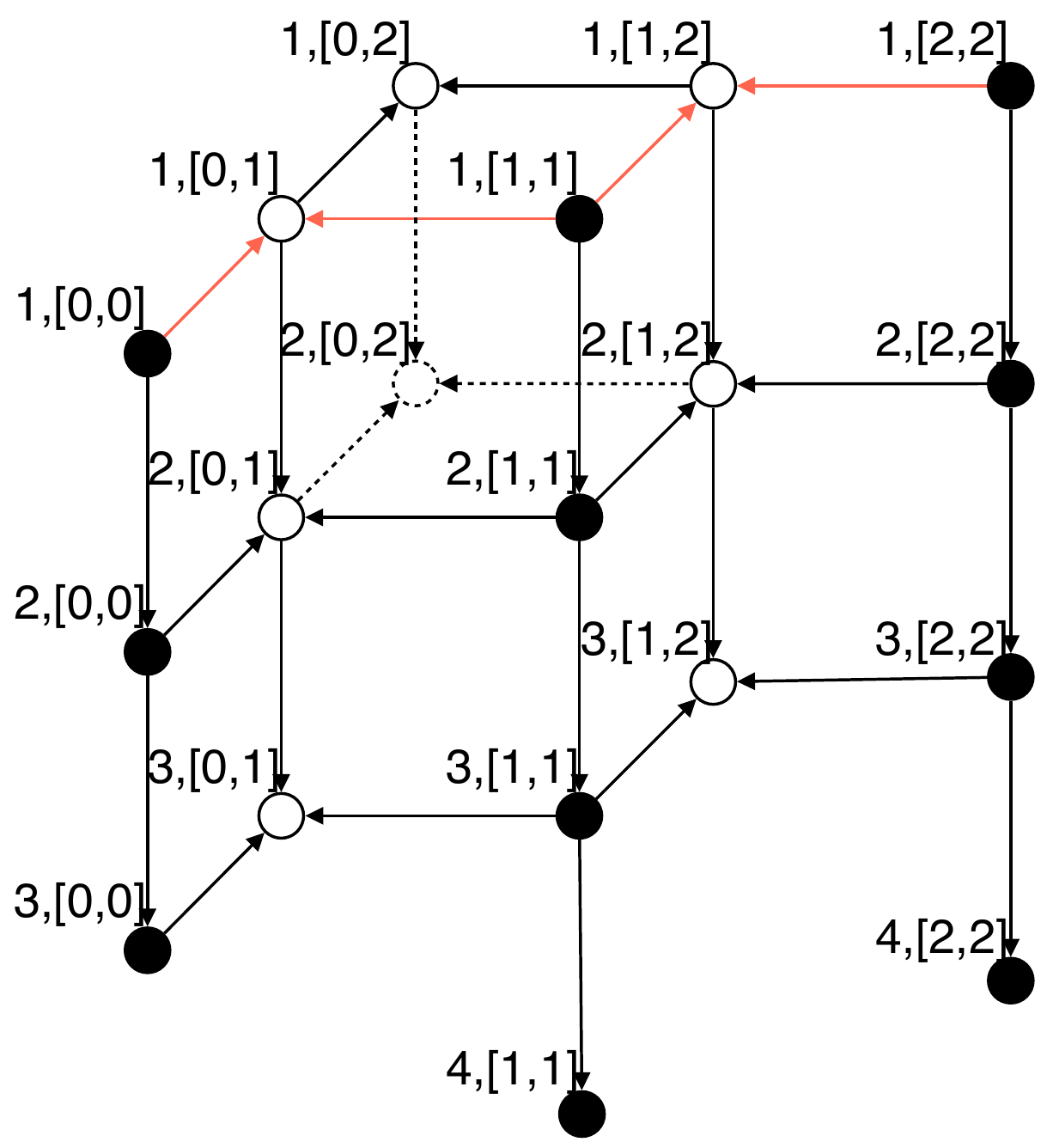}\\
(a) \hspace{0.5cm} & (b) & \hspace{0.5cm} (c)\\
\multicolumn{3}{c}{}\vspace{0.25cm}\\
\includegraphics[width=0.275\columnwidth]{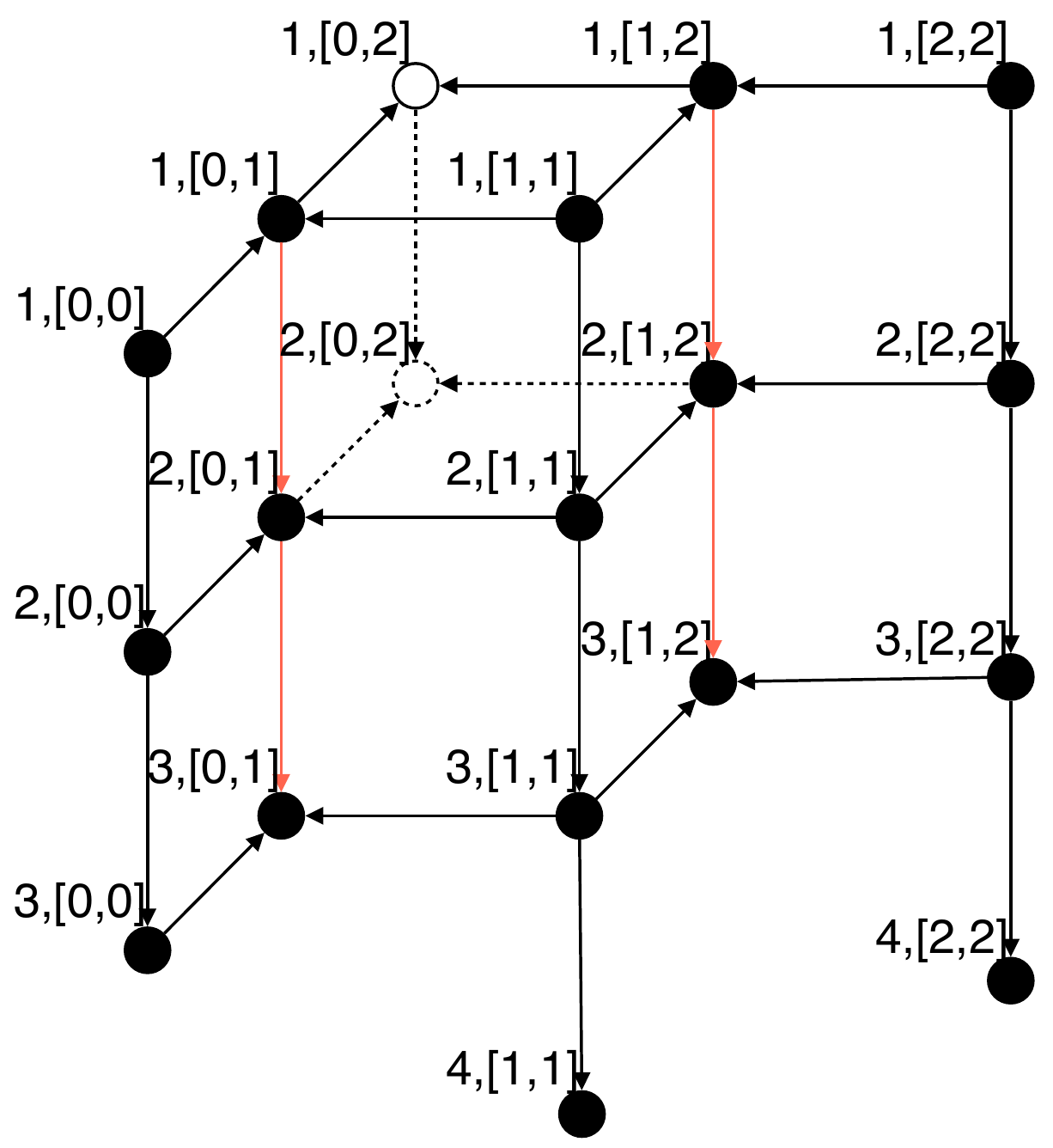} \hspace{0.5cm} & \includegraphics[width=0.275\columnwidth]{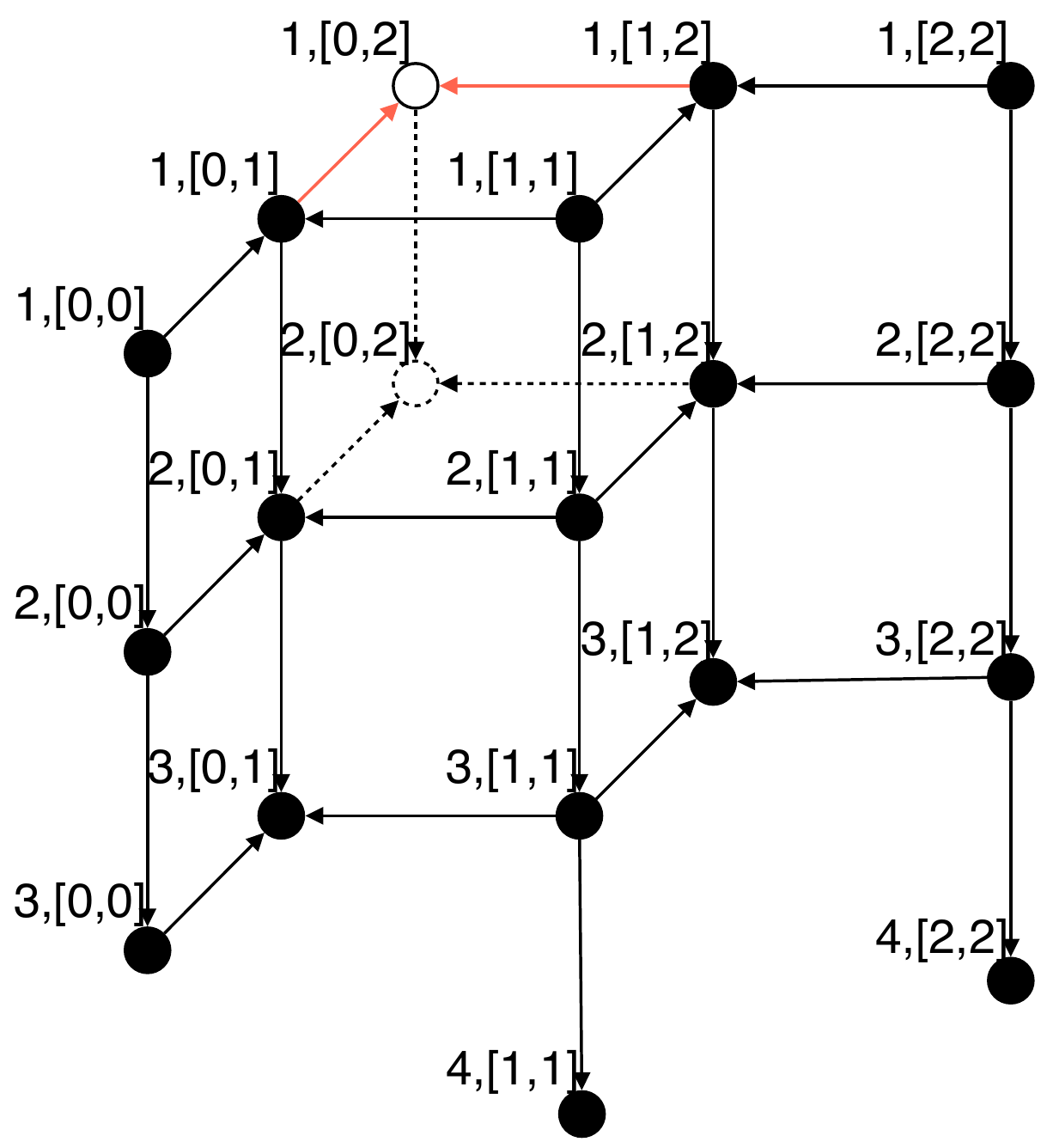} & \hspace{0.5cm} \includegraphics[width=0.275\columnwidth]{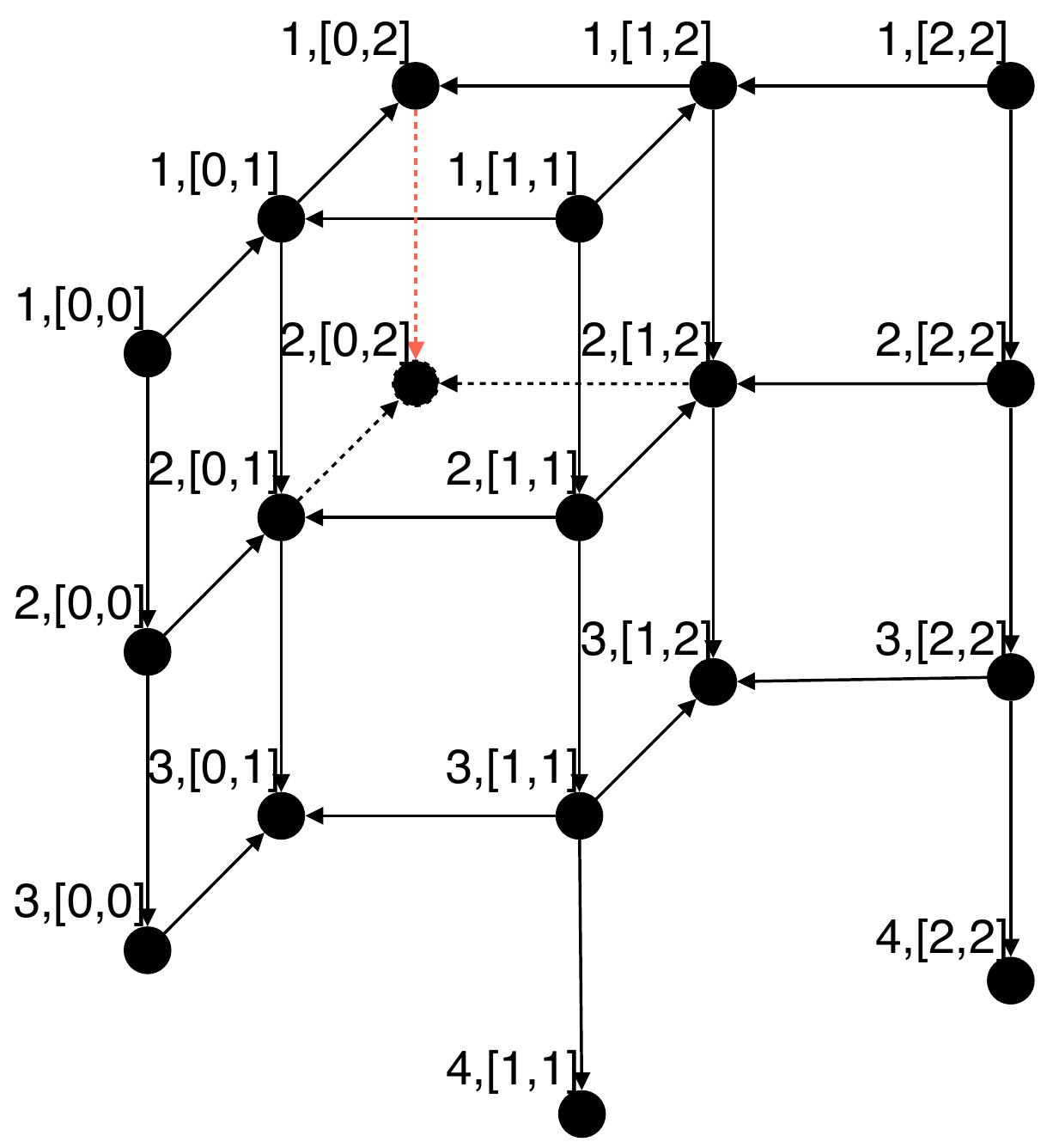}\\
(d) \hspace{0.5cm} & (e) & \hspace{0.5cm} (f)\\
\end{tabular}
\caption{\label{fig:running_spancores} Run-through example of the execution of \cores\ (Algorithm \ref{alg:decomposition}) over the search space of a temporal graph having $T = [0,2]$.
Full nodes represent computed \spancores, while empty nodes are \spancores that will be visited in the next steps of the algorithm.
Red arrows highlight the containment relationships exploited during the current step.}
\end{figure}

\begin{myexample}
Figure~\ref{fig:running_spancores} reports a run-through example, illustrating the execution of  \cores\ (Algorithm \ref{alg:decomposition}) over the search space of a toy temporal graph having $T = [0,2]$ (shown in Figure~\ref{fig:running_spancores}(a)).
The algorithm starts by computing all the \spancores having span of size 1 (Figure~\ref{fig:running_spancores}(b)); in this case, only the ``vertical containment'' is exploited by the \textsf{core-decomposition} subroutine.
Then, \cores\  proceeds with the computation of the \spancores having span of size 2.
At first, the algorithm exploits the ``horizontal containment'' relationships at the first level of the search space to restrict the starting set of vertices for computing the \spancores of $k = 1$ (Figure~\ref{fig:running_spancores}(c)).
Afterwards, the \textsf{core-decomposition} subroutine computes all the \spancores with span of size 2, by following the ``vertical containment'' (Figure~\ref{fig:running_spancores}(d)).
Finally, the same method is applied for visiting the \spancores with span of size 3 (Figure~\ref{fig:running_spancores}(e)-(f)).
\end{myexample}

\section{Algorithms: computing maximal \spancores}
\label{sec:maximal_spancores}

In this section we focus on Problem~\ref{pbl:maximal}: computing the \emph{maximal} \spancores of a temporal graph.

\spara{A filtering approach.}
As anticipated above, a straightforward way of solving this problem consists in filtering the \spancores computed during the execution of Algorithm~\ref{alg:decomposition}, so as to ultimately output only the maximal ones.
This can easily be accomplished by equipping Algorithm~\ref{alg:decomposition} with a data structure $\mathcal{M}$ that stores the \spancore of the highest order for every temporal interval $\Delta \sqsubseteq T$ that has been processed by the algorithm.
Moreover, at the storage of a \spancore $C_{k,\Delta}$ in $\mathcal{M}$, the \spancores previously stored in $\mathcal{M}$ for subintervals of the temporal interval $\Delta$ and with the same order $k$ are removed from $\mathcal{M}$.
This removal operation, together with the order in which \spancores are processed, ensures that $\mathcal{M}$ eventually contains only the maximal \spancores.

%


\spara{Efficient maximal-span-core finding.}
Our next goal is to design a more efficient algorithm that extracts maximal \spancores directly, without computing complete core decompositions,
passing over more peripheral ones, and without generating all temporal cores.
This is a quite challenging design principle, as it contrasts the intrinsic structural properties of core decomposition, based on which a core of order $k$ is usually computed from the core of order $k\!-\!1$, thus making the computation of the core of the highest order as hard as computing the overall decomposition.
Nevertheless, thanks to theoretical properties that relate the maximal \spancores to each other, in the temporal context such a challenge can be achieved.
In the following we discuss such properties in detail, by starting from a result that has already been discussed above, but only informally.

Consider the classic core decomposition in a standard (\mbox{non-temporal}) graph $G$ (Definition~\ref{def:kcores}) and
let $C_{k^*}[G]$ denote the \emph{innermost} core of $G$, i.e., the non-empty $k$-core of $G$ with the largest~$k$.

\begin{mylemma}\label{lemma1}
Given a temporal graph $G = (V,T,\tau)$, let $\imcores$ be the set of all maximal \spancores of $G$, and $\mathbf{C_{inner}} = \{ C_{k^*}[G_\Delta]  \mid \Delta \sqsubseteq T\}$ be the set of innermost cores of all graphs $G_\Delta$. It holds that $\imcores \subseteq \mathbf{C_{inner}}$.
\end{mylemma}
\begin{proof}
Every $C_{k,\Delta} \in \imcores$ is the innermost core of the non-temporal graph $G_\Delta$: else,
there would exist another core $C_{k',\Delta} \neq \emptyset$ with $k' > k$, implying that $C_{k,\Delta} \notin \imcores$.
\end{proof}

Lemma~\ref{lemma1} states that each maximal \spancore is an innermost core of a $G_\Delta$, for some temporal interval $\Delta \sqsubseteq T$.
Hence, there can exist at most one maximal \spancore for every $\Delta \sqsubseteq T$ (while an interval $\Delta$ may not yield any maximal \spancore).
The key question to design an efficient maximal-\spancore-mining algorithm thus becomes how to extract innermost cores of
the graphs $G_\Delta$ more efficiently than by computing the full core decompositions of all $G_\Delta$.
The answer to this question comes from the result stated in the next two lemmas (with Lemma~\ref{lemma2} being auxiliary to Lemma~\ref{lemma3}).

\begin{mylemma}\label{lemma2}
Given a temporal graph $G = (V,T,\tau)$, and three temporal intervals $\Delta = [t_s,t_e] \sqsubseteq T$, $\Delta' = [t_s\!-\!1,t_e] \sqsubseteq T$, and $\Delta'' = [t_s,t_e\!+\!1] \sqsubseteq T$.
The innermost core $C_{k^*}[G_\Delta]$ is a maximal span-core of $G$ if and only if $k^* > \max\{k',k''\}$ where $k'$ and $k''$ are the orders of the innermost cores of $G_{\Delta'}$ and $G_{\Delta''}$, respectively.
\end{mylemma}
\begin{proof}
The ``$\Rightarrow$'' part comes directly from the definition of maximal \spancore (Definition~\ref{def:maximal}): if $k^*$ were
not  larger than $\max\{k',k''\}$, then $C_{k^*}[G_\Delta]$ would be dominated by another \spancore both on the order and on the span (as both $\Delta'$ and $\Delta''$ are superintervals of $\Delta$).
For the ``$\Leftarrow$'' part, from Lemma~\ref{lemma1} and Proposition~\ref{prp:conteinment} it follows that $\max\{k',k''\}$ is an upper bound on the maximum order of a \spancore of a superinterval of $\Delta$.
Therefore, $k^* > \max\{k',k''\}$ implies that there cannot exist any other \spancore that dominates $C_{k^*}[G_\Delta]$ both on the order and on the span.
\end{proof}

\begin{mylemma}\label{lemma3}
Given $G$, $\Delta$,  $\Delta'$,  $\Delta''$,  $k'$, and $k''$ defined as in Lemma~\ref{lemma2}, let $\widetilde{V} = \{u \in V \mid \tdeg_{\Delta}(V, u) > \max\{k', k''\}\}$,
and let $C_{k^*}[G_\Delta[\widetilde{V}]]$  be the innermost core of $G_\Delta[\widetilde{V}]$.
If $k^* >  \max\{k', k''\}$, then $C_{k^*}[G_\Delta[\widetilde{V}]]$ is a maximal \spancore; otherwise, no maximal \spancore exists for $\Delta$.
\end{mylemma}
\begin{proof}
Lemma~\ref{lemma2} states that, to be recognized as a maximal \spancore, the innermost core of $G_{\Delta}$ should have order larger than $\max\{k', k''\}$.
This means that, if the innermost core of $G_{\Delta}$ is a maximal \spancore, all vertices $u \notin \widetilde{V}$ cannot be part of it.
Therefore, $G_{\Delta}$ yields a maximal \spancore only if the innermost core of subgraph  $G_{\Delta}[\widetilde{V}]$ has order $k^* >  \max\{k', k''\}$.
\end{proof}
Lemma~\ref{lemma3} provides the basis of our efficient method for extracting maximal \spancores.
Basically, it states that, to verify whether a certain temporal interval $\Delta = [t_s,t_e]$ yields a maximal \spancore (and, if so, compute it), there is no need to consider the whole graph $G_\Delta$, rather it suffices to start from a smaller subgraph, which is given by all vertices whose temporal degree is larger than the maximum between the orders of the innermost cores of intervals $\Delta' = [t_s\!-\!1,t_e]$ and $\Delta'' = [t_s,t_e\!+\!1]$.
This finding suggests a strategy that is opposite to the one used for computing the overall \spancore decomposition:
a \emph{top-down} strategy that processes temporal intervals starting from the larger ones.
Indeed, in addition to exploiting the result in Lemma~\ref{lemma3}, this way of exploring the temporal-interval space allows us to skip the computation of complete core decompositions of the whole ``singleton-interval'' graphs $\{G_{_{[t,t]}}\}_{t \in T}$, which may easily become a critical bottleneck, as they are the largest ones among the graphs induced by temporal intervals.

\begin{algorithm}[t]
\DontPrintSemicolon
\KwIn{A temporal graph $G=(V,T,\tau)$.}
\KwOut{The set $\imcores$ of all maximal \spancores of $G$.}

$\imcores \leftarrow \emptyset$\;
$\mathcal{K}'[t] \gets 0$, $\forall t \in T$\;

\ForAll{$t_s \in [0, 1, \ldots, t_{max}]$}
{\label{line:imcores:extfor}
	$t^* \leftarrow \max\{ t_e \in [t_s, t_{max}] \mid E_{_{[t_s,t_e]}} \neq \emptyset\}$\; \label{line:imcores:t}
	$k'' \leftarrow 0$\;
	\ForAll{$t_e \in [t^*, t^*\!-\!1, \ldots, t_s]$}
	{\label{line:imcores:intfor} 
		$\Delta \leftarrow [t_s,t_e]$\; \label{line:imcores:delta}
		$lb\gets \max\{\mathcal{K}'[t_e],k''\}$\; \label{line:imcores:lb}
		$V_{lb} \leftarrow \{u \in V \mid \tdeg_{\Delta}(V,u) > lb\}$\; \label{line:imcores:V}
		$E_\Delta[V_{lb}] \gets \{(u,v) \in E_{\Delta} \mid u \in V_{lb}, v \in V_{lb}\}$\; \label{line:imcores:E}
		$C \leftarrow $ \innermost $(V_{lb}, E_{\Delta}[V_{lb}])$\; \label{line:imcores:core}
		$k^* \leftarrow $ order of $C$\; \label{line:imcores:core2}
		\If{$k^* > lb$}
		{\label{line:imcores:updatesolution}
			$\imcores \leftarrow \imcores \cup \{C\}$\; \label{line:imcores:updatesolution2}
		}
		$k'' \gets \max\{k'', k^*\}$; \ $\mathcal{K}'[t_e] \gets \max\{\mathcal{K}'[t_e], k''\}$\; \label{line:imcores:updatek}
	}
}
\caption{\innermosts}\label{alg:imcores}
\end{algorithm}

\spara{The \innermosts\ algorithm.}
Algorithm~\ref{alg:imcores} iterates over all timestamps $t_s \in T$ in \emph{increasing order} (Line~\ref{line:imcores:extfor}), and for each $t_s$ it first finds all the maximal span-cores that have span starting in $t_s$.
This way of proceeding \emph{ensures that a span-core that is recognized as maximal will not be later dominated by another span-core}.
Indeed, an interval $[t_s,t_e]$ can never be contained in another interval $[t_s',t_e']$ with $t_s < t_s'$.
For a given $t_s$, all maximal span-cores are computed as follows.
First, the maximum timestamp $\geq t_s$ such that the corresponding edge set $E_{_{[t_s,t_e]}}$ is not empty is identified as $t^*$ (Line~\ref{line:imcores:t}).
Then, all intervals $\Delta = [t_s, t_e]$ are considered one by one in \emph{decreasing order} of $t_e$ (Lines~\ref{line:imcores:intfor}--\ref{line:imcores:delta}): this again \emph{guarantees that a span-core that is recognized as maximal will not be later dominated by another span-core, as the intervals are processed from the largest to the smallest}.
At each iteration of the internal cycle, the algorithm resorts to Lemma~\ref{lemma3} and computes the lower bound $lb$ on the order of the innermost core of $G_{\Delta}$ to be recognized as maximal, by taking the maximum between $\mathcal{K}'[t_e]$ and $k''$ (Line~\ref{line:imcores:lb}).
$\mathcal{K}'$ is a map that maintains, for every timestamp $t \in [t_s, t^*]$, the order of the innermost core of graph $G_{\Delta'}$, where $\Delta' = [t_s\!-\!1, t]$ (i.e., $\mathcal{K}'[t]$ stores what in Lemmas~\ref{lemma2}--\ref{lemma3} is denoted as $k'$).
Whereas $k''$ stores the order of the innermost core of $G_{\Delta''}$, where $\Delta'' = [t_s, t_e+1]$.
Afterwards, the sets of vertices $V_{lb}$ and of edges $E_{\Delta}[V_{lb}]$ that comply with this lower-bound constraint are built (Lines~\ref{line:imcores:V}--\ref{line:imcores:E}), and the innermost core of the subgraph $(V_{lb}, E_{\Delta}[V_{lb}])$ is extracted (Lines~\ref{line:imcores:core}--\ref{line:imcores:core2}).
Ultimately, based again on Lemma~\ref{lemma3}, such a core is added to the output set of maximal \spancores only if its order is actually larger than $lb$ (Lines~\ref{line:imcores:updatesolution}--\ref{line:imcores:updatesolution2}), and the values of $k''$ and $\mathcal{K}'[t_e]$ are updated (Line~\ref{line:imcores:updatek}).
Specifically, note that the order $k^*$ of core $C$ may in principle be less than $k''$, as $C$ is extracted from a subgraph of $G_{\Delta}$.
If this happens, it means that the actual order of the innermost core of $G_{\Delta}$ is equal to $k''$.
This motivates the update rules (and their order) reported in Line~\ref{line:imcores:updatek}.

\begin{mytheorem}\label{th:correctnessAlg3}
Algorithm~\ref{alg:imcores} is sound and complete for Problem~\ref{pbl:maximal}.
\end{mytheorem}
\begin{proof}
The algorithm processes all temporal intervals $\Delta \sqsubseteq T$ yielding a non-empty edge set $E_\Delta$, in an order such that no interval is processed before one of its superintervals: this guarantees that a span-core recognized as maximal will not be dominated by another span-core found later on. For every $\Delta$ it extracts a core $C$ that is used as a proxy of the innermost core of graph $G_{\Delta}$.
$C$ is added to the output set $\imcores$
only if Lemma~\ref{lemma3} recognizes it as a maximal \spancore, otherwise it is discarded.
This proves the soundness of the algorithm.
Completeness follows from Lemma~\ref{lemma1}, which states that to extract all maximal \spancores it suffices to focus on the innermost cores of graphs $\{G_{\Delta} \mid \Delta \sqsubseteq T\}$, and Lemma~\ref{lemma3} again, which states the condition for a proxy core $C$ to be safely  discarded because it is a non-maximal \spancore.
\end{proof}

\spara{Discussion.}
The worst-case time complexity of Algorithm~\ref{alg:imcores} is the same as the algorithm for computing the overall \spancore decomposition, i.e., $\bigO(|T|^2 \times |E|)$.
It is worth mentioning that it is not possible to do better than this, as the output itself is potentially quadratic in $|T|$.
However, as we will show in Section~\ref{sec:experiments}, the proposed algorithm is in practice much more efficient than computing the overall \spancore decomposition and filtering out the non-maximal span-cores as, in this case, we avoid the visit of portions of the \spancore search space and the computations are run over subgraphs of reduced dimensions.

To conclude, we discuss how the crucial operation of building the subgraph $(V_{lb}, E_{\Delta}[V_{lb}])$ may be carried out efficiently in terms of both time and space.
Consider a fixed timestamp $t_s \in [0, \ldots, t_{max}]$.
The following reasoning holds for every $t_s$.
Let $E^-(t_e) = E_{_{[t_s,t_e]}} \setminus E_{_{[t_s,t_e\!+1]}}$ be the set of edges that are in $E_{_{[t_s,t_e]}}$ but not in $E_{_{[t_s,t_e\!+1]}}$, for
 $t_e \in [t_s, \ldots, t^*\!-1]$.
As a first general step, for each $t_s$, we compute and store \emph{all} edge sets $\{E^-(t_e)\}_{t_e \in [t_s, t^*\!-1]}$.
These operations can be accomplished in $\mathcal{O}(|T| \times |E|)$ overall time, because every $E^-(t_e\!)$ can be computed incrementally from $E_{_{[t_s,t_e]}}$ as $E^-(t_e) = \{(u,v) \in E_{_{[t_s,t_e]}} \mid \tau(u,v,t_e\!+\!1) = 0\}$.
Moreover, for any timestamp $t_e$, we keep a map $\mathcal{D}$ storing all vertices of $G_{_{[t_s,t_e]}}$ organized by degree.
Specifically, the set $\mathcal{D}[k]$ contains all vertices having degree $> k$ in  $G_{_{[t_s,t_e]}}$.
Every vertex in $\mathcal{D}$  is thus replicated a number of times equal to its degree.
This way, the overall space taken by $\mathcal{D}$ is $\bigO(|E|)$, i.e., as much space as $G$.
$\mathcal{D}$ is initialized as empty (when $t_e = t^*$) and repeatedly augmented as $t_e$ decreases, by a linear scan of the various $E^-(t_e)$.
The overall filling of $\mathcal{D}$ (for all $t_e$) therefore takes $\mathcal{O}(|T| \times |E|)$ time.
Then, the desired $V_{lb}$ can be computed in constant time simply as $V_{lb} = \mathcal{D}[lb]$.

As for $E_{\Delta}[V_{lb}]$, for any $t_e$, we first reconstruct $E_{_{[t_s,t_e]}}$ as $E_{_{[t_s,t_e+\!1]}} \cup E^-(t_e)$, having previously computed $E_{_{[t_s,t_e+\!1]}}$.
Note that storing all $E^-(t_e)$ takes $\bigO(|E|)$ space.
That is why we store all $E^-(t_e)$ and reconstruct $E_{_{[t_s,t_e]}}$ afterward (instead of storing the latter, which would take $\bigO(|T| \times |E|)$ space).
$E_{\Delta}[V_{lb}]$ is ultimately derived by a linear scan of $E_{_{[t_s,t_e]}}$, taking all edges in $E_{_{[t_s,t_e]}}$ having both endpoints in $V_{lb}$.
This way, the step of building $E_{\Delta}[V_{lb}]$  for all $t_e$ takes again $\mathcal{O}(|T| \times |E|)$ overall time.


\begin{myexample}
We report here a run-through example of the execution of \innermosts\ (Algorithm \ref{alg:maximaltemporalcs}) over the search space of a temporal graph having $T = [0,2]$ (the same shown in Figure~\ref{fig:running_spancores}(a)).
\innermosts\ starts by identifying the \spancore of highest order in the largest possible temporal interval $\Delta$, i.e., $\Delta = T$ (Figure~\ref{fig:running_maximal}(a)).
Such a \spancore is guaranteed to be maximal, since the \spancore of highest order with span $T$ cannot be dominated in terms of span by any other \spancore.
The algorithm then processes interval $[0,1]$ (Figure~\ref{fig:running_maximal}(b)): here $lb = 2$, since the only constraint derives from the identification of core $C_{2,[0,2]}$ as maximal, therefore core $C_{3,[0,1]}$ is recognized as maximal.
Next, the algorithm searches for the last possible maximal \spancore having span $\Delta = [t_s,t_e]$ such that $t_s = 0$, i.e., $\Delta = [0,0]$ (Figure~\ref{fig:running_maximal}(c)).
Core $C_{3,[0,0]}$ is computed, but discarded from the solution maximal, because it has order equal to the lower bound $lb = 3$, derived from core $C_{3,[0,1]}$.
The algorithm proceeds in a similar way, by finding a maximal \spancore in all the remaining intervals, i.e., $[1,2]$, $[1,1]$, and $[2,2]$ (Figures~\ref{fig:running_maximal}(d)-(e)-(f)).
It is important to note that in all such cases, the lower bound $lb$ for the existence of a maximal \spancore in a given temporal interval accounts for two factors.
For example, consider interval $[1,1]$.
The innermost core in $G_{[1,1]}$ is a maximal \spancore if it has order greater than both cores $C_{3,[0,1]}$ and $C_{3,[1,2]}$, that is $lb = 3$.
\end{myexample}
\begin{figure}[t]
\begin{tabular}{ccc}
\includegraphics[width=0.275\columnwidth]{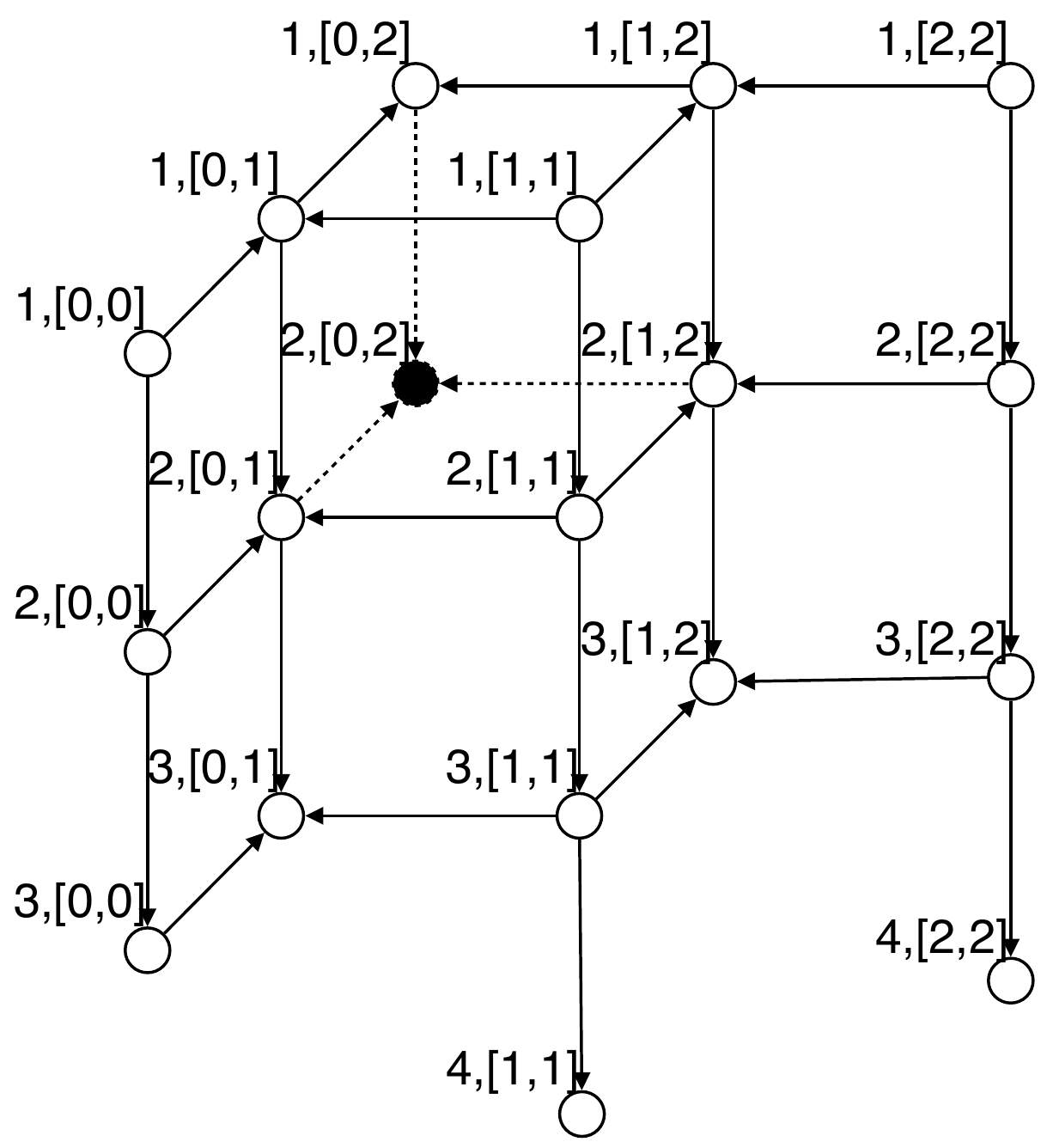} \hspace{0.5cm} & \includegraphics[width=0.275\columnwidth]{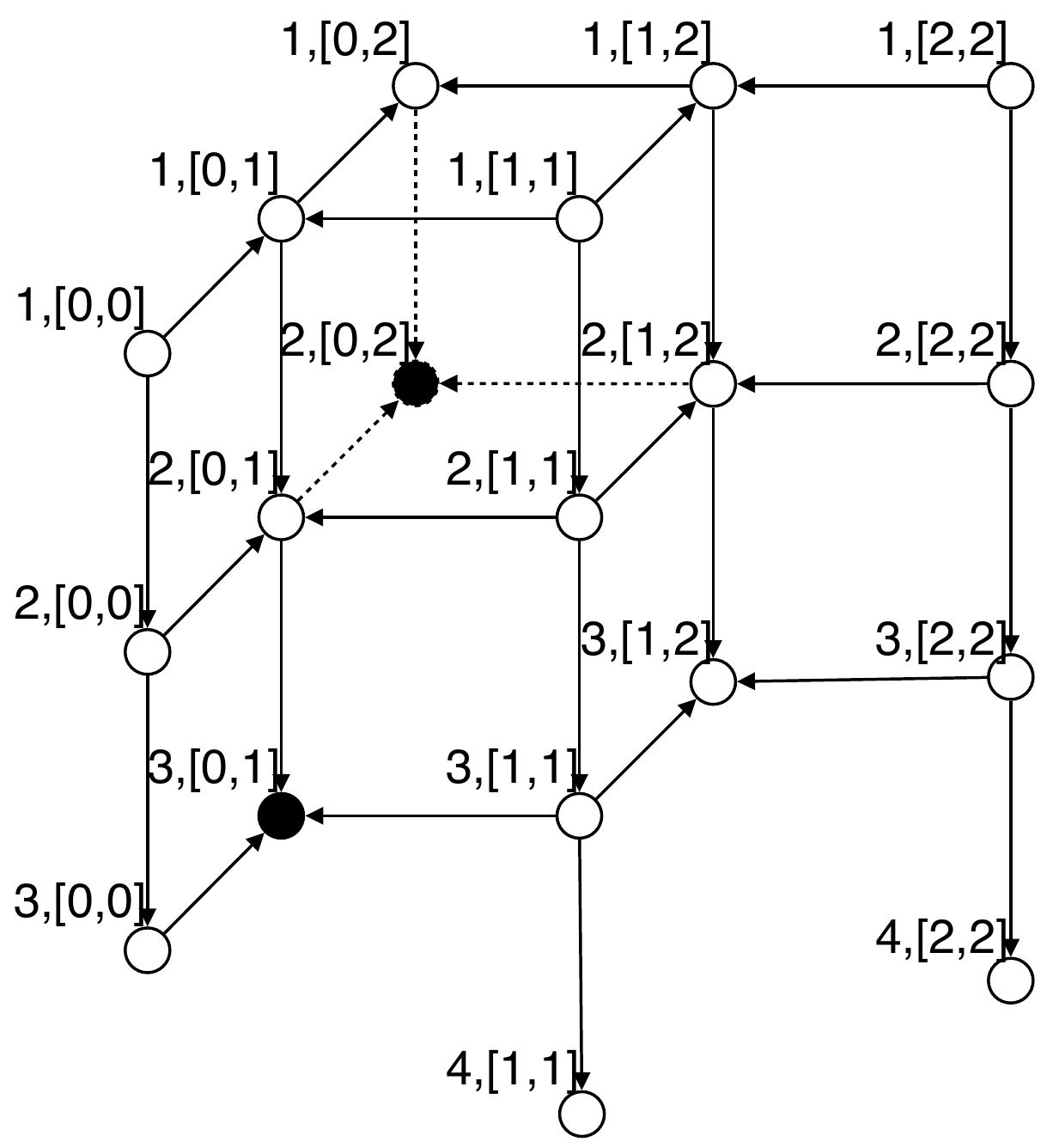} & \hspace{0.5cm} \includegraphics[width=0.275\columnwidth]{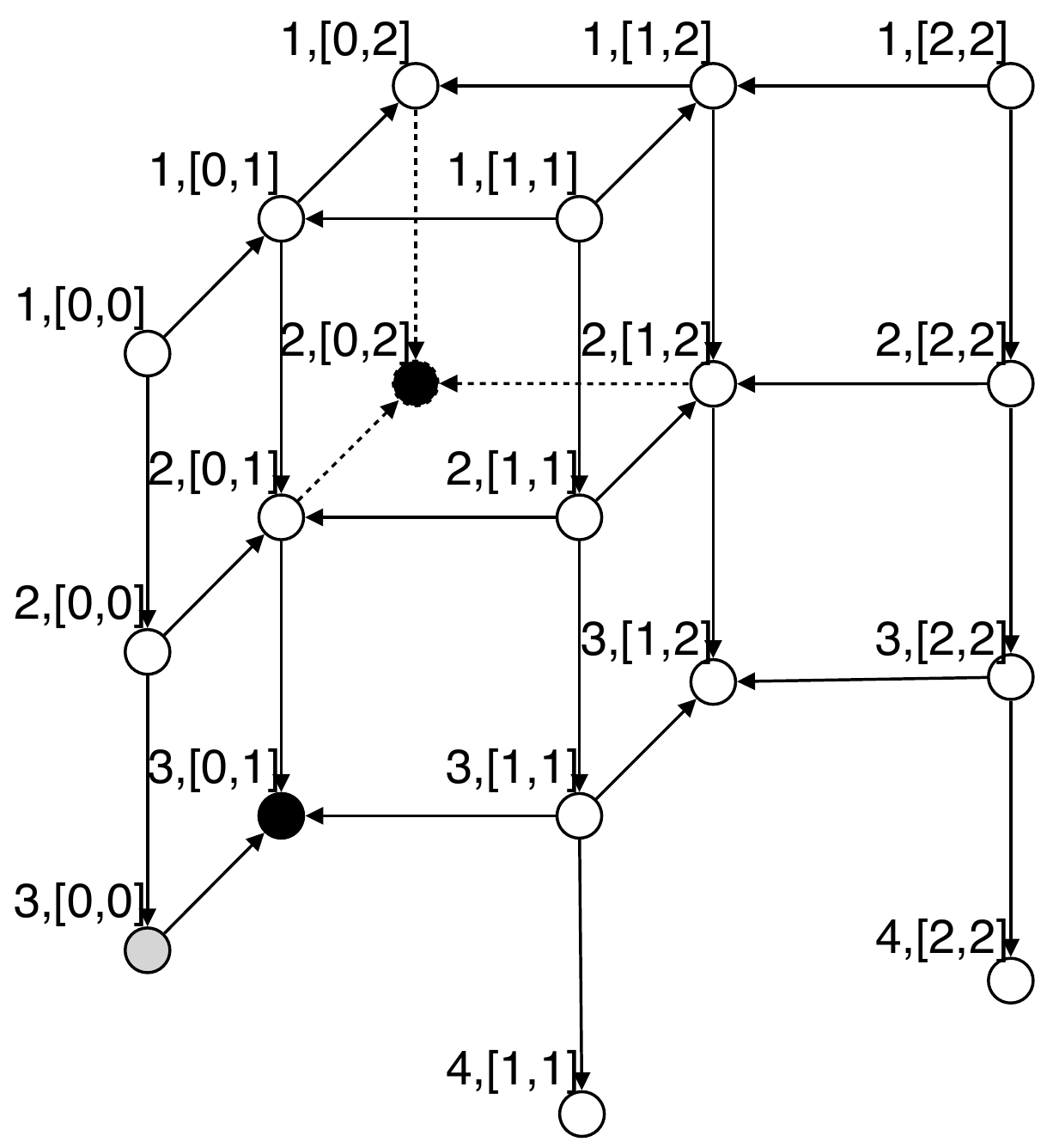}\\
(a) $[0,2]$ \hspace{0.5cm} & (b) $[0,1]$ & \hspace{0.5cm} (c) $[0,0]$\\
\multicolumn{3}{c}{}\vspace{0.25cm}\\
\includegraphics[width=0.275\columnwidth]{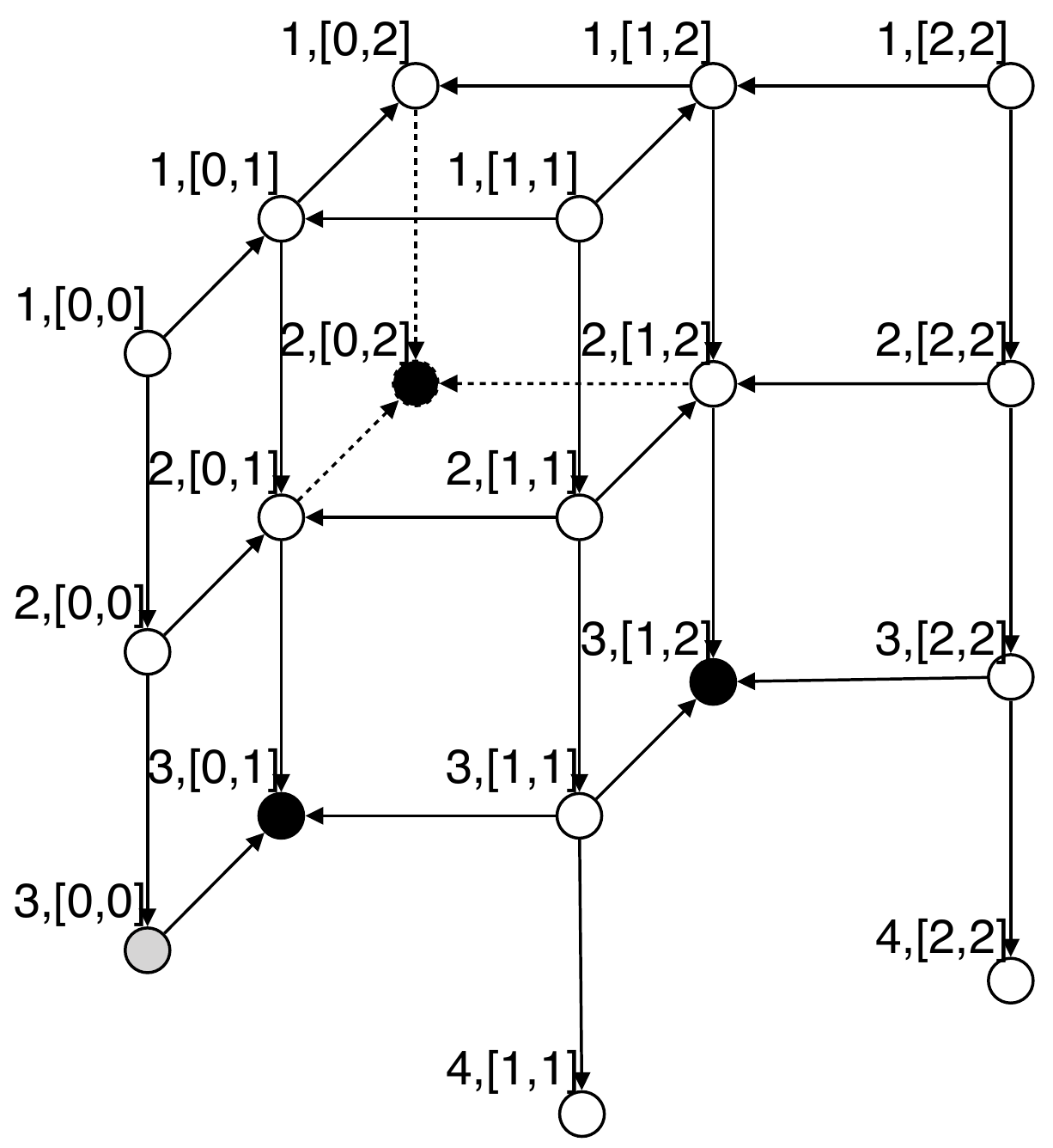} \hspace{0.5cm} & \includegraphics[width=0.275\columnwidth]{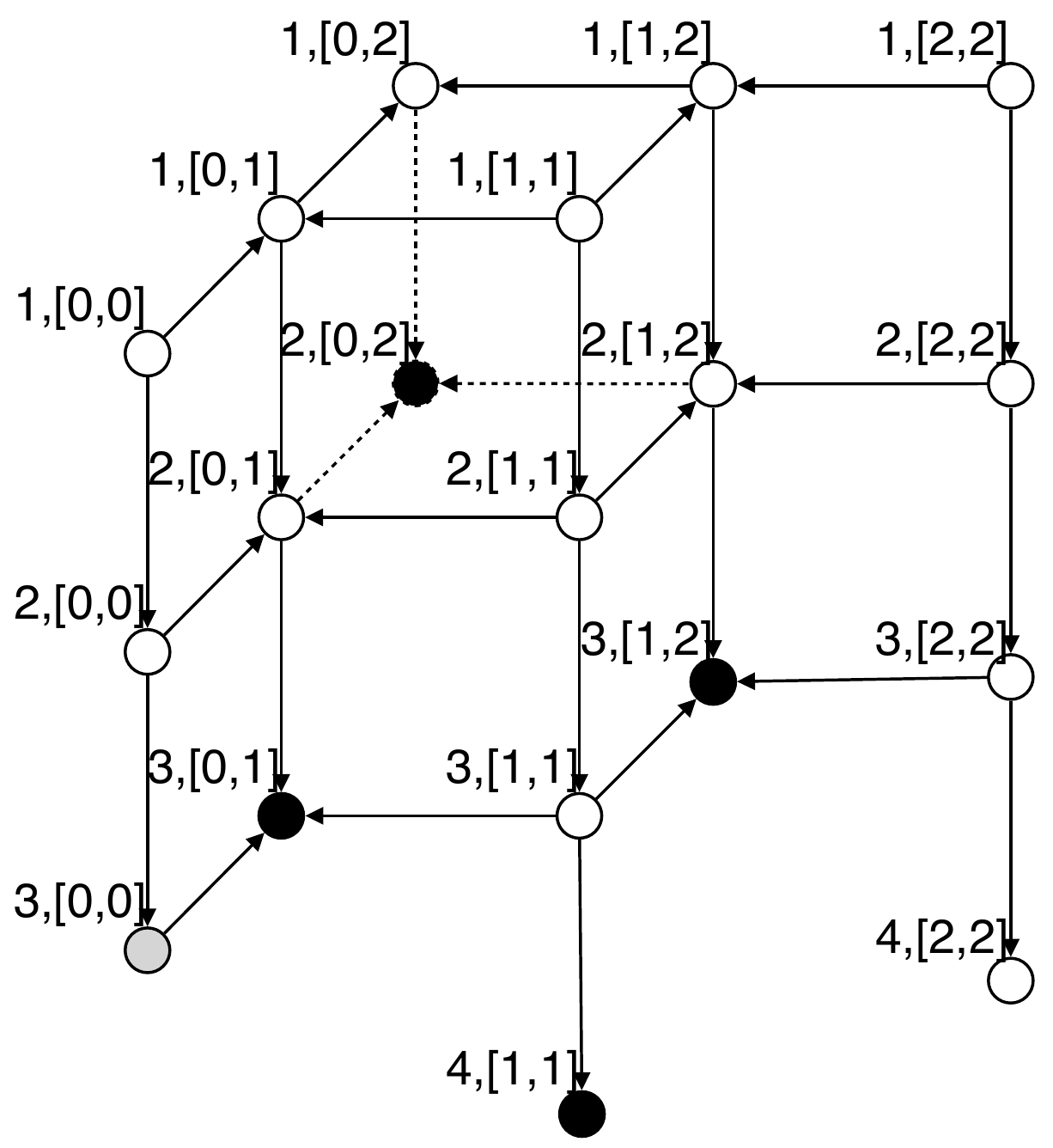} & \hspace{0.5cm} \includegraphics[width=0.275\columnwidth]{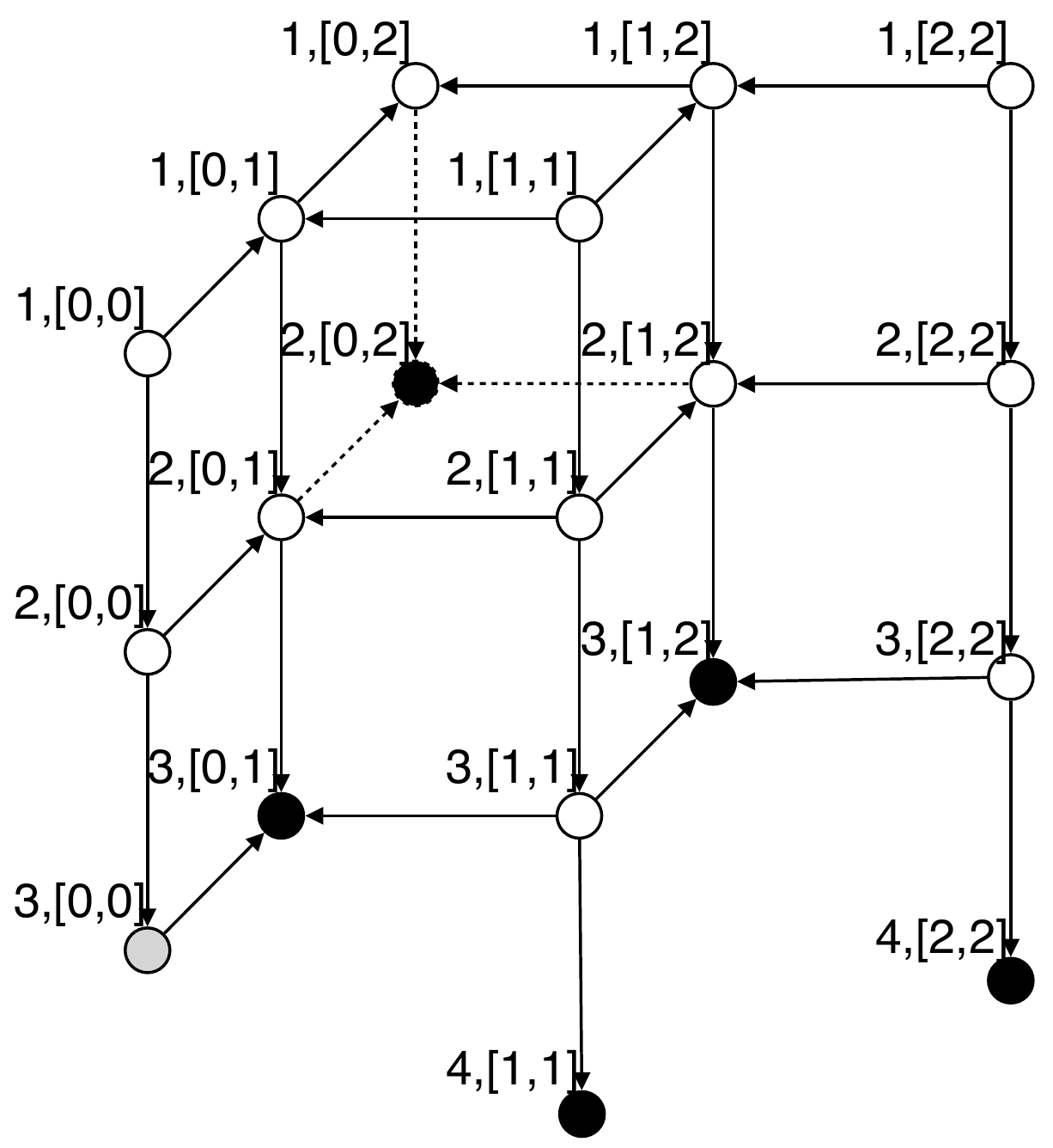}\\
(d) $[1,2]$ \hspace{0.5cm} & (e) $[1,1]$ & \hspace{0.5cm} (f) $[2,2]$\\
\end{tabular}
\caption{\label{fig:running_maximal} Run-through example of the execution of \innermosts\ (Algorithm \ref{alg:maximaltemporalcs}) over the search space of a temporal graph having $T = [0,2]$ (same as Figure~\ref{fig:running_spancores}(a)).
Full nodes represent computed \spancores: \spancores in black are recognized as maximal, while those in gray are discarded from the set of the maximal \spancores.}
\end{figure} 

\section{Temporal community search}
\label{sec:community_search}

Community search in static graphs aims at finding a dense subgraph (community) containing a set of input query vertices~\cite{fang2019survey, HuangLX17}.
In the temporal setting it is very likely that the communities spanning the query vertices change over time.
To be more precise, it may happen that a certain subgraph $S$ is a well-representative community for the given query vertices $Q$, but only for a certain time interval $\Delta$. Instead, for another time interval $\Delta'$, a relevant community for $Q$ might correspond to a completely different subgraph $S'$.
For this reason, we formulate community search on temporal networks  as the problem of finding $h$ subgraphs (with $h > 0$ being an input parameter) containing the query vertices, together with their temporal span, such that the sum of the density of those subgraphs is maximized and the union of their temporal spans corresponds to the whole input temporal domain.
{
Among the many densities proposed in the literature, here we follow the bulk of the literature on community search, and adopt the minimum-degree density~\cite{fang2019survey,HuangLX17}.
In fact, as well-discussed, among others, by  Sozio~and~Gionis~\cite{Sozio} in their seminal work, unlike other density notions, including the popular average degree, the minimum-degree density has the capability of mitigating the so-called ``\emph{free-rider}'' effect, i.e., the fact that (large) subgraphs may be arbitrarily added to a community-search solution to artificially increase the objective-function value, and thus lead to unintuitive yet unnecessarily large output solutions.
}
Formally, the problem we study in this work is:

\begin{problem}[\temporalcs]\label{prob:temporalcs}
Given a temporal graph $G = (V, T, \tau)$, a set $Q \subseteq V$ of query vertices, and a positive integer $h \in \mathbb{N}^+$,
find a set $\{\langle S_i, \Delta_i \rangle \}_{i=1}^h$ of $h$ pairs  such that ($i$) $\forall 1 \leq i \leq h : Q \subseteq S_i \subseteq V$, \ ($ii$) $\bigcup_{1 \leq i \leq h} \Delta_i = T$, \ and ($iii$) the following is maximized:
\begin{equation}\label{eq:temporalcs}
\sum_{i=1}^h \min_{u \in S_i}  \tdeg_{\Delta_i}(S_i,u).
\end{equation}
\end{problem}

The input integer $h$ is a user-defined parameter that gives the analyst the flexibility of requiring a specific number of output temporal communities, which might vary from application to application.


\subsection{Connection with \textsc{Sequence Segmentation}}

Here we provide some theoretical insights into the \temporalcs problem.
The main result we provide at the end of this subsection is an interesting connection with the well-established \seqseg problem~\cite{bellman61approximation}.
As shown in the next subsections, such a result forms the basis for algorithmic design.

Let us first consider a single-interval variant of Problem~\ref{prob:temporalcs}: for a fixed temporal interval $\Delta$, find a subgraph containing the input set $Q$ of query vertices that maximizes  the minimum temporal degree within $\Delta$.
Formally:

\begin{problem}[\singletemporalcs]\label{prob:temporalcs_subproblem}
Given a temporal graph $G = (V,T,\tau)$, a set $Q \subseteq V$ of query vertices, and an interval $\Delta \sqsubseteq T$, find
$$
S^* = \operatorname{argmax}_{Q \subseteq S \subseteq V} \min_{u \in S} \tdeg_{\Delta}(S,u).
$$
\end{problem}

It is easy to see that solving Problem~\ref{prob:temporalcs_subproblem}  corresponds to solving minimum-degree-based community search on graph $G_{\Delta}$.
Therefore, a solution to Problem~\ref{prob:temporalcs_subproblem} can straightforwardly be computed by applying a standard result on minimum-degree-based community search, which states that the highest-order core containing all query vertices is a solution to that problem~\cite{BarbieriBGG15}.
This finding is formalized next.

%

\begin{mydefinition}[$(Q, \Delta)$-highest-order-\spancore]\label{def:deltaQhighestordercore}
Given a temporal graph $G = (V,T,\tau)$, a set $Q \subseteq V$  of query vertices, and an interval $\Delta \sqsubseteq T$,
the $(Q, \Delta)$-\emph{highest-order-\spancore} of $G$, denoted $\Qdeltahighest$, is defined as the highest-order \spancore among all \spancores of $G$ with temporal span $\Delta$ and containing all query vertices in $Q$.
Let also $\Qdeltascore$ denote the order of $\Qdeltahighest$.
\end{mydefinition}

\begin{fact}\label{fact1}
Given a temporal graph $G = (V,T,\tau)$, a set $Q \subseteq V$  of query vertices, and an interval $\Delta \sqsubseteq T$, the $(Q, \Delta)$-highest-order-\spancore of $G$ is a solution to Problem~\ref{prob:temporalcs_subproblem} on input $\langle G, Q, \Delta \rangle$.
\end{fact}

Note that Problem~\ref{prob:temporalcs_subproblem} may have multiple solutions: $\Qdeltahighest$ is only one of those possibly many ones.
$\Qdeltahighest$ can be computed by running a core decomposition  on (static) graph $G_{\Delta}$, and stopping it when the first core that does not contain all query vertices in $Q$ has been encountered.
Therefore, Problem~\ref{prob:temporalcs_subproblem} can be solved in $\bigO(|\Delta| \times |E|)$ time.

In light of the above findings, an alternative yet equivalent way of formulating our \temporalcs problem is to ask for a \emph{segmentation} (i.e., a partition) of the time domain $T$ into a set $\{\Delta_i\}_{i=1}^h$ of $h$ intervals so as to maximize the sum $\sum_{i = 1}^h  v^*_{Q, \Delta_i}$ of the orders of the $(Q, \Delta)$-highest-order-\spancores of those identified intervals.
Once such an optimal segmentation of $T$ has been computed, the ultimate $\{\langle S_i, \Delta_i \rangle \}_{i=1}^h$ pairs are derived by simply setting $S_i = C^*_{Q,\Delta_i}$, $\forall 1 \leq i \leq h$.
Formally:

\begin{problem}[Alternative formulation of Problem~\ref{prob:temporalcs}]\label{prob:alternativetemporalcs}
Given a temporal graph $G = (V, T, \tau)$, a set $Q \subseteq V$ of query vertices, and a positive integer $h \in \mathbb{N}^+$,
find a set $\{\langle S_i, \Delta_i \rangle \}_{i=1}^h$ of $h$ pairs  such that
($i$) $\forall 1 \leq i \leq h : S_i = C^*_{Q,\Delta_i}$,
\ ($ii$) $\{\Delta_i\}_{i=1}^h$ is a partition of $T$,
\ and ($iii$) the following is maximized:
\begin{equation}\label{eq:alternativetemporalcs}
\sum_{i=1}^h   v^*_{Q, \Delta_i}.
\end{equation}

\end{problem}

Correspondence between Problem~\ref{prob:temporalcs} and Problem~\ref{prob:alternativetemporalcs} easily follows from Fact~\ref{fact1} and from the observation that for any feasible solution $\{\langle S_i, \Delta_i \rangle \}_{i=1}^h$ to Problem~\ref{prob:temporalcs} with overlapping intervals, there exists an overlapping-interval-free feasible solution with not smaller objective-function value.
To see the latter, for any two overlapping intervals $\Delta_i$ and $\Delta_j$, simply replace one of the two intervals, say $\Delta_i$, with $\Delta'_i = \Delta_i \setminus (\Delta_i \cap \Delta_j)$. As $\Delta_i' \sqsubseteq \Delta_i$, it holds that $v^*_{Q, \Delta'_i} \geq v^*_{Q, \Delta_i}$, therefore the resulting overlapping-interval-free solution will have objective-function value greater than or equal to the objective-function value of the starting solution with overlapping intervals.

Thanks to the reformulation in Problem~\ref{prob:alternativetemporalcs}, it is immediate to observe that our \temporalcs problem is an instance of the well-established \textsc{Sequence Segmentation} problem, which asks for partitioning a sequence of  numerical values into $b$ segments so as to minimize the sum of the penalties (according to some penalty function) on each identified segment~\cite{bellman61approximation}:

\begin{problem}[\seqseg~\cite{bellman61approximation}]\label{prob:sequencesegmentation}
Given a sequence $X = (x_0, x_1, \ldots, x_{max})$ of numerical values, and a function $p: \{Y\}_{Y \sqsubseteq X} \rightarrow \mathbb{R}$ that assigns a penalty score to every subsequence $Y$ of $X$,
partition $X$ into a set $\{X_i\}_{i=1}^b$ of $b$ subsequences such that $\sum_{i=1}^b p(X_i)$ is minimized.
\end{problem}

\begin{fact}\label{fact2}
\temporalcs (Problem~\ref{prob:temporalcs}) on input $\langle G = (V,T,\tau), Q, h \rangle$ is an instance of \seqseg (Problem~\ref{prob:sequencesegmentation}) with $X = T$, \ $b = h$, \  and $\forall \Delta \sqsubseteq T : p(\Delta) = -v^*_{Q,\Delta}$.
\end{fact}

In the following two subsections we show how to exploit the result in Fact~\ref{fact2} (and a further important finding about maximal \spancore{s}) to design efficient algorithms for our  \temporalcs problem.


\begin{algorithm}[t]
\DontPrintSemicolon
\newcommand\commentfont[1]{\small\ttfamily{#1}}
\SetCommentSty{commentfont}

\KwIn{A temporal graph $G = (V,E,T)$, a set $Q \subseteq V$ of query vertices, an integer $h \in \mathbb{N}^+$.}
\KwOut{A set $\{\langle S_i, \Delta_i \rangle\}_{i=1}^h$, where $Q \subseteq S_i \subseteq V$, $\forall 1 \leq i \leq h$, and $\{\Delta_i\}_{i=1}^h$ is a partition of $T$.}

\vspace{1mm}
\tcc{Initialization}
Compute $v^*_{Q,\Delta}$ and $C^*_{Q,\Delta}$, $\forall \Delta \sqsubseteq T$, via $Q$-constrained \spancore decomposition\;
$\mathbf{P} \gets$ an empty  $(|T| \times h)$-dimensional matrix \tcp*{Penalty matrix}
$\mathbf{R} \gets$ an empty $(|T| \times h)$-dimensional matrix \tcp*{Reconstruction matrix}
\ForAll{$t \in T$}
{
	$\mathbf{P}[t,0] \gets -v^*_{Q, [0,t]}$\;
	$\mathbf{R}[t,0] \gets 0$\;
}

\tcc{Dynamic-programming step}
\ForAll{$t \in T$}
{
	\ForAll{$i \in [1, h)$}
	{
		$\mathbf{P}[t,i] \gets \min_{\ell \in [0, t]} \mathbf{P}[\ell,i-1] - v^*_{Q,[\ell+1,t]}$\;
		$\mathbf{R}[t,i] \gets \operatorname{argmin}_{\ell \in [0, t]} \mathbf{P}[\ell,i-1] - v^*_{Q,[\ell+1,t]}$\;
	}
}

\tcc{Reconstruction of the solution}
$ub \gets t_{max}$\;
\ForAll{$i \in (h, 0]$}
{
	$lb \gets \mathbf{R}[ub,i]$\;
	$\Delta_i \gets [lb, ub]$\;
	$ub \gets lb - 1$\;
}
\ForAll{$i \in (h, 0]$}
{
	$S_i \gets C^*_{Q, \Delta_i}$
}

\caption{\cs}\label{alg:temporalcsnaive}
\end{algorithm}

\subsection{A basic algorithm (based on all span-cores)}\label{sec:alg_naive_cs}

\seqseg can be solved in \mbox{$\bigO(|X|^2 \times h + \tau_p)$} time via dynamic programming~\cite{bellman61approximation}, where $\tau_p$ is the overall time spent for computing the penalty score of all subsequences of the input sequence $X$ (according to the given penalty function $p$).
Thanks to the connection shown in Fact~\ref{fact2}, the dynamic-programming algorithm for \seqseg can be easily adapted to solve \temporalcs as well.
The pseudocode of this algorithm -- termed \cs\ -- is reported as Algorithm~\ref{alg:temporalcsnaive}, and described next.

\begin{algorithm}[t]
\DontPrintSemicolon
\newcommand\commentfont[1]{\small\ttfamily{#1}}
\SetCommentSty{commentfont}

\KwIn{A temporal graph $G = (V,E,T)$, a set $Q \subseteq V$ of query vertices, an integer $h \in \mathbb{N}^+$.}
\KwOut{A set $\{\langle S_i, \Delta_i \rangle\}_{i=1}^h$, where $Q \subseteq S_i \subseteq V$, $\forall 1 \leq i \leq h$, and $\{\Delta_i\}_{i=1}^h$ is a partition of $T$.}

\tcc{Identification of $T^*$}

Compute the set $\mathbf{C}_M(Q)$ of $Q$-constrained maximal \spancores of $G$ \label{line:maximaltemporalcs:cd}\;
$\mathbf{D} \leftarrow \{\Delta \sqsubseteq T \mid C_{k,\Delta} \in \mathbf{C}_M(Q) \}$\;
$T_{\mathbf{D}} \leftarrow \bigcup_{\Delta \in \mathbf{D}} \Delta$; \ $T_{\mathbf{D}}^+ \leftarrow \{\min\{t_e \!+\! 1, t_{max}\} \mid [t_s,t_e] \in \mathbf{D}\}$; \ $T_{\mathbf{D}}^- \leftarrow \{\max\{t_s \!-\! 1, 0\} \mid [t_s,t_e] \in \mathbf{D}\}$\;
$T_{sup} \leftarrow \{t_i \in T \setminus (T_{\mathbf{D}} \cup T_{\mathbf{D}}^- \cup T_{\mathbf{D}}^+ \cup \{t_{max}\}) \mid i \in [1, h + 1 - |T_{\mathbf{D}} \cup T_{\mathbf{D}}^- \cup T_{\mathbf{D}}^+ \cup \{t_{max}\}|]\}$\;
$T^* \leftarrow T_{\mathbf{D}} \ \cup \ T_{\mathbf{D}}^+ \ \cup \ T_{\mathbf{D}}^-  \ \cup \ \{t_{max} \} \ \cup \ T_{sup}$\;

\tcc{Initialization}
Compute $v^*_{Q,\Delta}$, $\forall \Delta \sqsubseteq T$\label{line:maximaltemporalcs:m}\;
$\mathbf{M} \leftarrow $ mapping function $[0, |T^*|) \rightarrow T^*$\;
$\mathbf{P} \gets$ an empty  $(|T^*| \times h)$-dimensional matrix \tcp*{Penalty matrix}
$\mathbf{R} \gets$ an empty $(|T^*| \times h)$-dimensional matrix \tcp*{Reconstruction matrix}

\ForAll{$r \in [0, |T^*|)$}
{
	$\mathbf{P}[r,0] \gets -v^*_{Q, [0,\mathbf{M}[r]]}$\;
	$\mathbf{R}[r,0] \gets 0$\;
}

\tcc{Dynamic-programming step}
\ForAll{$r \in [0, |T^*|)$}
{
	\ForAll{$i \in [1, h)$}
	{
		$\mathbf{P}[r,i] \gets \min_{\ell \in [0,r]} \mathbf{P}[\ell,i-1] - v^*_{Q, [\mathbf{M}[\ell+1],\mathbf{M}[r]]}$\;
		$\mathbf{R}[r,i] \gets \operatorname{argmin}_{\ell \in [0,r]} \mathbf{P}[\ell,i-1] - v^*_{Q, [\mathbf{M}[\ell+1],\mathbf{M}[r]]}$\;
	}
}

\tcc{Reconstruction of the solution}
$ub \gets |T^*| - 1$\;
\ForAll{$i \in (h, 0]$}
{
	$lb \gets \mathbf{R}[ub,i]$\;
	$\Delta_i \gets [\mathbf{M}[lb], \mathbf{M}[ub]]$\;
	$ub \gets lb - 1$\;
}
\ForAll{$i \in (h, 0]$}
{
	$S_i \gets C^*_{Q, \Delta_i}$
}

\caption{\mcs}\label{alg:maximaltemporalcs}
\end{algorithm}

The \cs algorithm makes use of two $(|T| \times h)$-dimensional matrices, i.e., $\mathbf{P}$ and $\mathbf{R}$.
Matrix $\mathbf{P}$ represents the \emph{penalty matrix}.
It contains, $\forall t \in T$, $\forall i \in [0, h)$, the minimum cost of segmenting the sequence corresponding to the first $t$ timestamps of $T$ into $i+1$ segments.
As a result, $\mathbf{P}[t_{max}, h-1]$ contains the objective-function value of the ultimate optimal solution to Problem~\ref{prob:alternativetemporalcs}.
Matrix $\mathbf{R}$ is the \emph{reconstruction matrix}.
It provides information about the optimal segmentation, and is used at the end of the algorithm to reconstruct the output $\{\Delta_i\}_{i=1}^h$.
Note that the algorithm does not explicitly compute the $S_i$ subgraphs  corresponding to the optimal $\Delta_i$ intervals.
In fact, as discussed above, each $S_i$ can be easily retrieved at the end of the algorithm, by simply setting it equal to the corresponding $(Q, \Delta_i)$-highest-order-\spancore $C^*_{Q,\Delta_i}$.
%
According to Fact~\ref{fact2}, the penalty score of an interval $\Delta \sqsubseteq T$ corresponds to $-v^*_{Q,\Delta}$, i.e., the negative of the order of the $(Q, \Delta)$-highest-order-\spancore $C^*_{Q,\Delta}$.
All individual $v^*_{Q,\Delta}$ values, for all $\Delta \sqsubseteq T$, are efficiently computed altogether, at the beginning of the algorithm, via a ``$Q$-constrained'' variant of span-core decomposition (an alternative, but much less efficient strategy consists in computing every single $v^*_{Q,\Delta}$ from scratch, on the fly).
Specifically,  a simple (yet more efficient) variant of the span-core decomposition algorithm (Algorithm~\ref{alg:decomposition}) is employed for this purpose, which outputs only those \spancores containing all the vertices in $Q$.
This is easily achievable by stopping the \textsf{core-decomposition} subroutine, for every interval $\Delta \sqsubseteq T$, as soon as a core not containing all query vertices in $Q$ has been encountered.

The time complexity of Algorithm~\ref{alg:temporalcsnaive} is $\bigO(|T|^2 \times h + \tau_{sc})$, where $\tau_{sc}$ is the time spent for computing the $Q$-constrained \spancore decomposition of the input graph $G$.

\subsection{A more efficient algorithm (based on maximal span-cores)}\label{sec:alg_mcs}

A more efficient algorithm can be designed by noticing that, actually, one does not need to consider all timestamps in $T$ in the dynamic-programming step.
Rather, focusing on a subset $T^* \subseteq T$ -- which is properly defined based on the maximal \spancores of the input graph, see next -- allows for significantly reducing the dimensionality of the penalty matrix $P$ and the reconstruction matrix $R$, hence the overall time complexity of the algorithm, without affecting optimality of the output solution.
The following fact provides the theoretical basis for defining such a reduced temporal domain $T^*$.


\begin{fact}\label{fact3}
Given a temporal graph $G = (V,T,\tau)$ and a set $Q \subseteq V$ of query vertices, let $\mathbf{C}_M(Q)$ be the set of all $Q$-constrained maximal \spancores of $G$. For a temporal interval $\Delta \sqsubseteq T$, it holds that $\Qdeltascore = \max \{0, \max \{ k \mid C_{k, \Delta'} \in \mathbf{C}_M(Q), \Delta \sqsubseteq \Delta'\}\}$.
\end{fact}

Fact~\ref{fact3}
states that the penalty score $\Qdeltascore$ of an interval $\Delta$ corresponds to the maximum among the orders of the $Q$-constrained maximal \spancores whose span includes $\Delta$, if some exist.
If an interval $\Delta$ is not a subset of any span of a $Q$-constrained maximal \spancore, then $\Qdeltascore = 0$.
In that case, therefore, $\Delta$ can be safely discarded, as it cannot be part of the optimal solution of the given \temporalcs problem instance (unless it is needed to fill possible ``holes'', see below).
The ultimate consequence of this finding is that  the aforementioned reduced temporal domain $T^*$ is identified by the timestamps covered by the spans of the maximal \spancores, along with auxiliary timestamps, which are needed to ensure a smooth execution of the dynamic-programming step, as well as a correct handling of some extreme cases.
Specifically, let $\mathbf{D} = \{\Delta \sqsubseteq T  \mid C_{k,\Delta} \in C_M(Q)\}$ be the set of the spans of the $Q$-constrained maximal \spancores of the input graph,
and $T_{\mathbf{D}} = \bigcup_{\Delta \in \mathbf{D}} \Delta$ be the set of timestamps that are part of a span of a $Q$-constrained maximal \spancore.
The first two sets of auxiliary timestamps correspond to the timestamps that immediately precede and succeed the intervals in $\mathbf{D}$, i.e., the sets $T_{\mathbf{D}}^+ = \{\min\{t_e + 1, t_{max}\} \mid [t_s,t_e] \in \mathbf{D}\}$ and $T_{\mathbf{D}}^- = \{\max\{t_s - 1, 0\} \mid [t_s,t_e] \in \mathbf{D}\}$, respectively.
The timestamps in $T_{\mathbf{D}}^+$ and $T_{\mathbf{D}}^-$ (along with the last timestamp $t_{max}$ of the input temporal domain $T$) are needed to allow the dynamic-programming step to identify a solution that actually covers the whole temporal domain $T$ (as per Condition ($ii$) of Problem~\ref{prob:temporalcs}).
In particular, such timestamps may be interpreted as a trick to give the dynamic-programming step the flexibility to select ``holes'' (i.e., time intervals in-between two consecutive but not necessarily contiguous timestamps in $T_{\mathbf{D}}$).
Moreover, we define $T_{sup}$ as the set of the first $h +1 - |T_{\mathbf{D}} \cup T_{\mathbf{D}}^- \cup T_{\mathbf{D}}^+ \cup \{t_{max}\}|$ timestamps of $T$ not contained in $T_{\mathbf{D}} \cup T_{\mathbf{D}}^- \cup T_{\mathbf{D}}^+ \cup \{t_{max}\}$, i.e., $T_{sup} = \{t_i \in T \setminus (T_{\mathbf{D}} \cup T_{\mathbf{D}}^- \cup T_{\mathbf{D}}^+ \cup \{t_{max}\}) \mid i \in [1, h +1 - |T_{\mathbf{D}} \cup T_{\mathbf{D}}^- \cup T_{\mathbf{D}}^+ \cup \{t_{max}\}|]\}$.
The timestamps in $T_{sup}$ are further auxiliary timestamps that are needed to return a correct $h$-sized solution when the timestamps in $T_{\mathbf{D}} \cup T_{\mathbf{D}}^- \cup T_{\mathbf{D}}^+ \cup \{t_{max}\}$ are less than $h+1$ (the minimum number of timestamps required in $T^*$ to have a solution of size $h$).
Note that $T_{sup}$ is nonempty only if $|T_{\mathbf{D}} \cup T_{\mathbf{D}}^- \cup T_{\mathbf{D}}^+ \cup \{t_{max}\}| < h+1$.
Ultimately, $T^*$ is defined as
$$
T^* = T_{\mathbf{D}} \ \cup \ T_{\mathbf{D}}^+ \ \cup \ T_{\mathbf{D}}^-  \ \cup \ \{t_{max} \} \ \cup \ T_{sup}.
$$

The proposed more efficient method for \temporalcs, termed \mcs,  is summarized in Algorithm~\ref{alg:maximaltemporalcs} and described next.
The first five lines of the algorithm are devoted to the identification of $T^*$.
As said above, matrices $\mathbf{P}$ and $\mathbf{R}$ have here reduced dimensionality with respect to Algorithm~\ref{alg:temporalcsnaive}: they are $(|T^*| \times h)$-dimensional matrices, where $|T^*| \leq |T|$.
A mapping function $\mathbf{M}$ is used to assign an index within $[0,|T^*|)$ to every timestamp in $|T^*|$ (Line~\ref{line:maximaltemporalcs:m}). Such a mapping is needed to have every timestamp in $|T^*|$ logically assigned to a row of matrices $\mathbf{P}$ and $\mathbf{R}$.
The rest of the algorithm resembles Algorithm~\ref{alg:temporalcsnaive}, except for the fact that $\mathbf{M}$ is used every time that a row index has to be mapped to its corresponding timestamp (e.g., during the reconstruction of the solution).

An important point to clarify is that, during the execution of the \mcs algorithm, we might need the penalty score $\Qdeltascore$ of intervals $\Delta \sqsubseteq T$ corresponding to \emph{non-maximal} ($Q$-constrained) \spancores.
Therefore, the algorithm needs the $\Qdeltascore$ score of \emph{all} intervals $\Delta \sqsubseteq T$.
To compute these $\Qdeltascore$ scores (and, related to this, the set $\mathbf{C}_M(Q)$ of $Q$-constrained maximal \spancores, at Line~\ref{line:maximaltemporalcs:cd}), there are two main options.
The first one consists in computing the whole $Q$-constrained \spancore decomposition (as done in Algorithm~\ref{alg:temporalcsnaive}), keep the $\Qdeltascore$ scores of all such cores, and eventually compute $\mathbf{C}_M(Q)$ by simply filtering out non-maximal \spancores.
The second option corresponds instead to compute $\mathbf{C}_M(Q)$ directly, without passing through the whole  $Q$-constrained \spancore decomposition.
This may be carried out by running a simple variant of the algorithm for computing maximal \spancores (Algorithm~\ref{alg:imcores}), where containment of query vertices is added as a further constraint.
The computation of all the $\Qdeltascore$ scores comes for free during the execution of this algorithm for $Q$-constrained maximal \spancores: these scores can therefore be retained by adding a few straightforward (constant-time) instructions to that algorithm.
In our implementation we stick to the latter, as the \innermosts\ algorithm has been experimentally recognized as faster than the na\"ive filtering approach in all tested datasets.

The time complexity of the proposed \mcs algorithm is $\bigO(|T^*|^2 \times h + \tau_{msc})$, with $\tau_{msc}$ being the time spent in computing the $Q$-constrained maximal \spancores and the penalty scores $\Qdeltascore$.
As in practice (attested by our experiments) $|T^*| \ll |T|$,
the proposed \mcs algorithm is expected to be much more efficient than its na\"ive counterpart, i.e., Algorithm~\ref{alg:temporalcsnaive}.

\subsection{Minimum community search}\label{sec:min_cs}
An instance of \temporalcs may admit several optimal solutions which might differ either in terms of output intervals $\{\Delta_i\}_{i=1}^h$, or in terms of subgraphs assigned to the various identified intervals.
More precisely, the latter
refers to the fact that two optimal solutions might find the same segmentation $\{\Delta_i\}_{i=1}^h$ of the input temporal domain, but select different subgraphs $S_i$ for any interval $\Delta_i$.
Therefore, if the communities $S_i$ are not chosen carefully, they may result to be excessively large, not really cohesive, and containing redundant/outlying vertices.
This is a well-recognized issue of minimum-degree-based community search~\cite{Sozio}.
At the same time, large communities might include more cohesive and denser subgraphs that still exhibit optimality.
Motivated by this, in this subsection we devise a method to refine the communities originally found by our algorithms for \temporalcs, specifically attempting to minimize their size while preserving optimality.
The main idea behind our refinement method is based on the following result:

\begin{myproposition}[Community containment]\label{prp:cs_conteinment}
Given a temporal graph $G = (V,T,\tau)$, a set $Q \subseteq V$ of query vertices, and a positive integer $h \in \mathbb{N}^+$, let $\{\langle S_i, \Delta_i \rangle \}_{i=1}^h$ be a solution to Problem~\ref{prob:temporalcs} on input $\langle G, Q,h\rangle$ with $S_i$ corresponding to the $(Q, \Delta_i)$-highest-order-\spancore of $G$, $\forall i \in [1, h]$.
For every other solution $\{\langle S'_i, \Delta_i \rangle \}_{i=1}^h$ (referring to the same segmentation $\{\Delta_i\}_{i=1}^h$) to Problem~\ref{prob:temporalcs} on input $\langle G, Q,h\rangle$ it holds that $S_i' \subseteq S_i$, $\forall i \in [1, h]$.
\end{myproposition}
\begin{proof}
Let $k_i$ be the minimum degree of $S_i$, i.e., $k_i = v^*_{Q,\Delta_i}$ is the order of the $(Q, \Delta_i)$-highest-order-\spancore.
Assume that there exists a solution $S'_i$ to Problem~\ref{prob:temporalcs_subproblem} that is not contained in $S_i$.
This implies that $(i)$ the minimum degree of a vertex of $S'_i$ in $\Delta_i$ is $k_i$,
and ($ii$) the minimum degree of a vertex of $S_i \cup S'_i$ in $\Delta_i$ is $k_i$ as well.
This violates the maximality condition of the definition of \spancore, since, by hypothesis, $S_i$ corresponds to the $(Q, \Delta_i)$-highest-order-\spancore of $G$.
\end{proof}

The above proposition states that, given a solution $\{\langle S_i, \Delta_i \rangle \}_{i=1}^h$ to the \temporalcs problem where every $S_i$ corresponds to the $(Q, \Delta_i)$-highest-order-\spancore of the input graph, one can focus on the various $S_i$ solely to refine the output communities, as such $S_i$ are guaranteed to contain \emph{all} optimal solutions of the underlying problem instance (while keeping the segmentation $\{\Delta_i\}_{i=1}^h$ fixed).
Within this view, we formulate the following optimization problem (which is a variant of Problem~\ref{prob:temporalcs_subproblem}, with the additional constraint of requiring a smallest-sized solution):

\begin{problem}\label{prob:temporalcs_subproblem_min}
Given a temporal graph $G = (V,T,\tau)$, a set $Q \subseteq V$ of query vertices, and an interval $\Delta \sqsubseteq T$,
let $S^* \subseteq V$ be the subset of vertices containing all the solutions to Problem~\ref{prob:temporalcs_subproblem} on input  $\langle G, Q, \Delta \rangle$ (according to what stated in Proposition~\ref{prp:cs_conteinment}).
Find
$$
S^*_{min} =  { \operatorname{argmin}_{\{ S \mid Q \subseteq S \subseteq S^*, \min_{u \in S} \tdeg_{\Delta}(S,u) \geq \min_{u \in S^*} \tdeg_{\Delta}(S^*,u)\}} |S|}.
$$
\end{problem}

\begin{mytheorem}\label{th:nphardness}
Problem~\ref{prob:temporalcs_subproblem_min} in $\mathbf{NP}$-hard.
\end{mytheorem}
\begin{proof}
Consider (the optimization version of) the $\mathbf{NP}$-hard \textsc{mCST} problem introduced by Cui~\emph{et al.}~\cite{SozioLocalSIGMOD14}: given a graph $H = (V_H,E_H)$ and a query vertex $q \in V_H$, find a minimum-sized subgraph that contains $q$, is connected, and maximizes the minimum degree.
Given an instance $\langle H, q \rangle$ of the  \textsc{mCST} problem,  construct an instance $\langle G, Q, \Delta \rangle$ of  Problem~\ref{prob:temporalcs_subproblem_min} by defining $G$ as composed by a single temporal snapshot corresponding to graph $H$, $\Delta$ as a singleton interval composed of the single timestamp of $G$, and setting $Q = \{q\}$.
It is straightforward to see that solving Problem~\ref{prob:temporalcs_subproblem_min} on input $\langle G, Q, \Delta \rangle$ is equivalent to solving \textsc{mCST} on input $\langle H, q \rangle$, as the constraint about connectedness is automatically satisfied in Problem~\ref{prob:temporalcs_subproblem_min} for the special case of a single query vertex.
\end{proof}

\begin{algorithm}[t]
\DontPrintSemicolon
\newcommand\commentfont[1]{\small\ttfamily{#1}}
\SetCommentSty{commentfont}

\KwIn{A temporal graph $G = (V,E,T)$, a set $Q \subseteq V$ of query vertices, an interval $\Delta \sqsubseteq T$, a subset of vertices $S^* \subseteq V$ containing all the solutions to Problem~\ref{prob:temporalcs_subproblem} on input $\langle G, Q,\Delta \rangle$.}
\KwOut{A subset $S^*_{min}$ of vertices such that $Q \subseteq S^*_{min} \subseteq S^*$ and $\min_{u \in S^*_{min}} \tdeg_{\Delta}(S^*_{min},u) \geq \min_{u \in S^*} \tdeg_{\Delta}(S^*,u)$.}

$S^*_{min} \leftarrow \emptyset$; \ \ $P \leftarrow \emptyset$; \ \ $\mathcal{A} \leftarrow \emptyset$\;
add every $q \in Q$ to $P$ with priority $+\infty$\;
$k^* \leftarrow \min_{u \in S^*} \tdeg_{\Delta}(S^*,u)$; \ \ $k^*_{min} \leftarrow 0$\;

\While{$k^*_{min} < k^*$ or $Q \not\subseteq S^*_{min}$}
{\label{line:greedy:intwhile} 
	dequeue $u$ from $P$\;
	$S^*_{min} \leftarrow S^*_{min} \cup \{u\}$\;

	\ForAll{$v \in \neigh(S^*, u) \setminus S^*_{min} \setminus P$}
	{
		$\mathcal{A}[v] \leftarrow \score(v)$\; \label{line:greedy:prioiritystart} 
		add $v$ to $P$ with priority $\mathcal{A}[v]$ \label{line:greedy:prioirityend} 
	}

	\ForAll{$v \in \neigh(S^*_{min}, u)$}
	{

		\If{$\tdeg_{\Delta}(S^*_{min}, v) = k^*$}
		{
			\ForAll{$w \in \neigh(P, v)$}
			{
				$\mathcal{A}[w] \leftarrow \mathcal{A}[w] - 1$\;
			}
		}
	}
	$k^*_{min} \leftarrow \min_{v \in S^*_{min}} \tdeg_{\Delta}(S^*_{min},v)$\;
}
\caption{\greedy}\label{alg:greedy}
\end{algorithm}

As Problem~\ref{prob:temporalcs_subproblem_min} is $\mathbf{NP}$-hard, we devise a heuristic that is inspired to the greedy one proposed for the \textsc{Minimum Community Search} problem in~\cite{BarbieriBGG15}.
The proposed heuristic is outlined in Algorithm~\ref{alg:greedy} and described next.
In the pseudocode and in the following we denote as $k^*$ and $k^*_{min}$ the minimum degree of $S^*$ and $S^*_{min}$, respectively, and as $\neigh(S, u)$ the neighbors of a vertex $u \in V$ in the subgraph induced by $S \subseteq V$ and $\Delta \sqsubseteq T$.
Algorithm~\ref{alg:greedy} iteratively adds vertices to the solution $S^*_{min}$ according to a priority queue $P$.
Priorities of vertices in $P$  are defined based on a score that measures how promising a vertex is for making the current solution $S^*_{min}$  reach the optimal minimum degree.
Specifically, the score of a vertex $u \in S^*$ is defined as:
$$
\score(u) = \scorep(u) - \scorem(u),
$$
where
$$
\scorep(u) = |\{v \in \neigh(S^*_{min}, u) \mid \tdeg_{\Delta}(S^*_{min}, v) < k^*\}|;
$$
$$
\scorem(u) = \max\{0, k^* - \tdeg_{\Delta}(S^*_{min}, u)\}.
$$
$\scorep(u)$ is the gain effect of adding $u$ to $S^*_{min}$, while $\scorem(u)$ is the penalty effect.
In particular, $\scorep(u)$ counts the number of neighbors of $u$ in $S^*_{min}$ that would benefit from the inclusion of $u$ to $S^*_{min}$, i.e., that have degree less than $k^*$.
On the other hand, $\scorem(u)$ represents the number of neighbors of $u$ still required in $S^*_{min}$ so that $u$ has degree at least $k^*$.
The algorithm starts by adding the query vertices to the queue $P$ with priority $+\infty$, in order to ensure that they will be selected at the very beginning.
At each iteration of the main  cycle of the algorithm (starting at Line~\ref{line:greedy:intwhile}), the vertex $u$ exhibiting the highest priority is dequeued from $P$ and is added to the solution $S^*_{min}$.
As a consequence, a couple of updates are performed.
First, $u$'s neighbors not in the priority queue $P$ are added to it (Lines~\ref{line:greedy:prioiritystart}-\ref{line:greedy:prioirityend}).
Note that this is the only step of the algorithm where the score of a vertex is computed from scratch and stored in $\mathcal{A}$, a map that keeps  the scores of all vertices in $P$ up-to-date during the whole execution of the algorithm.
The second update consists in recomputing the score of every $v$'s neighbor $w$ in the queue, if a vertex $v \in S^*_{min}$ has reached the desired minimum degree $k^*$ after the addition of $u$.


\section{Experiments}
\label{sec:experiments}


\begin{table}[t!]
\centering
\caption{Temporal graphs used in the experiments.}
\label{tab:datasets}

\begin{tabular}{c|ccccccc}
\multicolumn{1}{c}{dataset} & \multicolumn{1}{c}{$|V|$} & \multicolumn{1}{c}{$|E|$} & \multicolumn{1}{c}{$|T|$} & \multicolumn{1}{c}{window size} & \multicolumn{1}{c}{domain} \\
\hline
\textsf{HighSchool} & $327$ & $47$k & $1212$ & $5$ mins & face-to-face \\
\textsf{PrimarySchool} & $242$ & $55$k & $390$ & $5$ mins & face-to-face \\
\textsf{HongKong} & $806$ & $2$M & $2976$ & $5$ mins & face-to-face \\
\textsf{ProsperLoans} & $89$k & $3$M & $307$ & $7$ days & economic \\
\textsf{Last.fm} & $992$ & $4$M & $77$ & $21$ days & co-listening \\
\textsf{WikiTalk} & $2$M & $10$M & $192$ & $28$ days & communication \\
\textsf{DBLP} & $1$M & $11$M & $80$ & $366$ days & co-authorship \\
\textsf{StackOverflow} & $2$M & $16$M & $51$ & $56$ days & question answering \\
\textsf{Wikipedia} & $343$k & $18$M & $101$ & $56$ days & co-editing \\
\textsf{Amazon} & $2$M & $22$M & $115$ & $28$ days & co-rating \\
\textsf{Epinions} & $120$k & $33$M & $25$ & $21$ days & co-rating \\
\hline
\end{tabular}
\end{table}
\begin{table}
\centering
\caption{Evaluation of the proposed algorithms: number of output \spancores, running time, memory, and number of processed vertices.}
\label{tab:evaluation}

\renewcommand{\arraystretch}{0.975}
\begin{tabular}{c|c|c|ccc}
\multicolumn{1}{c}{} & \multicolumn{1}{c}{} & \multicolumn{1}{c}{\# output} & \multicolumn{1}{c}{running} & \multicolumn{1}{c}{memory} & \multicolumn{1}{c}{\# processed} \\
\multicolumn{1}{c}{dataset} & \multicolumn{1}{c}{method} & \multicolumn{1}{c}{\spancores} & \multicolumn{1}{c}{time (s)} & \multicolumn{1}{c}{(GB)} & \multicolumn{1}{c}{vertices} \\
\hline
 \multirow{4}{*}{\textsf{HighSchool}} & \baseline & \multirow{2}{*}{$12\,320$} & $18$ & $0.1$ & $3$M \\
 & \cores & & $1$ & $0.1$ & $581$k \\ \cline{2-6}
 & \baselineinnermosts & \multirow{2}{*}{$450$} & $1$ & $0.1$ & $581$k \\
 & \innermosts & & $0.3$ & $0.1$ & $181$k \\
\hline
 \multirow{4}{*}{\textsf{PrimarySchool}} & \baseline & \multirow{2}{*}{$4\,703$} & $4$ & $0.1$ & $818$k \\
 & \cores & & $0.6$ & $0.1$ & $174$k \\ \cline{2-6}
 & \baselineinnermosts &  \multirow{2}{*}{$409$} & $0.6$ & $0.1$ & $174$k \\
 & \innermosts & & $0.1$ & $0.1$ & $63$k \\
\hline
 \multirow{4}{*}{\textsf{HongKong}} & \baseline & \multirow{2}{*}{$2\,367\,743$} & $85\,180$ & $1$ & $819$M \\
 & \cores & & $18\,389$ & $0.8$ & $216$M \\ \cline{2-6}
 & \baselineinnermosts &  \multirow{2}{*}{$1\,807$} & $18\,641$ & $0.8$ & $216$M \\
 & \innermosts & & $339$ & $0.5$ & $212$M \\
\hline
 \multirow{4}{*}{\textsf{ProsperLoans}} & \baseline & \multirow{2}{*}{$4\,273$} & $101$ & $2$ & $55$M \\
 & \cores & & $46$ & $2$ & $27$M \\ \cline{2-6}
 & \baselineinnermosts &  \multirow{2}{*}{$293$} & $48$ & $2$ & $27$M \\
 & \innermosts & & $8$ & $2$ & $980$k \\
\hline
\multirow{4}{*}{\textsf{Last.fm}} & \baseline &  \multirow{2}{*}{$126\,819$} & $707$ & $0.5$ & $2$M \\
 & \cores & & $199$ & $0.5$ & $531$k \\ \cline{2-6}
 & \baselineinnermosts &  \multirow{2}{*}{$1\,670$} & $202$ & $0.5$ & $531$k \\
 & \innermosts & & $57$ & $0.5$ & $271$k \\
\hline
\multirow{4}{*}{\textsf{WikiTalk}} & \baseline &  \multirow{2}{*}{$19\,693$} & $322\,302$ & $36$ & $25$B \\
 & \cores & & $1\,084$ & $36$ & $555$M \\ \cline{2-6}
 & \baselineinnermosts &  \multirow{2}{*}{$632$} & $1\,194$ & $36$ & $555$M \\
 & \innermosts & & $126$ & $35$ & $2$M \\
\hline
\multirow{4}{*}{\textsf{DBLP}} & \baseline &  \multirow{2}{*}{$6\,135$} & $10\,506$ & $11$ & $1$B \\
 & \cores & & $278$ & $11$ & $150$M \\ \cline{2-6}
 & \baselineinnermosts &  \multirow{2}{*}{$268$} & $292$ & $11$ & $150$M  \\
 & \innermosts & & $116$ & $11$ & $620$k \\
\hline
\multirow{4}{*}{\textsf{StackOverflow}} & \baseline &  \multirow{2}{*}{$1\,238$} & $5\,360$ & $10$ & $1$B \\
 & \cores & & $245$ & $10$ & $127$M \\ \cline{2-6}
 & \baselineinnermosts &  \multirow{2}{*}{$129$} & $245$ & $10$ & $127$M \\
 & \innermosts & & $128$ & $10$ & $3$M \\
\hline
\multirow{4}{*}{\textsf{Wikipedia}} & \baseline &  \multirow{2}{*}{$125\,191$} & $17\,155$ & $4$ & $1$B \\
 & \cores & & $522$ & $4$ & $35$M \\ \cline{2-6}
 & \baselineinnermosts &  \multirow{2}{*}{$2\,147$} & $537$ & $4$ & $35$M \\
 & \innermosts & & $201$ & $4$ & $320$k \\
\hline
\multirow{4}{*}{\textsf{Amazon}} & \baseline &  \multirow{2}{*}{$29\,318$} & $10\,415$ & $18$ & $2$B \\
 & \cores & & $409$ & $18$ & $247$M \\ \cline{2-6}
 & \baselineinnermosts &  \multirow{2}{*}{$303$} & $580$ & $18$ & $247$M \\
 & \innermosts & & $123$ & $18$ & $688$k \\
\hline
\multirow{4}{*}{\textsf{Epinions}} & \baseline &  \multirow{2}{*}{$63\,111$} & $699$ & $4$ & $39$M \\
 & \cores & & $186$ & $4$ & $3$M \\ \cline{2-6}
 & \baselineinnermosts &  \multirow{2}{*}{$320$} & $201$ & $4$ & $3$M \\
 & \innermosts & & $154$ & $5$ & $129$k \\
\hline
\end{tabular}
\end{table}

In this section we present an experimental evaluation to empirically assess the performance of all the proposed methods.
Specifically, we focus on  whole \spancore decomposition (Section~\ref{sec:experiments_sc}), maximal \spancores  (Section~\ref{sec:experiments_msc}), characterization  of the extracted \spancores (Section~\ref{sec:experiments_characterization}), and temporal community search  (Section~\ref{sec:experiments_cs}).

\spara{Datasets.}
We use eleven real-world datasets recording timestamped interactions between entities.
For each dataset we select a window size to define a discrete time domain, composed of contiguous timestamps of the same duration, and build the corresponding temporal graph.
If multiple interactions occur between two entities during the same discrete timestamp, they are counted as one.
The characteristics of the resulting temporal graphs, along with the selected window sizes, are reported in Table~\ref{tab:datasets}.

The three smallest datasets were gathered by using wearable proximity sensors in schools, with a temporal resolution of $20$ seconds.
\textsf{PrimarySchool}\footnote{\href{http://www.sociopatterns.org}{sociopatterns.org}\label{foot:school}} contains the contact events between $242$ volunteers ($232$ children and $10$ teachers) in a primary school in Lyon, France, during two days~\cite{Stehle:2011}.
\textsf{HighSchool}$^{\ref{foot:school}}$ describes the close-range proximity interactions between students and teachers ($327$ individuals overall) of nine classes during five days in a high school in Marseilles, France~\cite{Fournet:PLOS2015}.
\textsf{HongKong} reports the same kind of interactions for a primary school in Hong Kong, whose population consists of $709$ children and $65$ teachers divided into thirty classes, for eleven consecutive days~\cite{sapienza2015detecting}.

\textsf{ProsperLoans}\footnote{\href{http://konect.cc}{konect.cc}\label{foot:konect}} represents the network of loans between the users of Prosper, a marketplace of loans between privates.
\textsf{Last.fm}$^{\ref{foot:konect}}$ records the co-listening activity of the Last.fm streaming platform: an edge exists between two users if they listened to songs of the same band within the same discrete timestamp.
\textsf{WikiTalk}$^{\ref{foot:konect}}$ is the communication network of the English Wikipedia.
\textsf{DBLP}$^{\ref{foot:konect}}$ is the co-authorship network of the authors of scientific papers from the DBLP computer science bibliography.
\textsf{StackOverflow}\footnote{\href{http://snap.stanford.edu}{snap.stanford.edu}}  includes the answer-to-question interactions on the stack exchange of the stackoverflow.com website.
\textsf{Wikipedia}$^{\ref{foot:konect}}$ connects users of the Italian Wikipedia that co-edited a page during the same discrete timestamp.
Finally, for both \textsf{Amazon}$^{\ref{foot:konect}}$ and \textsf{Epinions}$^{\ref{foot:konect}}$, vertices are users and edges represent the rating of at least one common item within the same discrete timestamp.

\spara{Implementation.}
All methods are implemented in Python (v. 2.7.16) and compiled by Cython.
All the experiments were run on a machine equipped with Intel Xeon CPU at 2.1GHz.
The experiments reported in Sections~\ref{sec:experiments_sc}~and~\ref{sec:experiments_msc} used 64GB RAM, while the ones in Section~\ref{sec:experiments_cs} used 32GB RAM.

\subsection{\Spancore decomposition}\label{sec:experiments_sc}
We compare the two methods to compute a complete decomposition
described in Section~\ref{sec:spancores}, i.e.,
the baseline \baseline\ and the proposed \cores, in terms of execution time, memory, and total number of vertices input to the $\textsf{core-decomposition}$ subroutine. We report these measures, together with the number of \spancores and maximal \spancores of each dataset, in Table~\ref{tab:evaluation}.

In terms of execution time, \cores\ considerably outperforms \baseline\ in all datasets, achieving a speed-up from $2.1$ up to two orders of magnitude.
The speed-up is explained by the number of vertices processed by the $\textsf{core-decomposition}$ subroutine, which is the most time-consuming step of the algorithms albeit linear in the size of the input subgraph.
The difference of this quantity between \cores\ and \baseline\ reaches over an order of magnitude in the \textsf{WikiTalk}, \textsf{Wikipedia}, and \textsf{Epinions} dataset, confirming the effectiveness of the ``horizontal containment'' relationships.
The memory required by the two procedures is comparable in all cases since the largest structures needed in memory are the temporal graph itself and the set $\coresset$ of all \spancores.

\subsection{Maximal \spancores}\label{sec:experiments_msc}
We compare our \innermosts\ algorithm to the na\"ive approach, described at the beginning of Section~\ref{sec:maximal_spancores}, based on running the \cores\ algorithm and filtering out the non-maximal \spancores, which we refer to as \baselineinnermosts.
The results are again reported in Table~\ref{tab:evaluation}.

\baselineinnermosts\ behaves very similarly to \cores: they only differ for the filtering mechanism which requires a few additional seconds in most cases.
\innermosts\ is much faster than \baselineinnermosts\ for all datasets, with a speed-up from $1.3$ for the \textsf{Epinions} dataset to one order of magnitude for the \textsf{HongKong} dataset.
Except for the school datasets and \textsf{Last.fm},
the difference in terms of number of processed vertices is between one and three orders of magnitude, attesting the advantages of the top-down strategy of \innermosts, which avoids the visit of portions of the \spancore search space and handles the overhead of reconstructing graphs, i.e., $(V_{lb}, E_{\Delta}[V_{lb}])$, efficiently.
Finally, the memory requirements of the two methods are comparable for all datasets.

\begin{figure}[t!]
\centerline{
\begin{tabular}{cc}
\includegraphics[width=0.35\columnwidth]{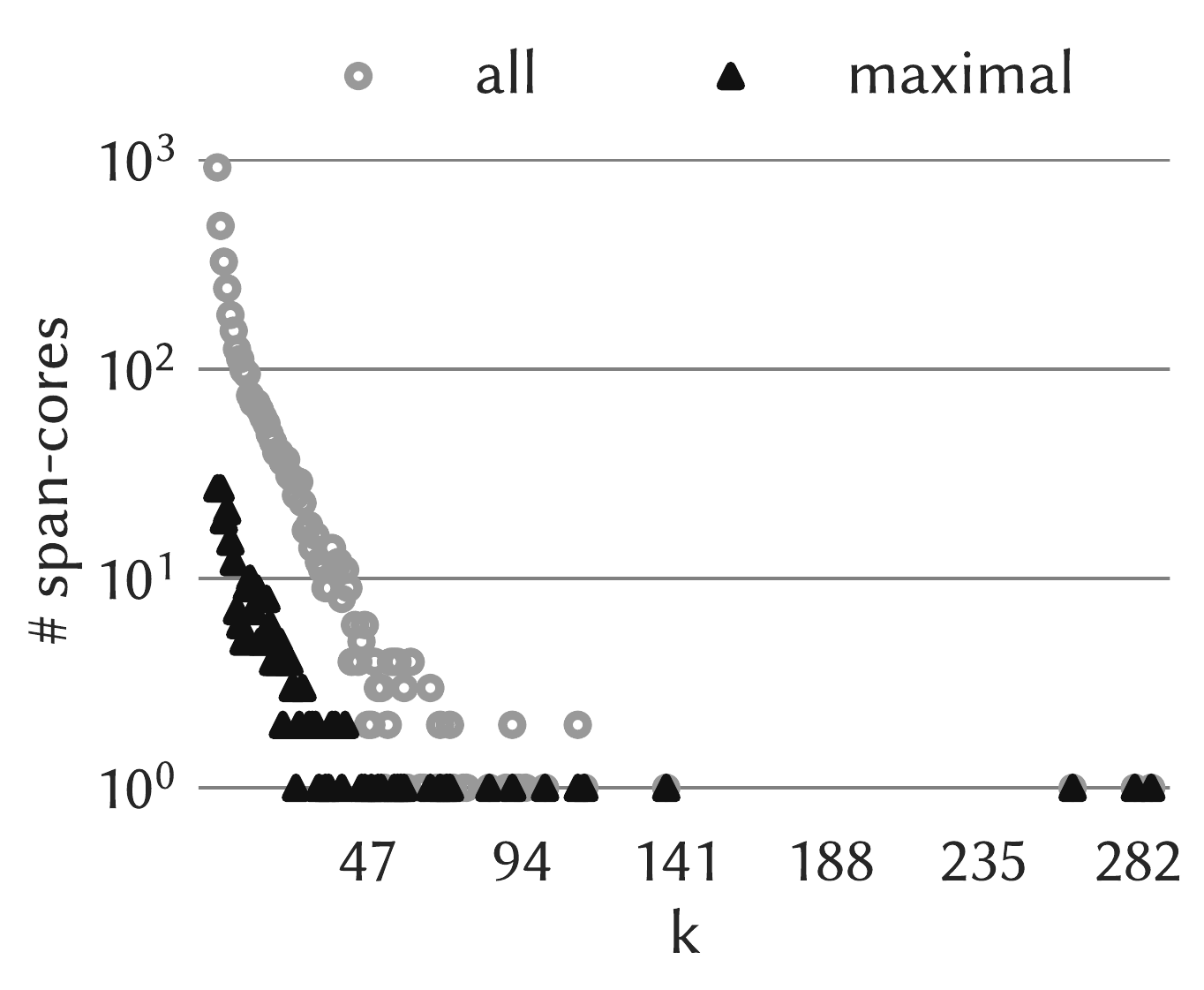} & \includegraphics[width=0.35\columnwidth]{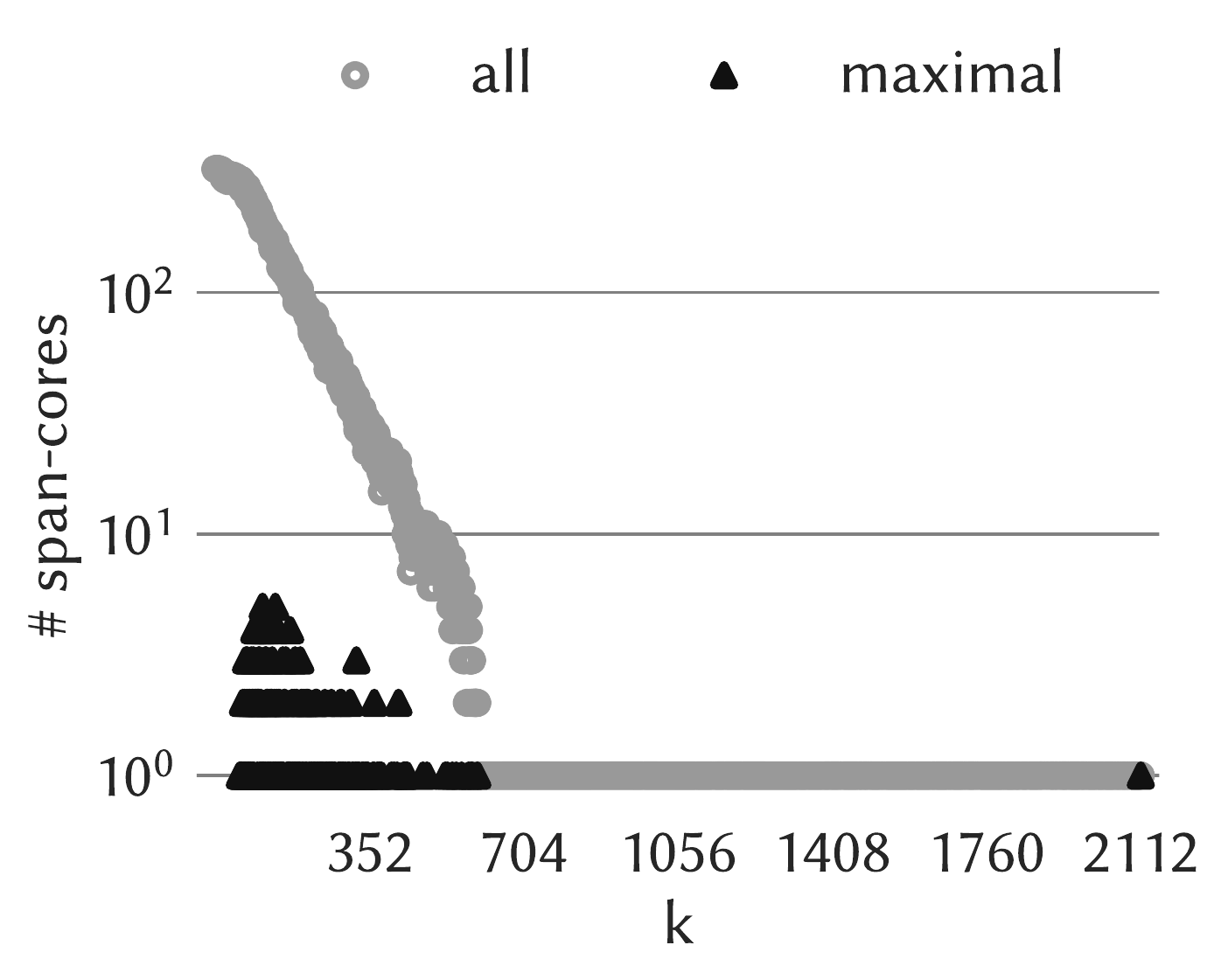}\\
\includegraphics[width=0.35\columnwidth]{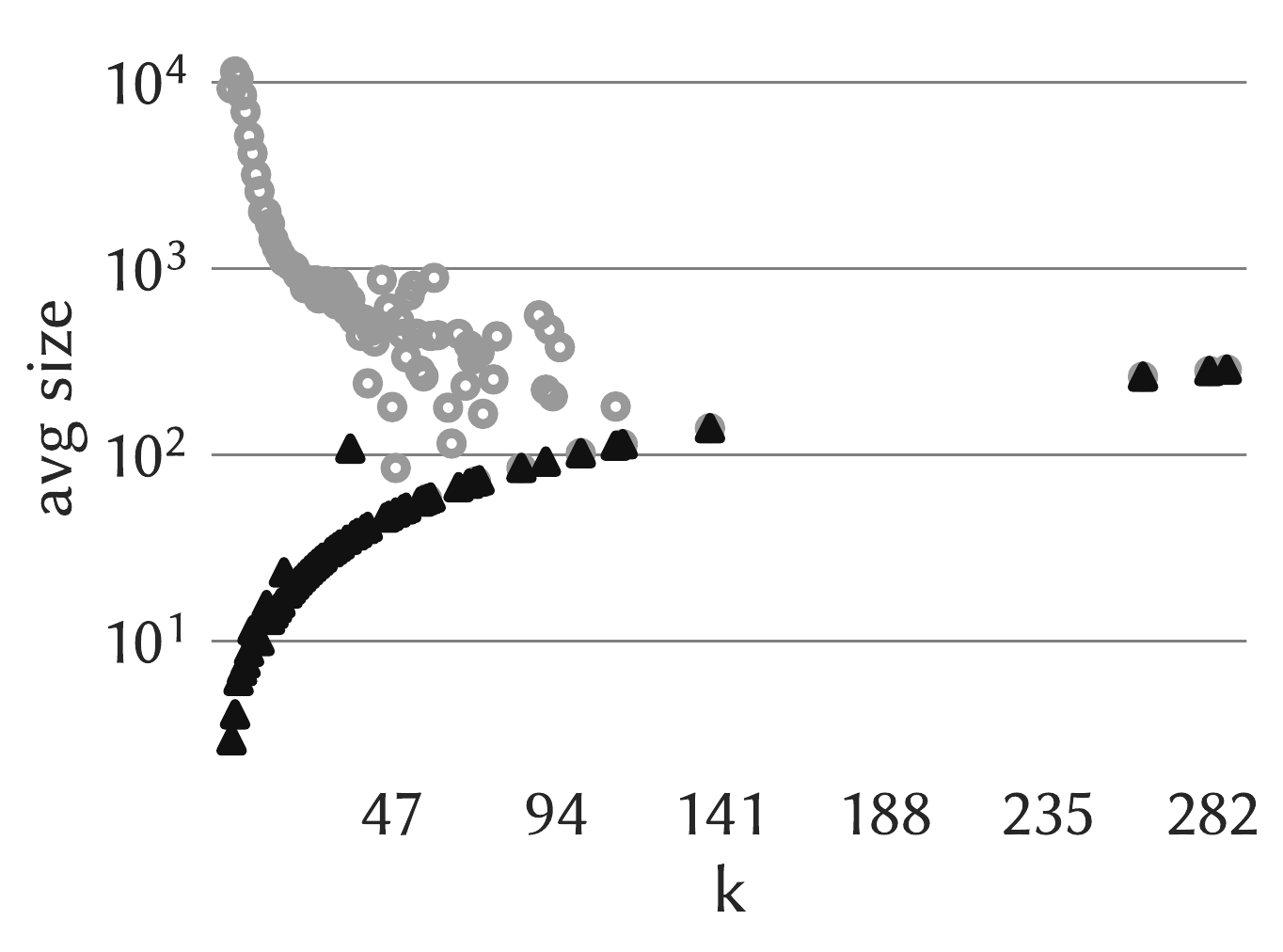} & \includegraphics[width=0.35\columnwidth]{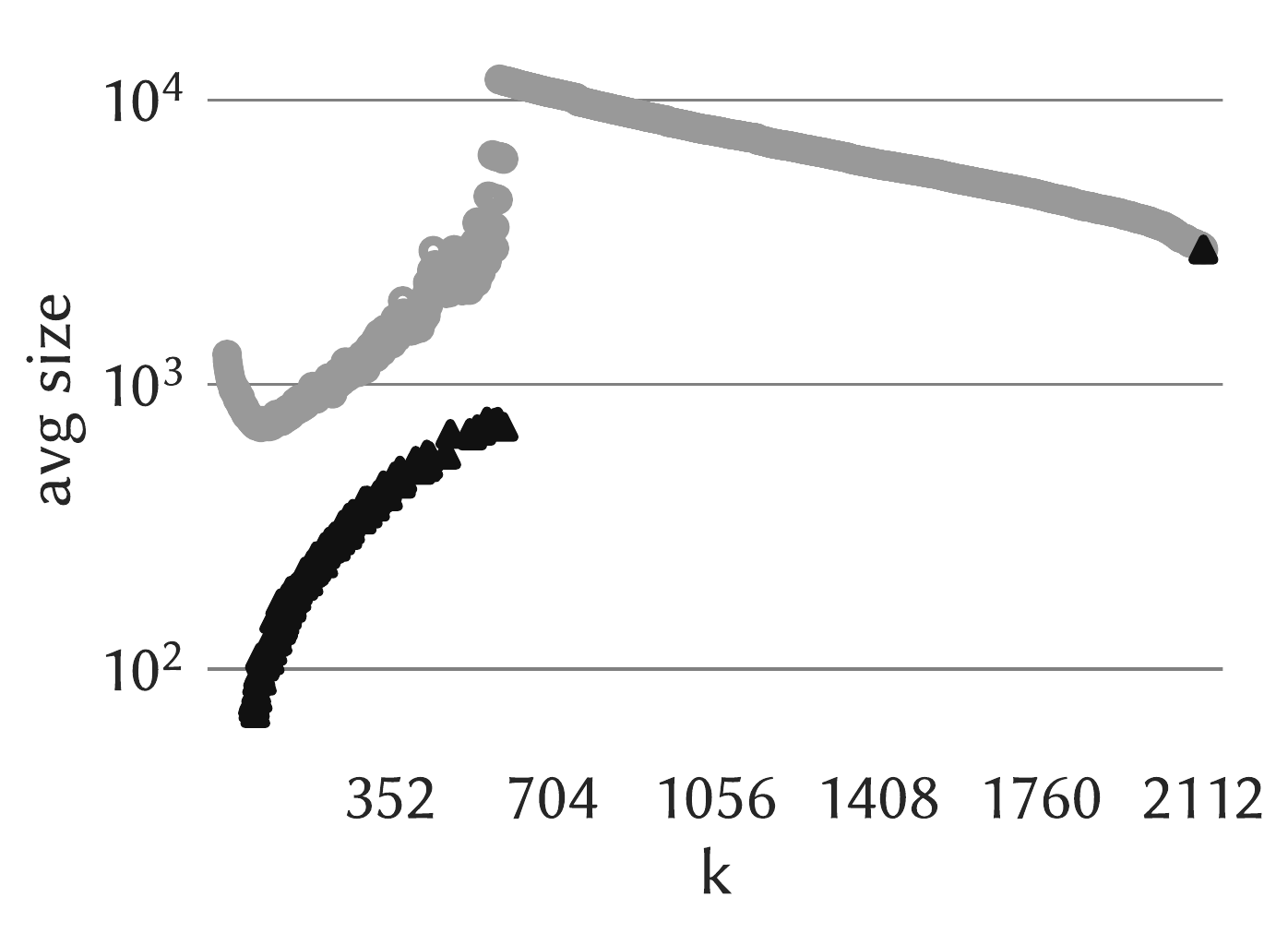} \vspace{-2mm}\\
\footnotesize{\textsf{DBLP}} & \footnotesize{\textsf{Epinions}}
\end{tabular}
}
\caption{\label{fig:coresvsmaximals_k} Top plots: number of \spancores and maximal \spancores as a function of the order $k$. Bottom plots: average size of all \spancores and maximal \spancores as a function of the order $k$.}
\end{figure}

\begin{figure}[t!]
\centerline{
\begin{tabular}{cc}
\includegraphics[width=0.35\columnwidth]{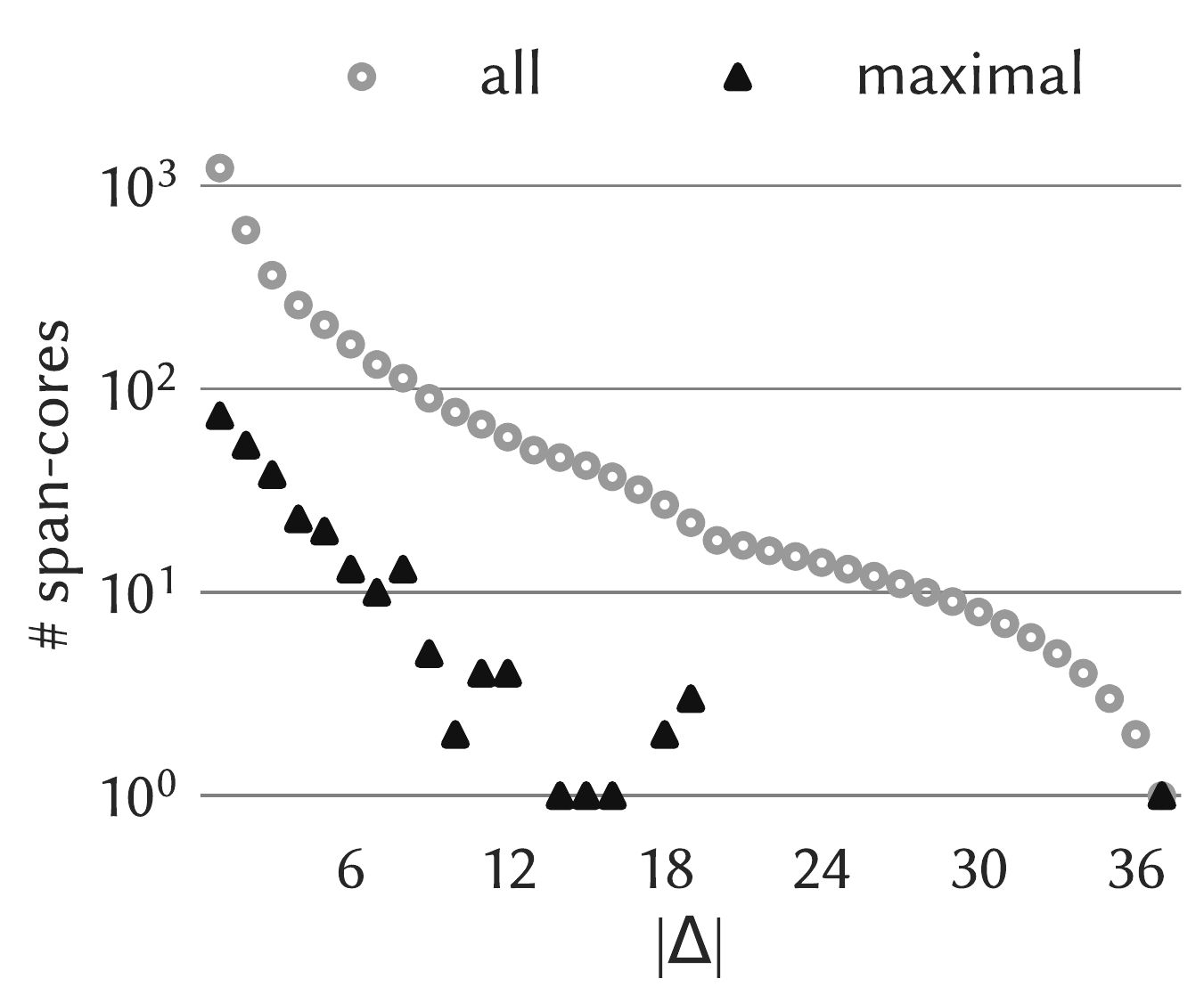} & \includegraphics[width=0.35\columnwidth]{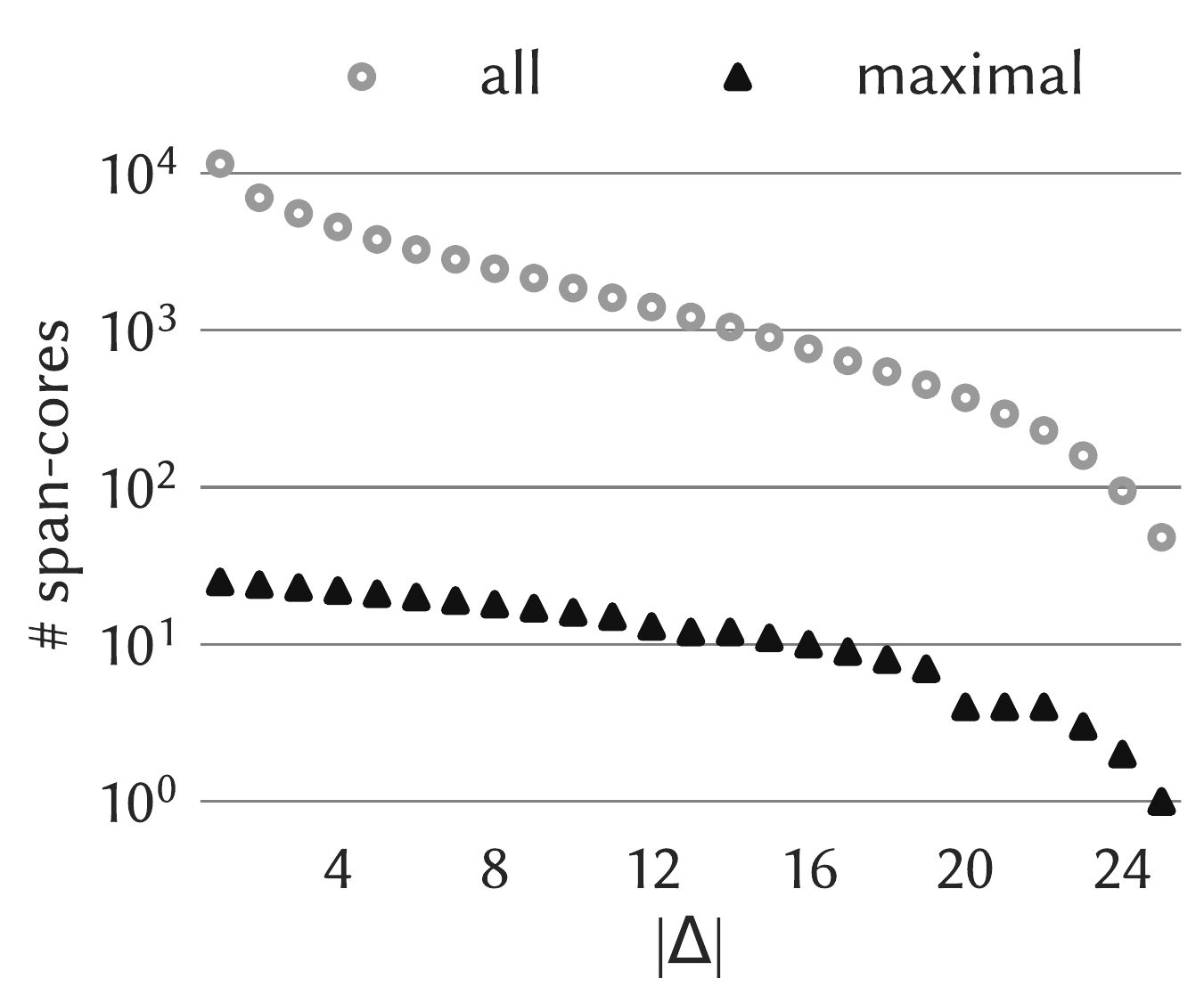}\\
\includegraphics[width=0.35\columnwidth]{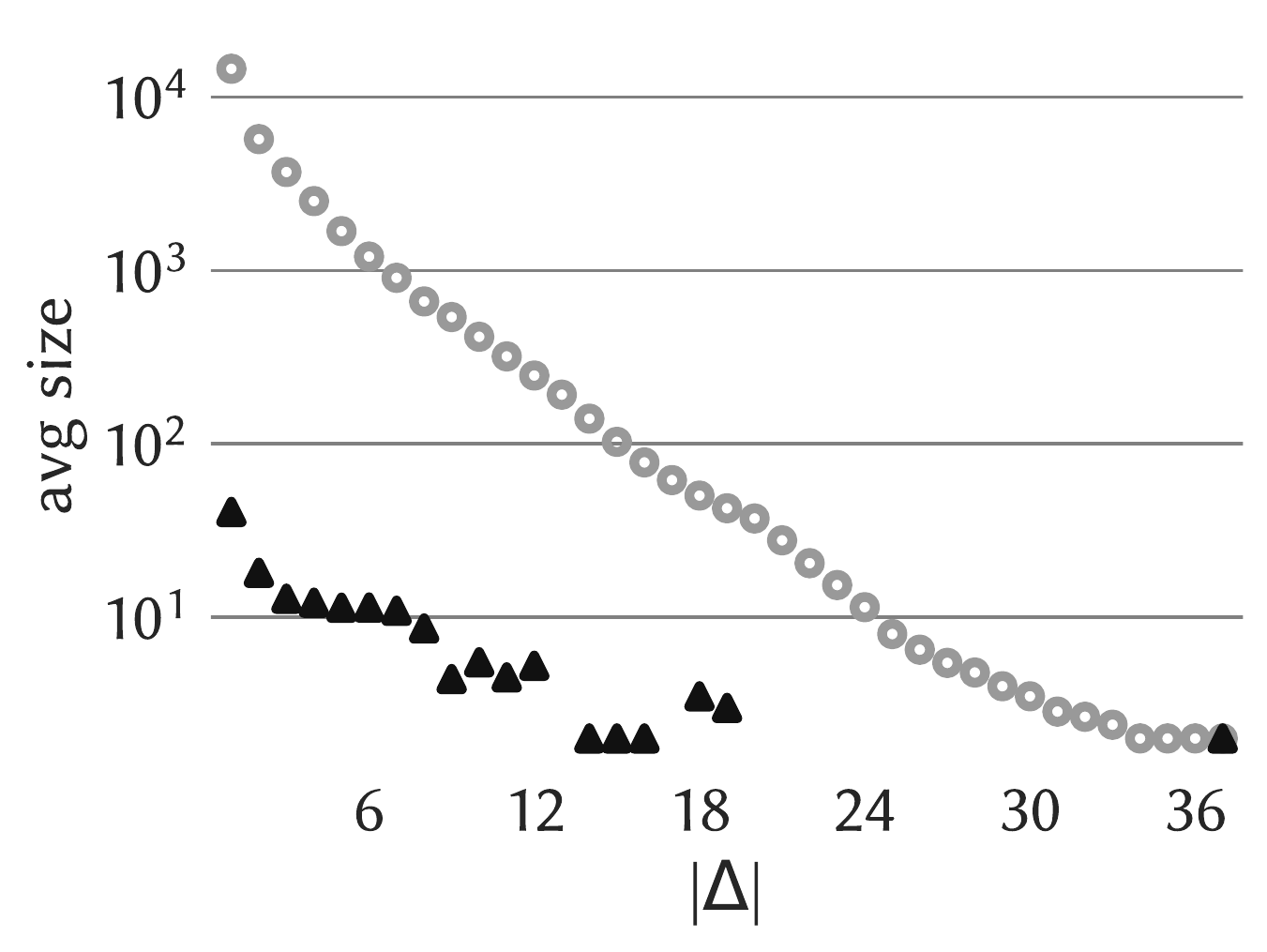} & \includegraphics[width=0.35\columnwidth]{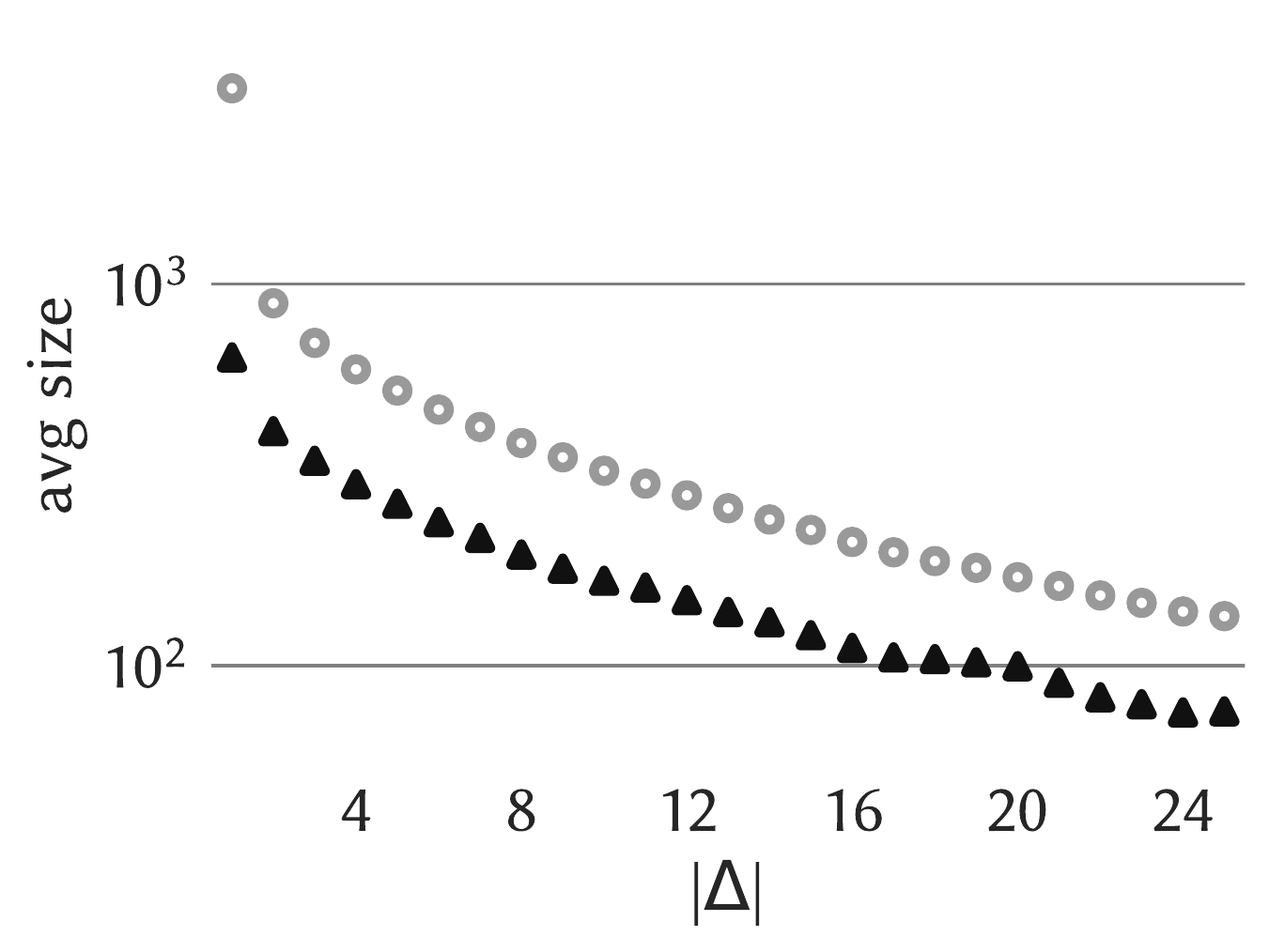} \vspace{-2mm}\\
\footnotesize{\textsf{DBLP}} & \footnotesize{\textsf{Epinions}}
\end{tabular}
}
\caption{\label{fig:coresvsmaximals_delta} Top plots: number of  \spancores and maximal \spancores as a function of the size of the temporal span $|\Delta|$. Bottom plots: average size of all \spancores and maximal \spancores  as a function of the size of the temporal span $|\Delta|$.}
\end{figure}

\subsection{Span-cores characterization}\label{sec:experiments_characterization}
%
We compare and characterize all \spancores against maximal \spancores.
At first, Table~\ref{tab:evaluation} shows that \spancores are at least one order of magnitude more numerous than maximal \spancores for all datasets, with the maximum difference of three orders of magnitude for the \textsf{HongKong} dataset.

In Figure~\ref{fig:coresvsmaximals_k} we show the number (top) and the average size (bottom) of \spancores and maximal \spancores as a function of the order $k$ for the \textsf{DBLP} and \textsf{Epinions} datasets.
For both datasets, the number of maximal \spancores is at least one order of magnitude lower than the total number of \spancores up to a quarter of the $k$ domain, where the \spancores are more numerous.
Instead, in the rest of the domain, they mostly coincide due to the maximality condition over $|\Delta|$.
The average size is also smaller for maximal \spancores, difference that wears thin when the gap between
the numbers of \spancores and maximal \spancores starts decreasing since, for high values of $k$, most (or all) \spancores are maximal.

\begin{figure}[t!]
\centerline{
\begin{tabular}{cc}
\includegraphics[width=0.35\columnwidth]{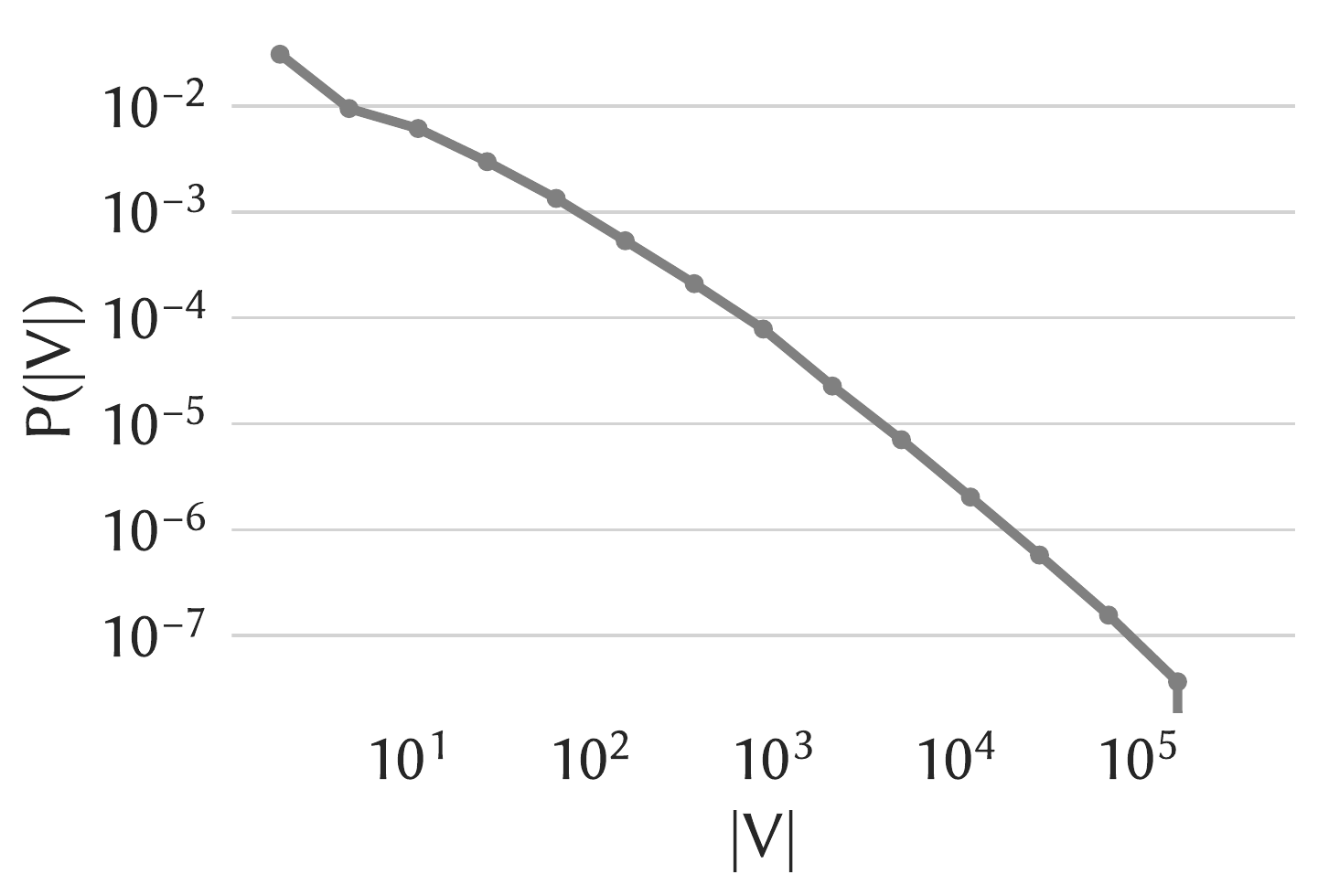} &
\includegraphics[width=0.35\columnwidth]{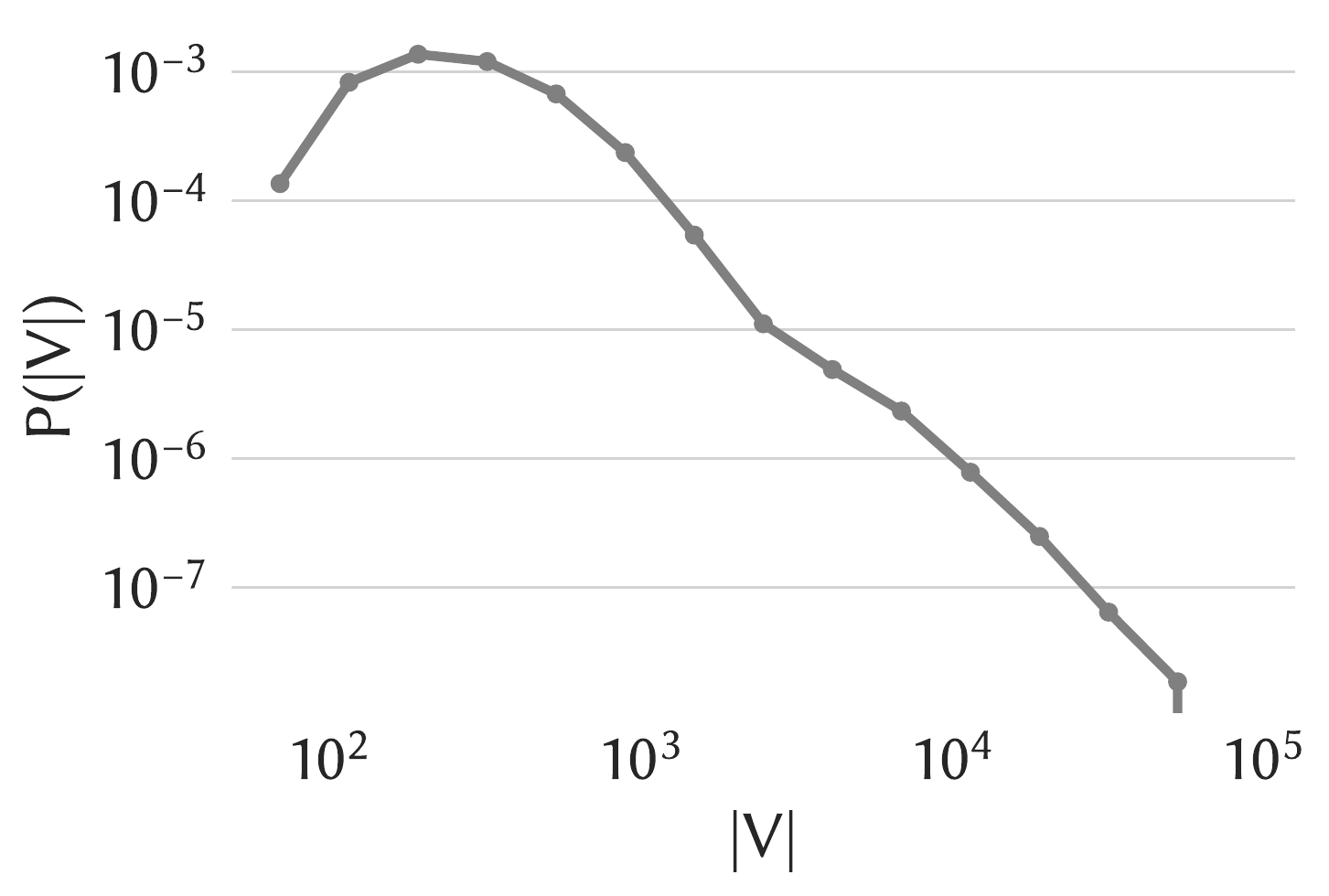}\\
\includegraphics[width=0.35\columnwidth]{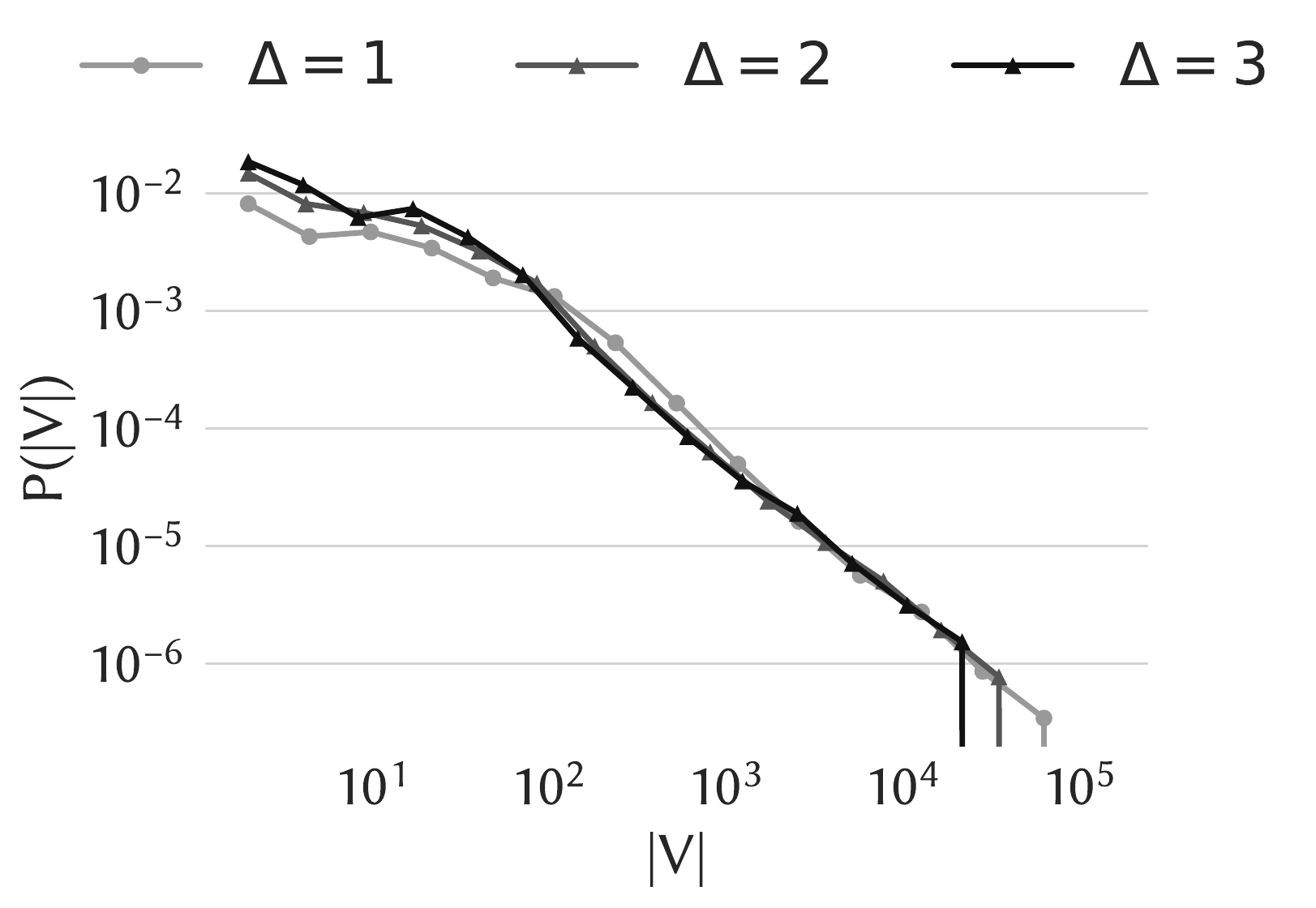} &
\includegraphics[width=0.35\columnwidth]{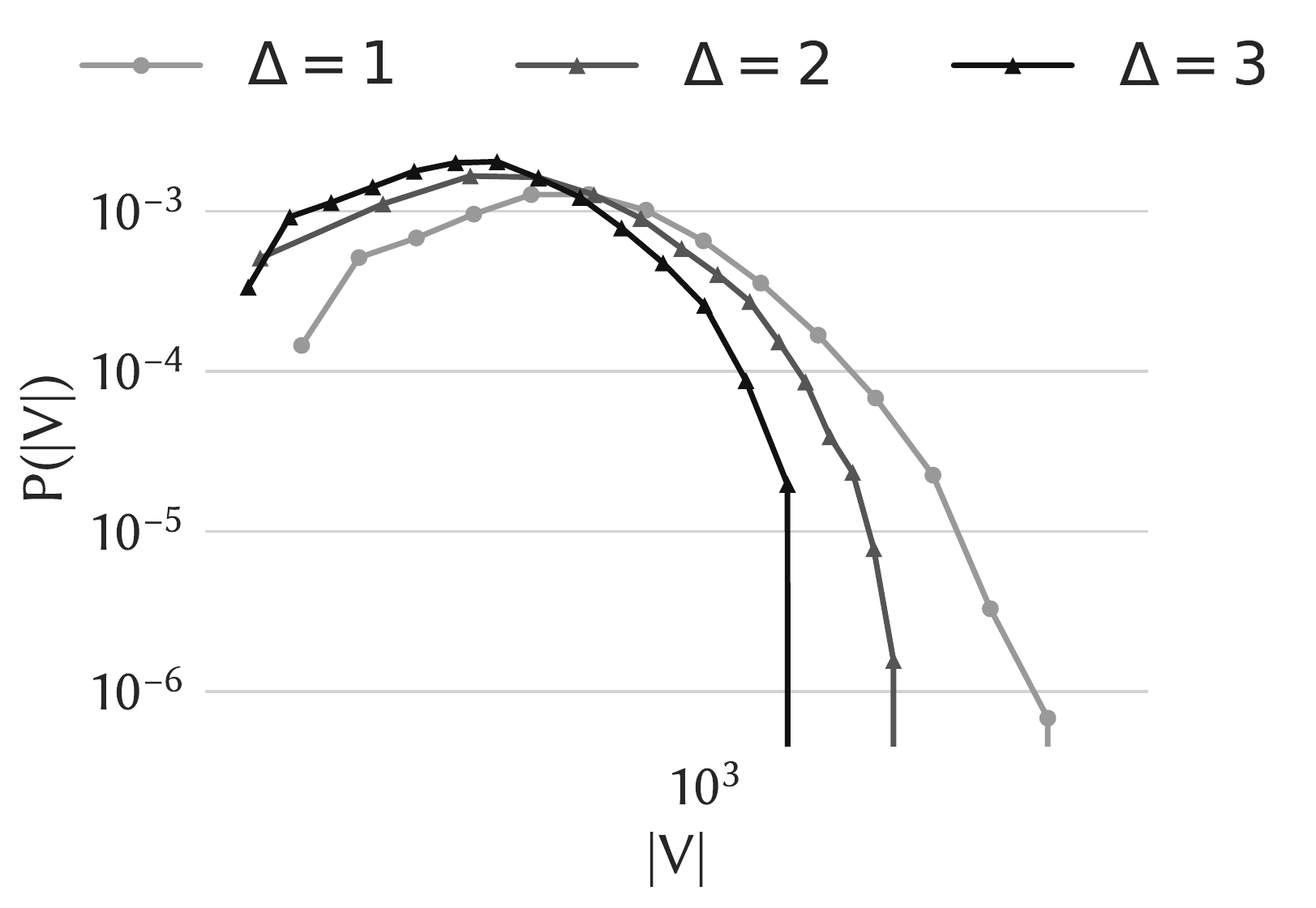}\\
\includegraphics[width=0.35\columnwidth]{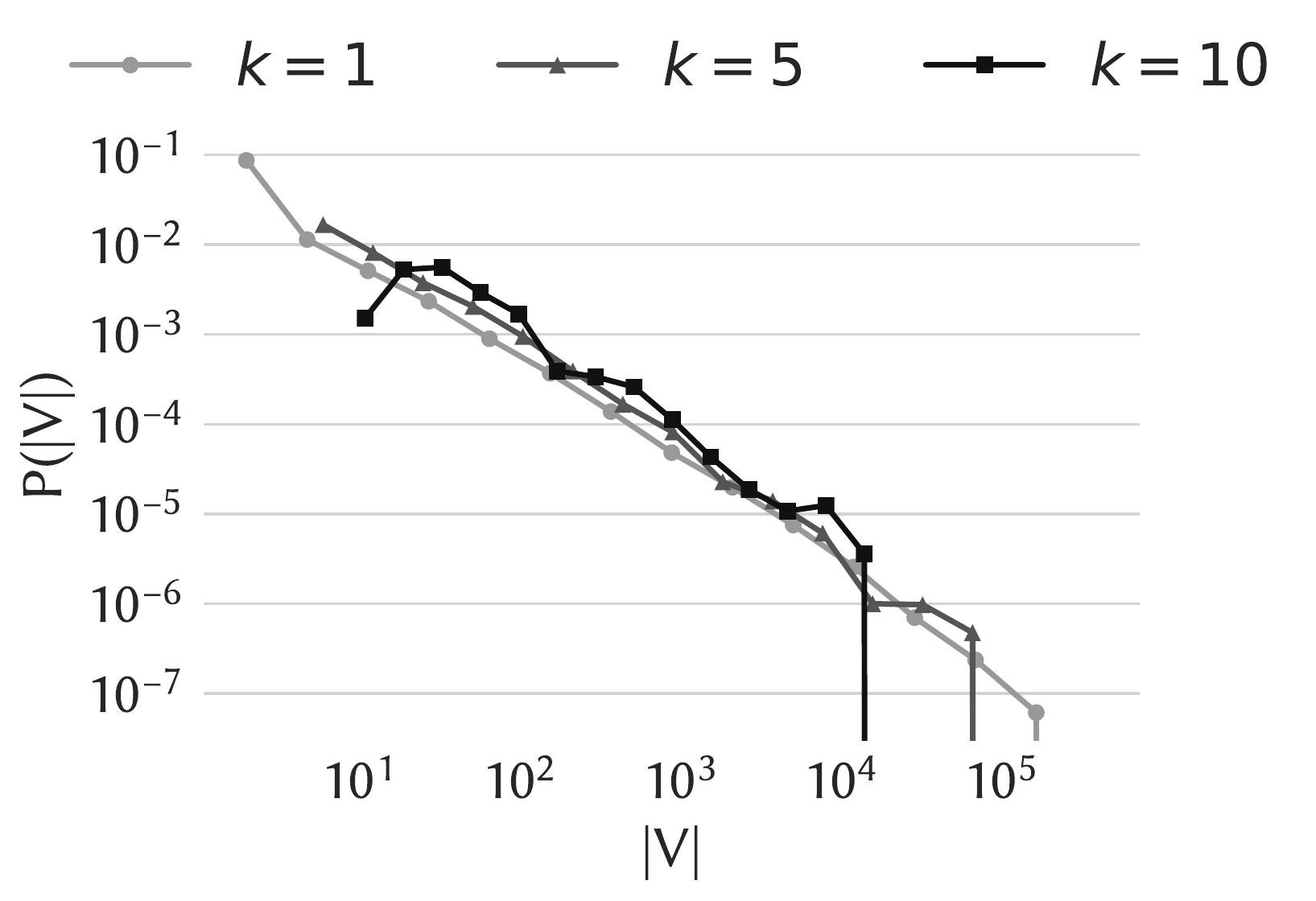} &
\includegraphics[width=0.35\columnwidth]{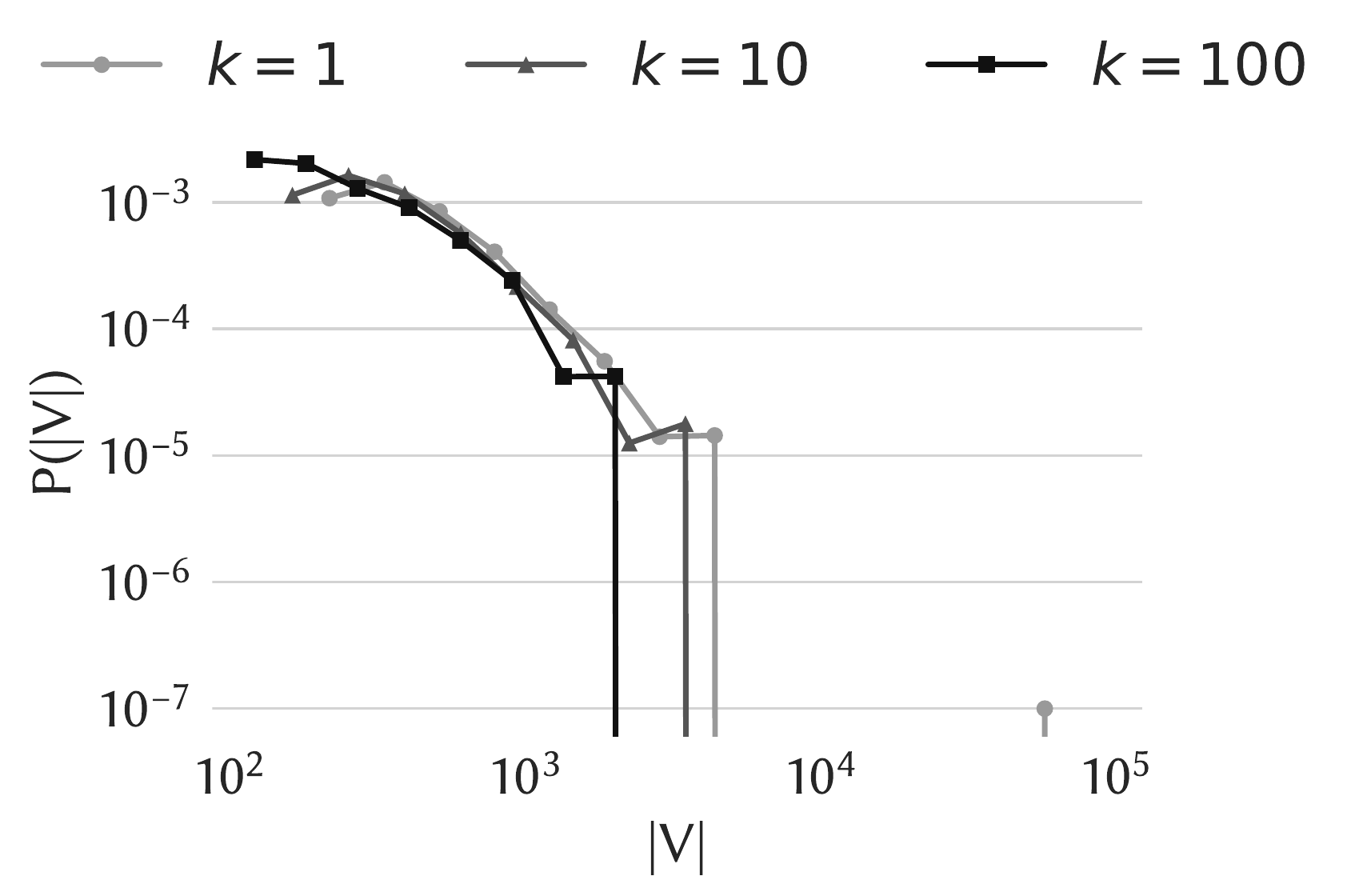}\\
\vspace{-2mm}\\
\footnotesize{\textsf{DBLP}} & \footnotesize{\textsf{Epinions}}
\end{tabular}
}
\caption{\label{fig:distrsizes}
Distribution of sizes of the \spancores.
Top plots: overall distributions. Middle plots: distributions of the sizes of \spancores with fixed temporal span.
Bottom plots: distributions of the sizes of \spancores with fixed order.}
\end{figure}

Figure~\ref{fig:coresvsmaximals_delta} shows a different picture when numbers and average sizes of \spancores are
shown as a function of the size of the span $|\Delta|$.
For both datasets, the number of \spancores and maximal \spancores is decreasing --
which is expected since the number of intervals decreases when $|\Delta|$ increases --
with a constant gap close to one and two orders of magnitude, respectively.
On the other hand, the behavior of the average size is quite different between the two datasets.
For low values of $|\Delta|$, the average size of \spancores of the \textsf{DBLP} dataset is much higher than the average size
of maximal \spancores, then the difference decreases and
vanishes at the end of domain where a maximal \spancore of $|\Delta| = 37$ dominates all other \spancores with $|\Delta| \geq 20$.
Instead, for the \textsf{Epinions} dataset, the average size of all \spancores and 
of maximal \spancores follow the same behavior, with a difference of less than
an order of magnitude, because the maximality condition over $k$ excludes the largest \spancores from the set of maximal \spancores.

Figure \ref{fig:distrsizes} yields some additional insights by showing the whole distribution of sizes of the \spancores:
these distributions are very skewed and span several orders of magnitude. The figure also shows the size distributions 
of \spancores with a given order or duration: all are broad, becoming narrower as the order or duration increase. 
We have also considered randomized versions of the datasets, in which edges are reshuffled at random at each timestamp.
In this case as well, the distributions of the sizes of the  \spancores are found to be broad, showing
that the heterogeneity in \spancore sizes can also be obtained in largely random data. However, 
the cohesive temporal structures are destroyed in the reshuffled data (see also Section \ref{sec:temppatterns}):
all the \spancores have then very small order and span.

{

\begin{figure}[t!]
	\centerline{
		\begin{tabular}{cc}
			\includegraphics[width=0.35\columnwidth]{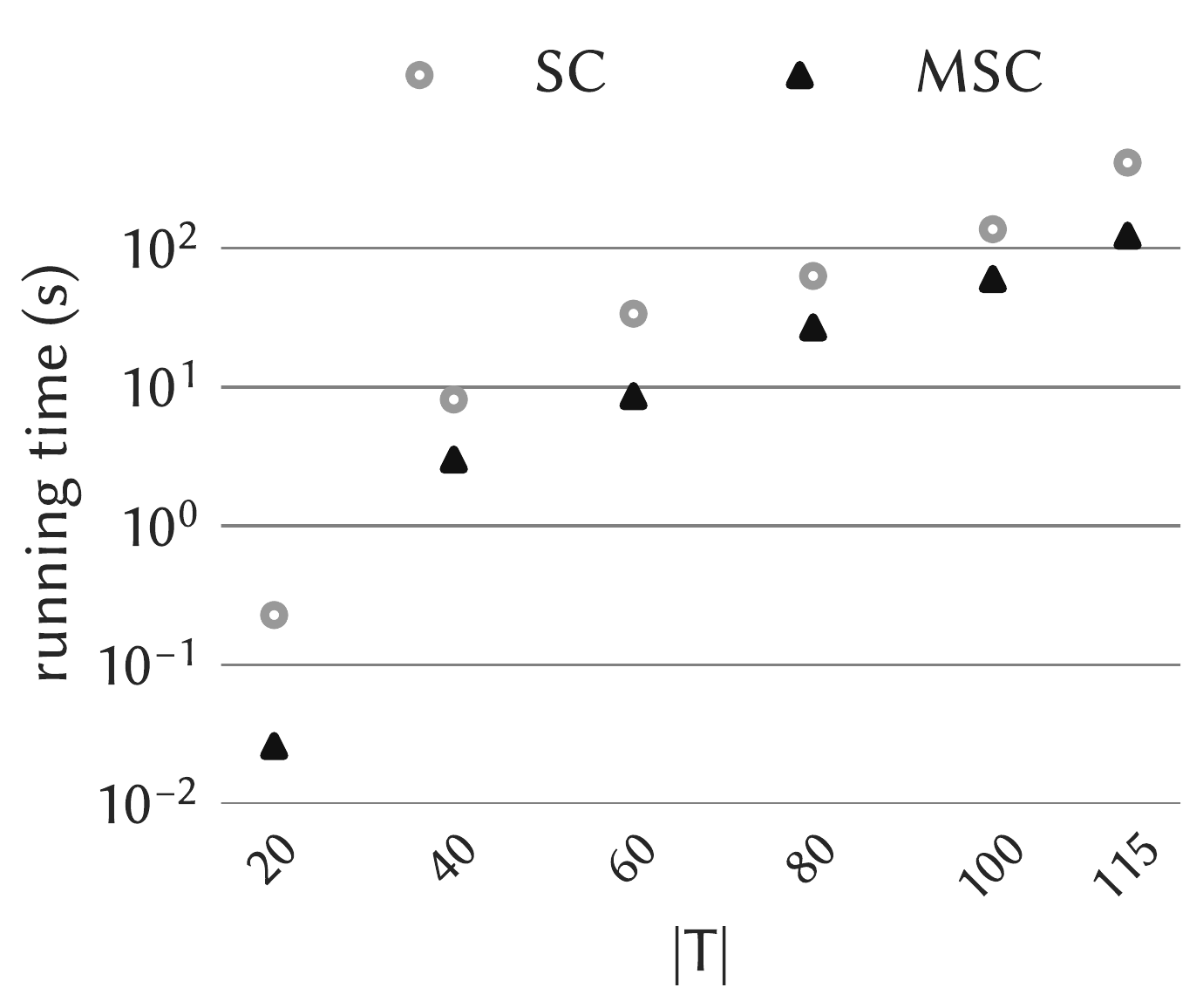} &
			\includegraphics[width=0.35\columnwidth]{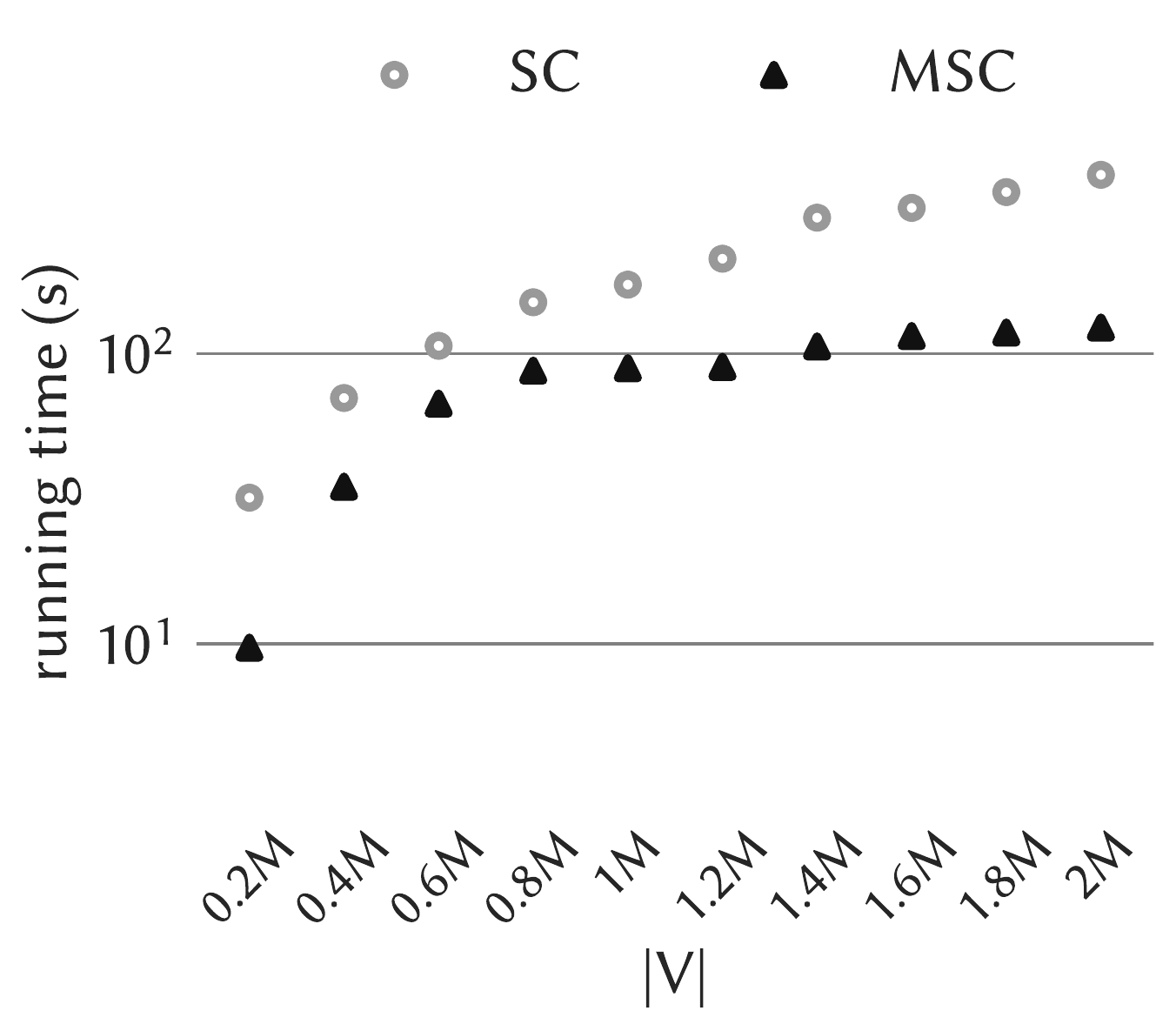}\\
		\end{tabular}
	}
	\caption{\label{fig:scability}Scalability analysis: running time of the algorithms for extracting all the span-cores (\textsf{SC}, Algorithm~\ref{alg:decomposition}) and only the maximal span-cores (\textsf{MSC}, Algorithm~\ref{alg:imcores}) as a function of the number timestamps and vertices, on the \textsf{Amazon} dataset.}
\end{figure}

\subsection{Scalability analysis}\label{sec:scalability}

Here we evaluate the scalability of Algorithms~\ref{alg:decomposition}~and~\ref{alg:imcores}.
To this aim, we consider temporal subgraphs derived from the \textsf{Amazon} dataset with increasing number of vertices and timestamps.
To obtain temporal subgraphs with varying number of timestamps, we simply consider the temporal graphs associated with the first $20$, $40$, $60$, $80$, $100$ timestamps, while considering all the vertices and edges existing in these time-frames.
For what concerns the temporal subgraphs with controlled number of vertices, we consider the entire temporal domain and sample sets of vertices of size $0.2M$, $0.4M$, $0.6M$, $0.8M$, $1M$, $1.2M$, $1.4M$, $1.6M$, $1.8M$.
Vertices are sampled according to the following simple procedure:
\begin{itemize}
	\item Select a vertex uniformly at random from the whole $V$ and add such a vertex to the set of sampled nodes $S$
	\item Starting from the first timestamp of the temporal domain $T$, iteratively
	\begin{itemize}
		\item For each vertex $v$ in $S$ select a neighbor of $v$ uniformly at random and add it to $S$
		\item Move to the next timestamp
	\end{itemize}
	\item If the last timestamp is reached, restart from the first one
\end{itemize}

The results of this scalability experiment are reported in Figure~\ref{fig:scability}.
It can be observed that the trends shown in the figure comply with the time complexities discussed in Sections~\ref{sec:spancores}--\ref{sec:maximal_spancores}: running times are quadratic in the number of timestamps and linear in the number of vertices/edges.
Another relevant consideration is that the difference between the algorithm to compute maximal \spancores and the one for computing all the \spancores gets larger as the number of timestamps or vertices increase. This further attests the usefulness of introducing an algorithm that is specifically devoted to maximal-\spancore computation.
}

\subsection{Temporal community search}\label{sec:experiments_cs}
In this subsection we assess the performance of the proposed algorithms for temporal community search (presented in Sections~\ref{sec:alg_naive_cs}--\ref{sec:alg_mcs}), as well as the greedy procedure for reducing the size of the output communities (presented in Section~\ref{sec:min_cs}).
In the remainder of this subsection we refer to our basic algorithm (i.e., Algorithm~\ref{alg:temporalcsnaive}, which precomputes the penalty scores via \spancore decomposition) as \textsf{SC-TCS}, and to our more efficient  algorithm (i.e., Algorithm~\ref{alg:maximaltemporalcs}, which exploits maximal \spancores to reduce the number of timestamps to be considered) as \textsf{MSC-TCS}.
We also involve in the comparison a na\"ive version of Algorithm~\ref{alg:temporalcsnaive}, where the penalty scores of the various intervals are computed from scratch during the execution of the algorithm, instead of precomputing them all via \spancore decomposition.
We refer to such a na\"ive method as \textsf{Na\"ive-TCS}.

The experimental setting we consider here is as follows.
We vary the number $|Q|$ of query vertices from 1 to 3.
In particular, when $|Q| = 1$, we sample the single query vertex uniformly at random from the whole vertex set $V$.
Instead, for $|Q| > 1$, we employ a more sophisticated sampling strategy that aims at finding meaningful query-vertex sets, i.e., vertices interacting with each other during the temporal observations, and, at the same time, independent from the specific form of the resulting \spancore decomposition.
Specifically, the sampling strategy we use is based on an adaptation of random walk to the temporal settings:
\begin{itemize}
\item Select a vertex uniformly at random from the whole $V$ and add such a vertex to the set $Q_{\visited}$ of visited vertices
\item Starting from the first timestamp of the temporal domain $T$, iteratively:
\begin{itemize}
\item With probability $p$, move the random walker to a neighbor of the current vertex and add the neighbor to $Q_{\visited}$. If the current vertex has no neighbors in a given timestamp, the random walker jumps to the first next timestamp in which that vertex has at least one neighbor
\item With probability $1 - p$, keep the random walker at the current vertex, but go to the next timestamp
\end{itemize}
\item Restart if the last timestamp of $T$ is reached
\item Stop when $|Q_{\visited}|$ reaches a proper (user-defined) size $\nu$
\item Sample $|Q|$ query vertices from $Q_{\visited}$ with probability proportional to the frequency of the visits during the random walk
\end{itemize}
In our experiments we set $p = 0.8$ and $\nu = 3|Q|$.
As far as the number $h$ of output communities, we consider the range $h \in [10, 20, 30, 40, 50, 60]$ on all datasets, with the exception of \textsf{StackOverflow}, for which we discard $h = 60$, and \textsf{Epinions}, for which we consider $h \in [4, 8, 12, 16, 20, 24]$.
For every parameter configuration, we perform five runs of every algorithm (in every run we sample a different query-vertex set).
Note that we were not able to run the algorithms for temporal community search on the \textsf{WikiTalk} dataset due to memory constraints.

\spara{Running time.}
In Figure~\ref{fig:cs_runtime} we show the running time of the proposed algorithms as a function of the number $h$ of output communities, for the \textsf{HighSchool}, \textsf{DBLP}, \textsf{Wikipedia}, and \textsf{Amazon} datasets.
The first general observation we make is that the running time of all algorithms increases as $h$ gets higher.
This in accordance with the time-complexity analysis reported in Section~\ref{sec:community_search}.
Also, running times are independent of the selected query-vertex set $Q$.
Looking at the individual performance, we notice that, as expected, the \textsf{Na\"ive-TCS} method has severe limitations in terms of efficiency:
it takes hours to run on the \textsf{HighSchool} and \textsf{Wikipedia} datasets, while it is not able to terminate in less than 10 days on the remaining datasets.
\textsf{SC-TCS} and \textsf{MSC-TCS} are much faster than \textsf{Na\"ive-TCS}, achieving a speedup of up to more than four orders of magnitude.
\textsf{MSC-TCS} is in most cases faster than  \textsf{SC-TCS}, with speedup up to one order of magnitude (on \textsf{HighSchool}, for $h = 60$).
This confirms that the exploitation of the maximal \spancores is effective in both shortening the precomputation time and reducing the temporal domain considered in the dynamic-programming step.
The only exception is the \textsf{Wikipedia} dataset.
To dive deeper into the motivations of this exception, we report in Figure~\ref{fig:cs_runtime_detail} the split of the average running time of \textsf{SC-TCS} and \textsf{MSC-TCS} into the time spent in the dynamic-programming step (DP) (which also includes the identification of the reduced temporal domain $T^*$ for \textsf{MSC-TCS}),
and the precomputation time (i.e., the time required for computing all penalty scores via \spancore decomposition or maximal \spancores).
Interestingly, what affects the most the running time is the precomputation of the scores.
Apparently, the $Q$-constrained version of \cores\ is more efficient than \innermosts\ in some datasets, which we believe is due to the  structure of the search space.
On the other hand, these results confirm that the reduction of the temporal domain considered by the dynamic-programming step is actually effective since the DP running time of \textsf{MSC-TCS} is always less than (or equal to) the DP running time of \textsf{SC-TCS}.

\begin{figure}
\centerline{
\begin{tabular}{cc}
\includegraphics[width=0.35\columnwidth]{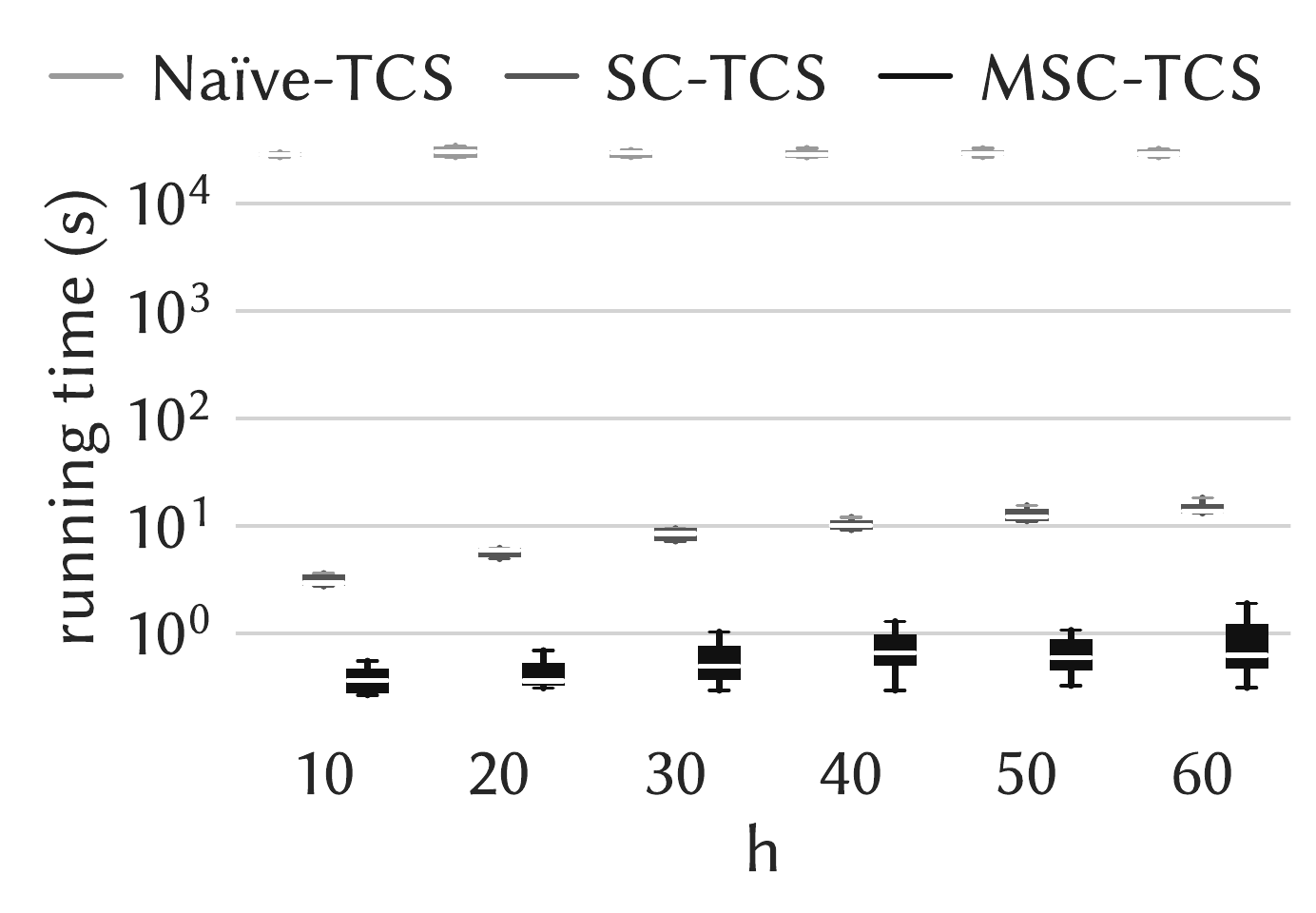} & \includegraphics[width=0.35\columnwidth]{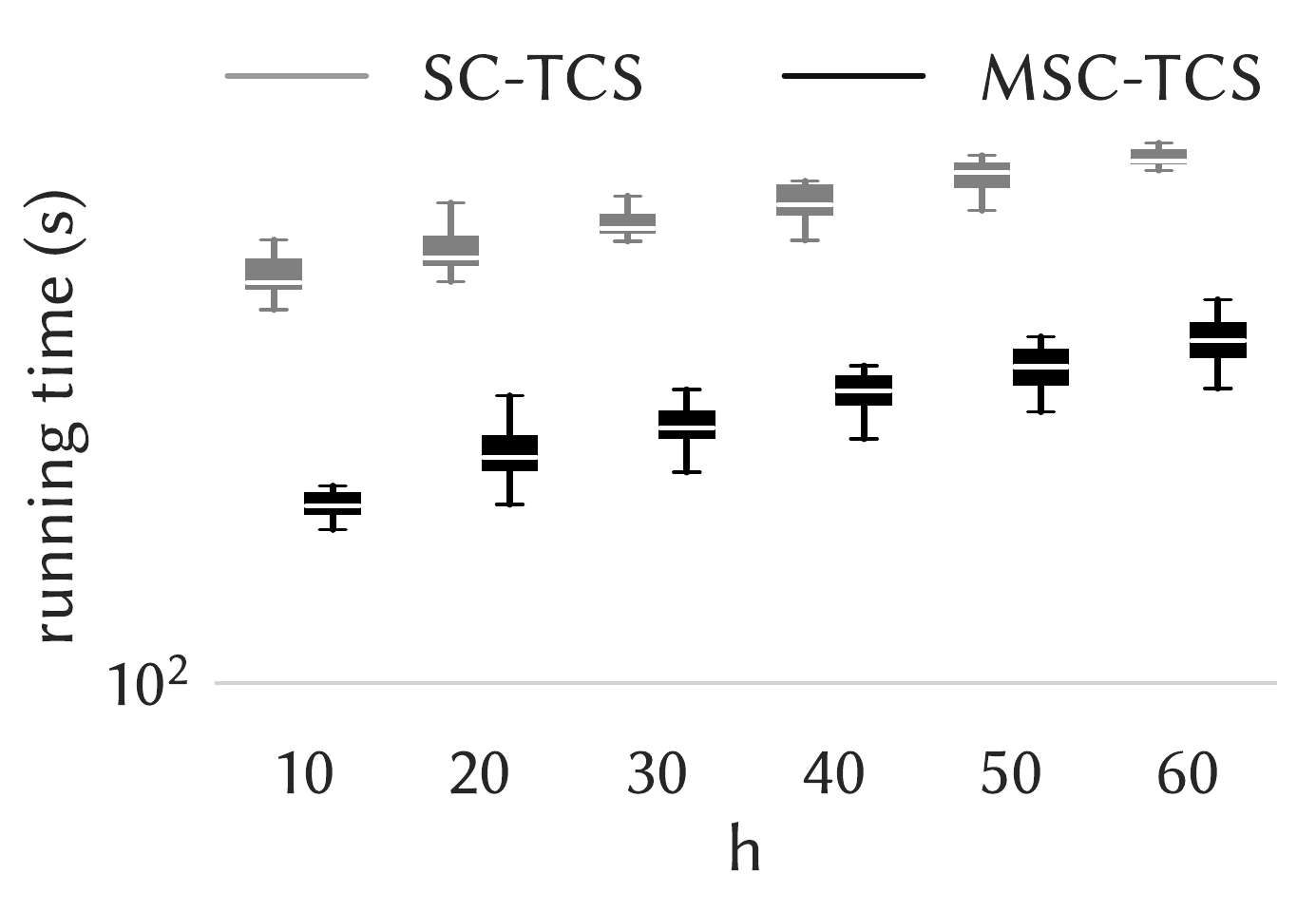} \vspace{-2mm}\\
\footnotesize{\textsf{HighSchool}} & \footnotesize{\textsf{DBLP}} \vspace{4mm}\\
\includegraphics[width=0.35\columnwidth]{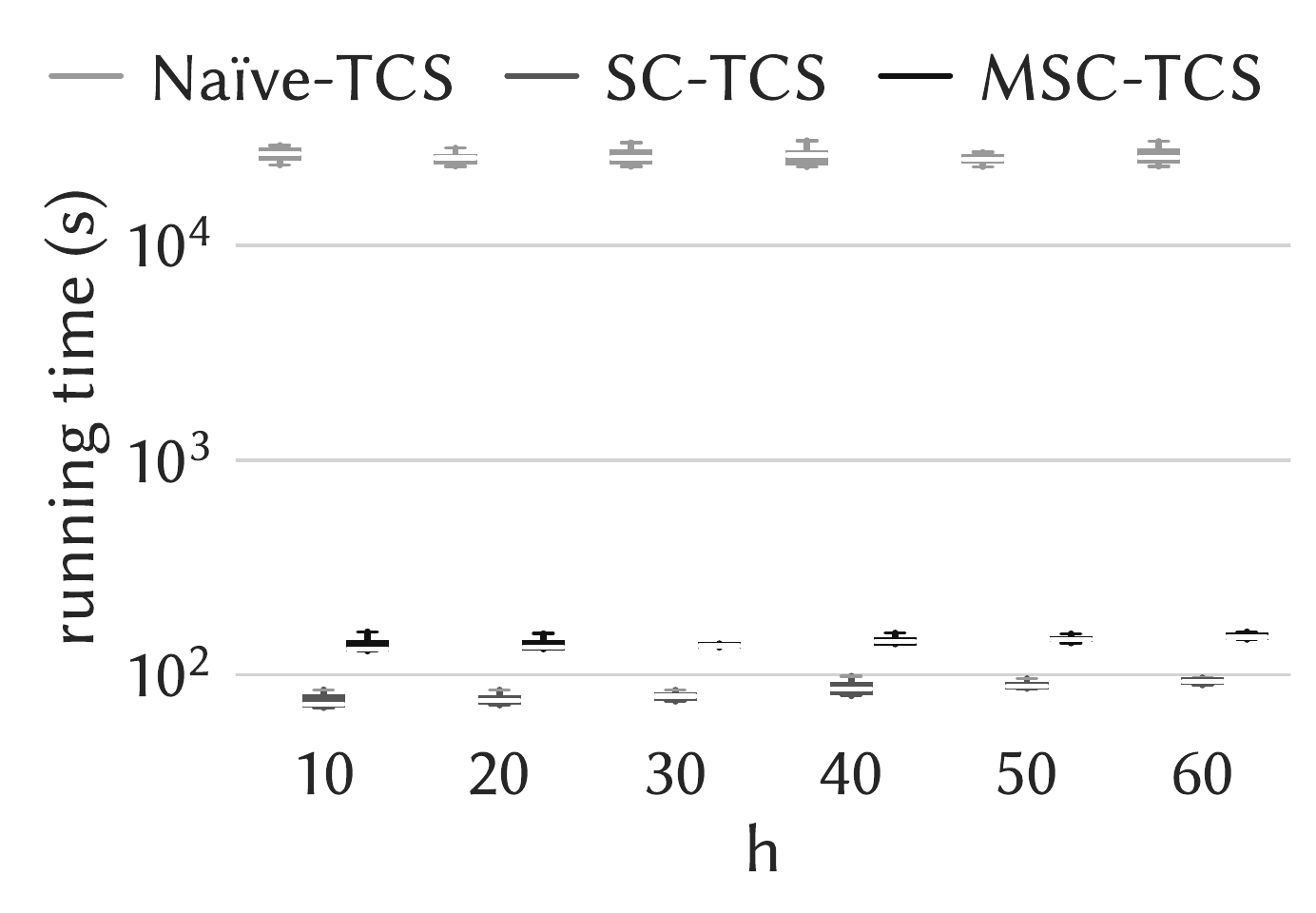} & \includegraphics[width=0.35\columnwidth]{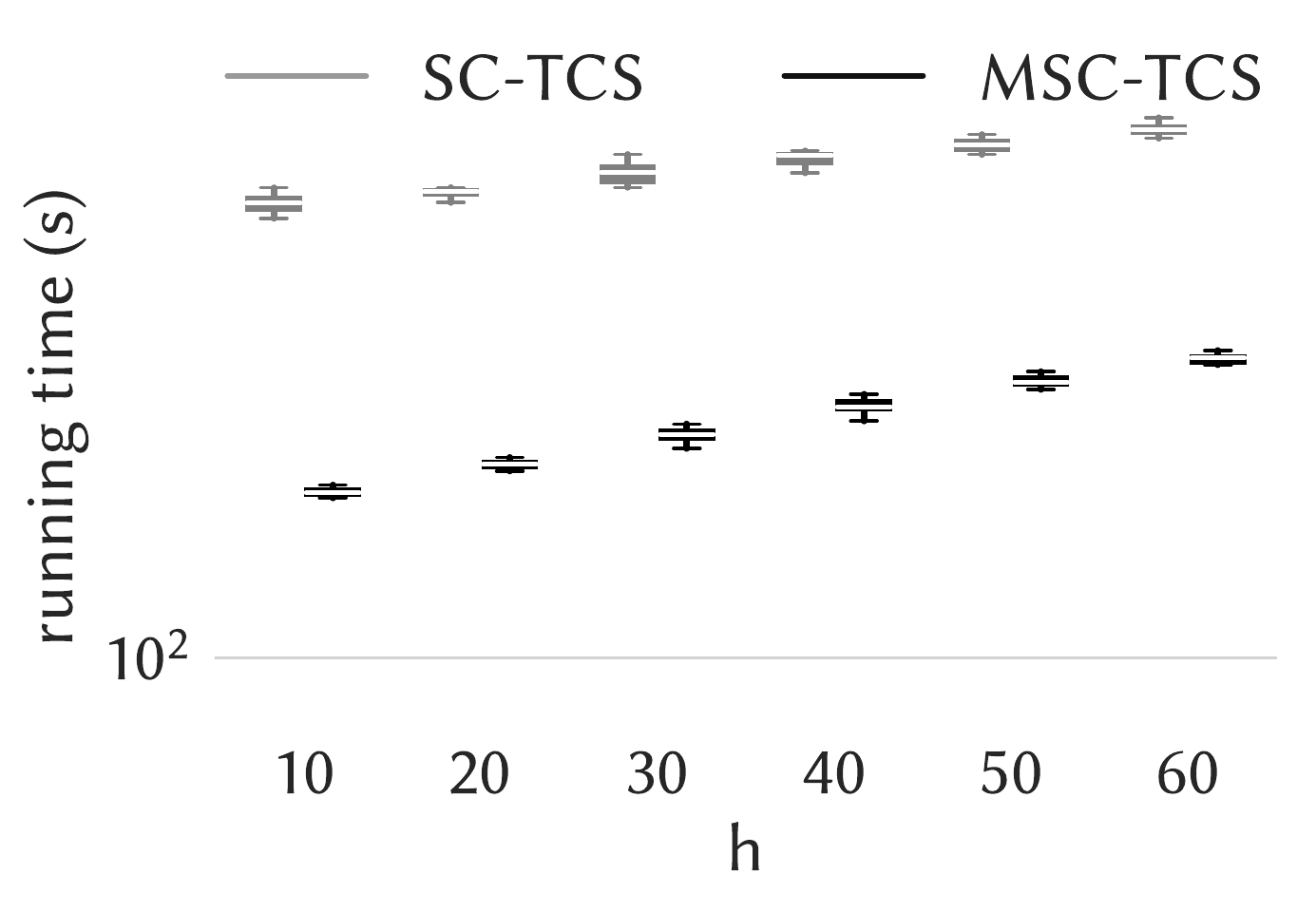} \vspace{-2mm}\\
\footnotesize{\textsf{Wikipedia}} & \footnotesize{\textsf{Amazon}}\\
\end{tabular}
}
\caption{\label{fig:cs_runtime} Running time of the algorithms for \temporalcs, as a function of the number $h$ of output communities.
Each boxplot corresponds to 15 data points.}
\end{figure}

\begin{figure}
\centerline{
\begin{tabular}{cc}
\includegraphics[width=0.35\columnwidth]{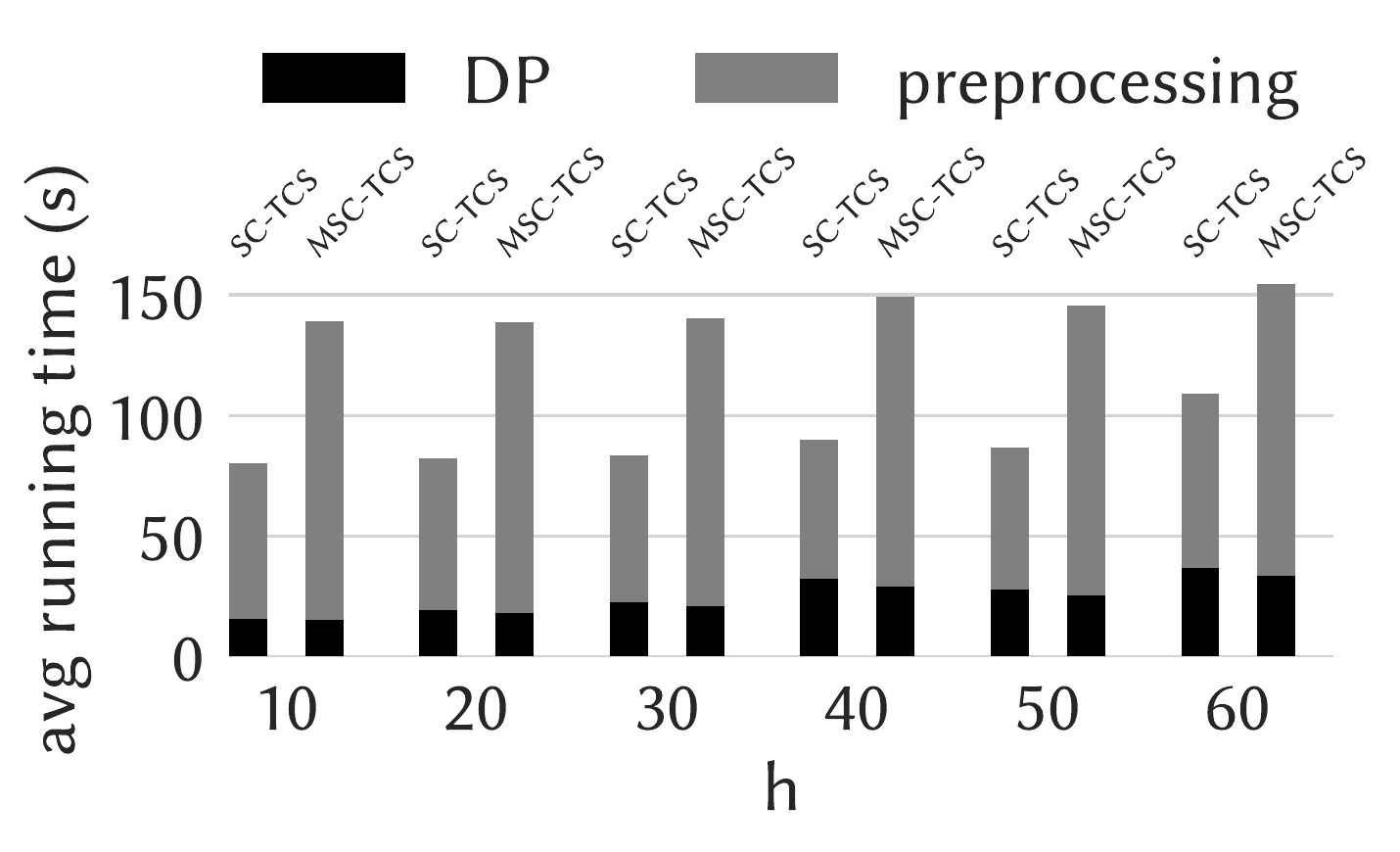} & \includegraphics[width=0.35\columnwidth]{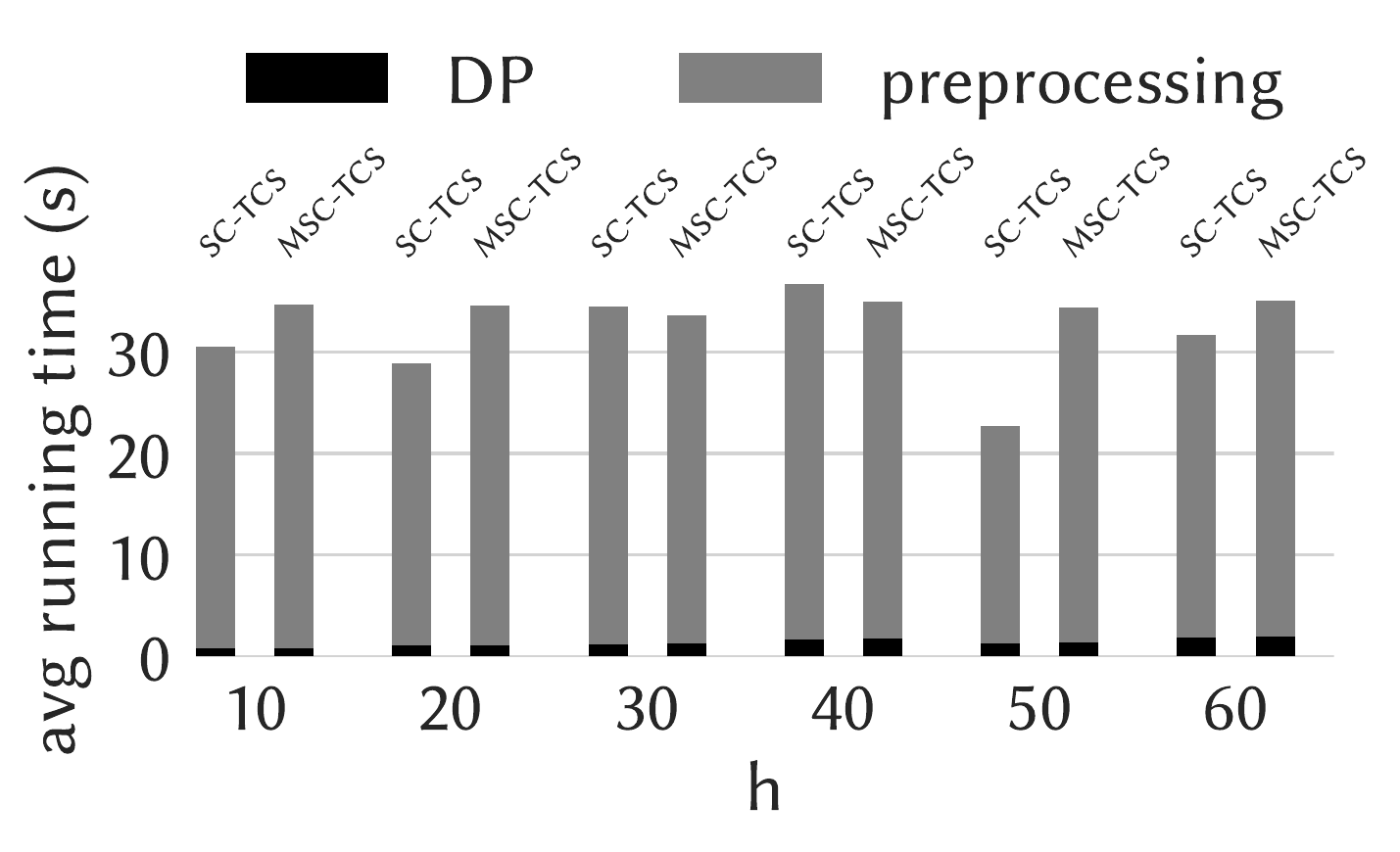} \vspace{-2mm}\\
\footnotesize{\textsf{Wikipedia}} & \footnotesize{\textsf{Last.fm}}\\
\end{tabular}
}
\caption{\label{fig:cs_runtime_detail} Split of the average running time of the \textsf{SC-TCS} and \textsf{MSC-TCS} algorithms into dynamic programming (DP) and precomputation, for the \textsf{Wikipedia} and \textsf{Last.fm} datasets.}
\end{figure}

\begin{figure}
\centerline{
\begin{tabular}{cc}
\includegraphics[width=0.35\columnwidth]{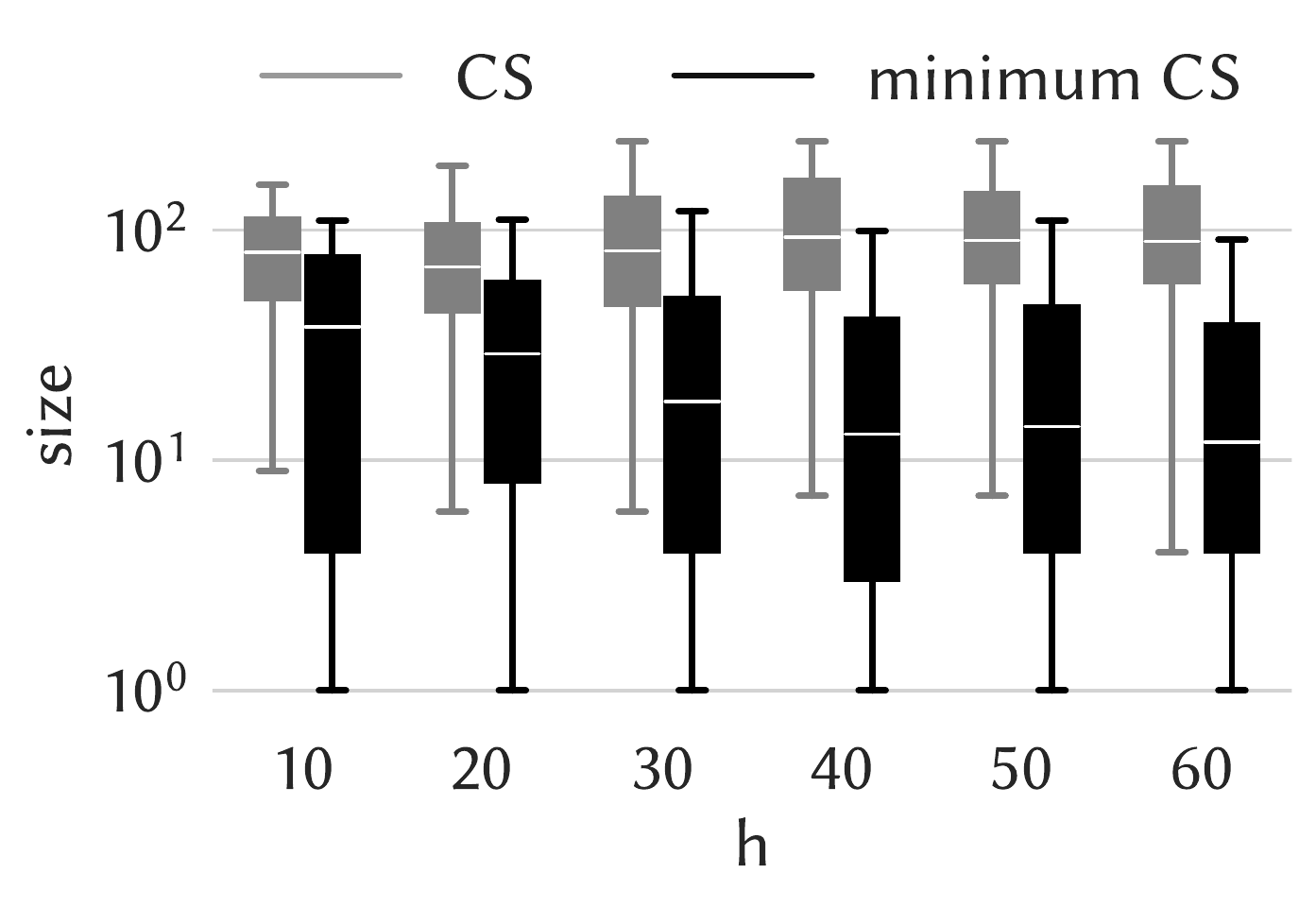} & \includegraphics[width=0.35\columnwidth]{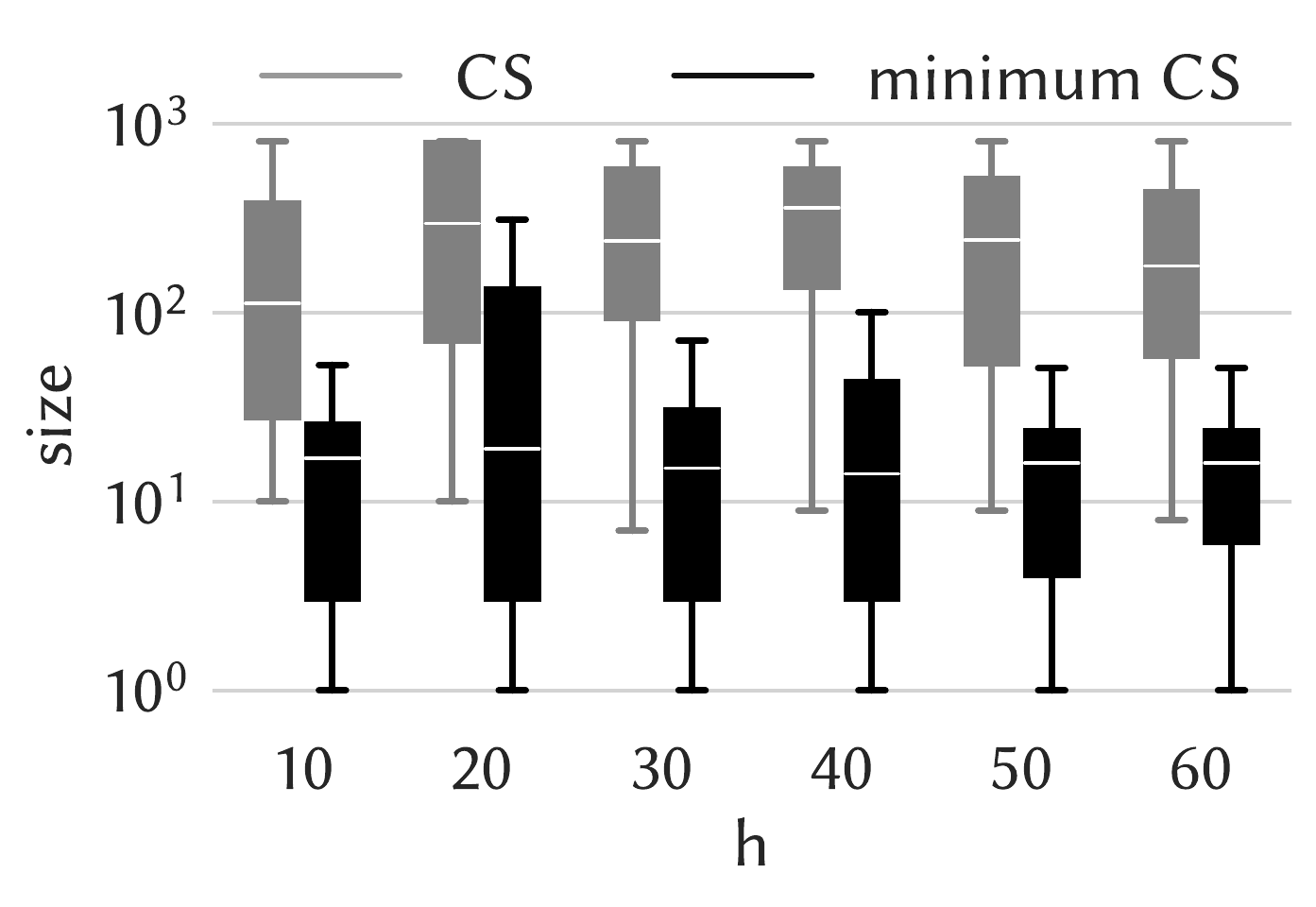} \vspace{-2mm}\\
\footnotesize{\textsf{PrimarySchool}} & \footnotesize{\textsf{HongKong}} \vspace{4mm}\\
\includegraphics[width=0.35\columnwidth]{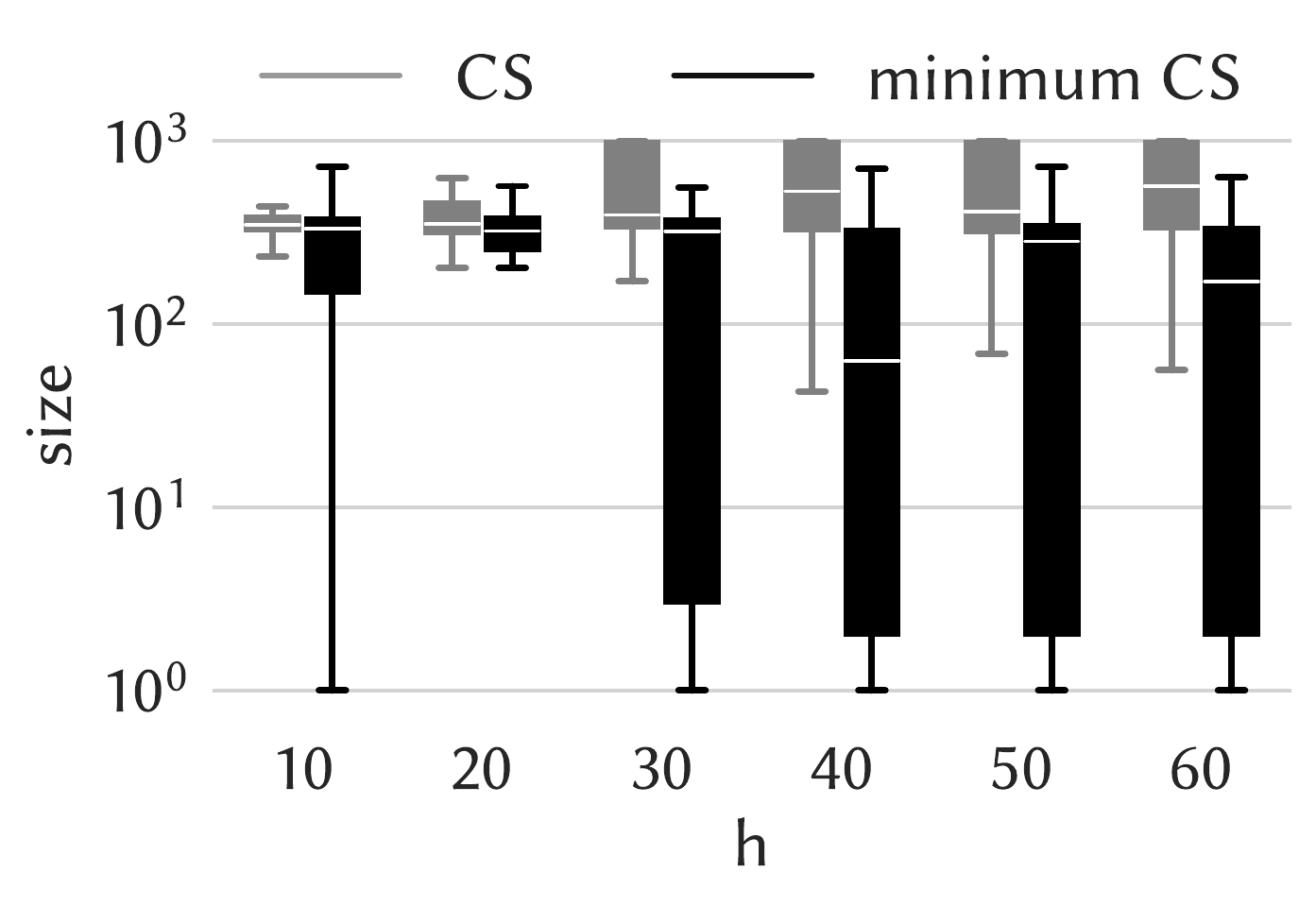} & \includegraphics[width=0.35\columnwidth]{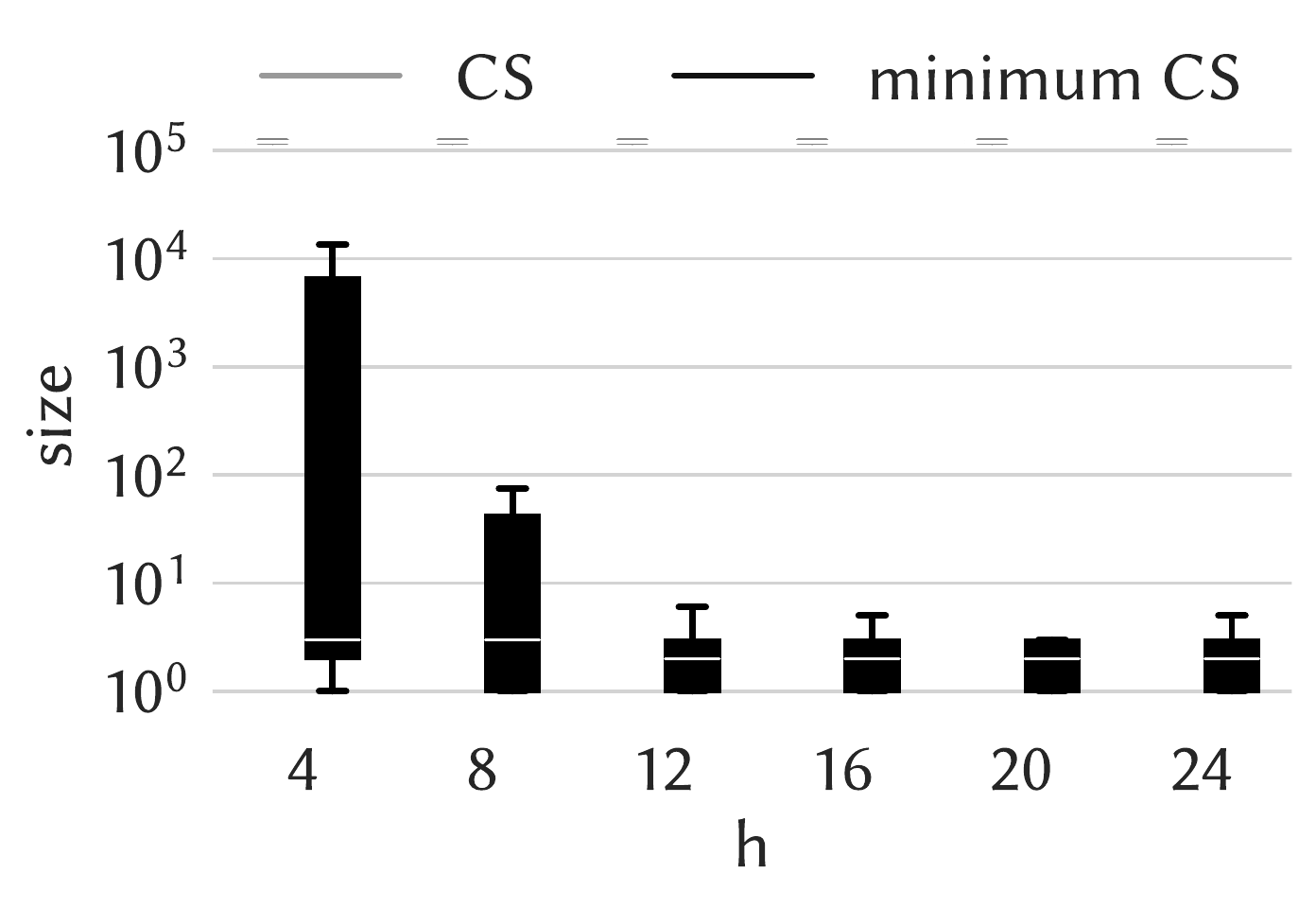} \vspace{-2mm}\\
\footnotesize{\textsf{Last.fm}} & \footnotesize{\textsf{Epinions}}\\
\end{tabular}
}
\caption{\label{fig:greedy_size} Comparison of the size of the communities in the solutions to \temporalcs:
original output of the algorithms for \temporalcs (CS) and after running the \greedy algorithm on top of them (minimum CS).
Each boxplot corresponds to 15 data points.}
\vspace{3mm}
\end{figure}

\spara{\greedy.}
Here we evaluate the performance of the proposed \greedy algorithm (Algorithm~\ref{alg:greedy}) for reducing the size of the output communities.
We recall that the proposed algorithms for \temporalcs (evaluated above) output communities corresponding to the $(Q, \Delta_i)$-highest-order-\spancores for all $\{\Delta_i\}_{i=1}^h$ temporal intervals identified.
The \greedy algorithm takes every $(Q, \Delta_i)$-highest-order-\spancore and attempts to reduce its size, while preserving optimality.
Thus, the ultimate goal of the evaluation presented next is to show how well \greedy is able to reduce the size of the original \spancores, and what is its overhead in terms of running time.


Figure~\ref{fig:greedy_size} compares the size of the starting $(Q, \Delta_i)$-highest-order-\spancores and the size of the corresponding reduced community yielded by the \greedy algorithm,  for the \textsf{PrimarySchool}, \textsf{HongKong}, \textsf{Last.fm}, and \textsf{Epinions} datasets.
It can be easily observed that, as a general trend, the reduced communities are much smaller than the original ones, in all datasets, up to four orders of magnitude.
The results on the \textsf{Epinions} dataset are a bit different than the other three datasets.
In fact, on that dataset, the original communities (CS) always include the whole $120$k vertices of the graph, while the communities found by \greedy (minimum CS) have median size smaller than $10$, and, in many cases, they correspond to communities composed of the query vertices only.
This means that, on the \textsf{Epinions} dataset, for our tested queries, the algorithms for \temporalcs do not extract communities that are really cohesive around the query vertices.
This way, the benefits of exploiting an a-posteriori community-size-reduction step are less evident.
Also, we do not notice any evident pattern as a function of $h$, for any dataset.

In Table~\ref{tab:greedy_runtime} we report the average running time of an execution of \greedy, for all datasets.
Note that this is the average time required to process one of the $h$ communities in a solution to \temporalcs.
\greedy runs in $8$ seconds or less in all tested datasets.
Therefore, the additional running time required by the algorithm is rather negligible.

To summarize, \greedy is empirically recognized as a powerful post-processing method for improving the quality of the solutions to \temporalcs:
it finds much smaller communities at a very small additional computational cost.

\begin{table}
\centering
\caption{Average running time of an execution of the \greedy algorithm.}
\label{tab:greedy_runtime}

\begin{tabular}{c|ccccccc}
& \textsf{HighSchool} & \textsf{PrimarySchool}  & \textsf{HongKong} & \textsf{ProsperLoans} & \textsf{Last.fm} \\
\hline
running time (s) & $0.003$ & $0.001$ & $0.02$ & $0.3$ & $0.06$ \\
\hline
\end{tabular}

\vspace{0.5cm}

\begin{tabular}{c|ccccccc}
& \textsf{DBLP}  & \textsf{StackOverflow} & \textsf{Wikipedia} & \textsf{Amazon} & \textsf{Epinions}  \\
\hline
running time (s) & $7$ & $8$ & $1$ & $7$ & $6$ \\
\hline
\end{tabular}
\end{table}

\section{Applications}
\label{sec:applications}

In this section we illustrate applications of (maximal) \spancores in the analysis of face-to-face interaction networks, and how the methods for \temporalcs can be profitably exploited in a task of graph classification.
For these applications we use the three networks gathered in schools, i.e., \textsf{PrimarySchool}, \textsf{HighSchool}, and \textsf{HongKong}, which are described above, at the beginning of Section~\ref{sec:experiments}.
We use a window size of $5$ minutes and, in the analysis, we discard \spancores of $|\Delta| = 1$, i.e., having span of $5$ minutes, since they represent short interactions, not significant for our purposes.
In the following we show
($i$) three types of interesting temporal patterns (Section~\ref{sec:temppatterns}), i.e., social activities of groups of students within a school day, mixing of gender and class, and length of social interactions in groups; ($ii$) a procedure to detect anomalous contacts and intervals that exploits maximal \spancores (Section~\ref{sec:anomaly});
and, ($iii$) an approach to graph classification based on temporal community search (Section~\ref{sec:classification}).

\begin{figure}[h!t!]
\begin{tabular}{c}
\centerline{\includegraphics[width=0.7\columnwidth]{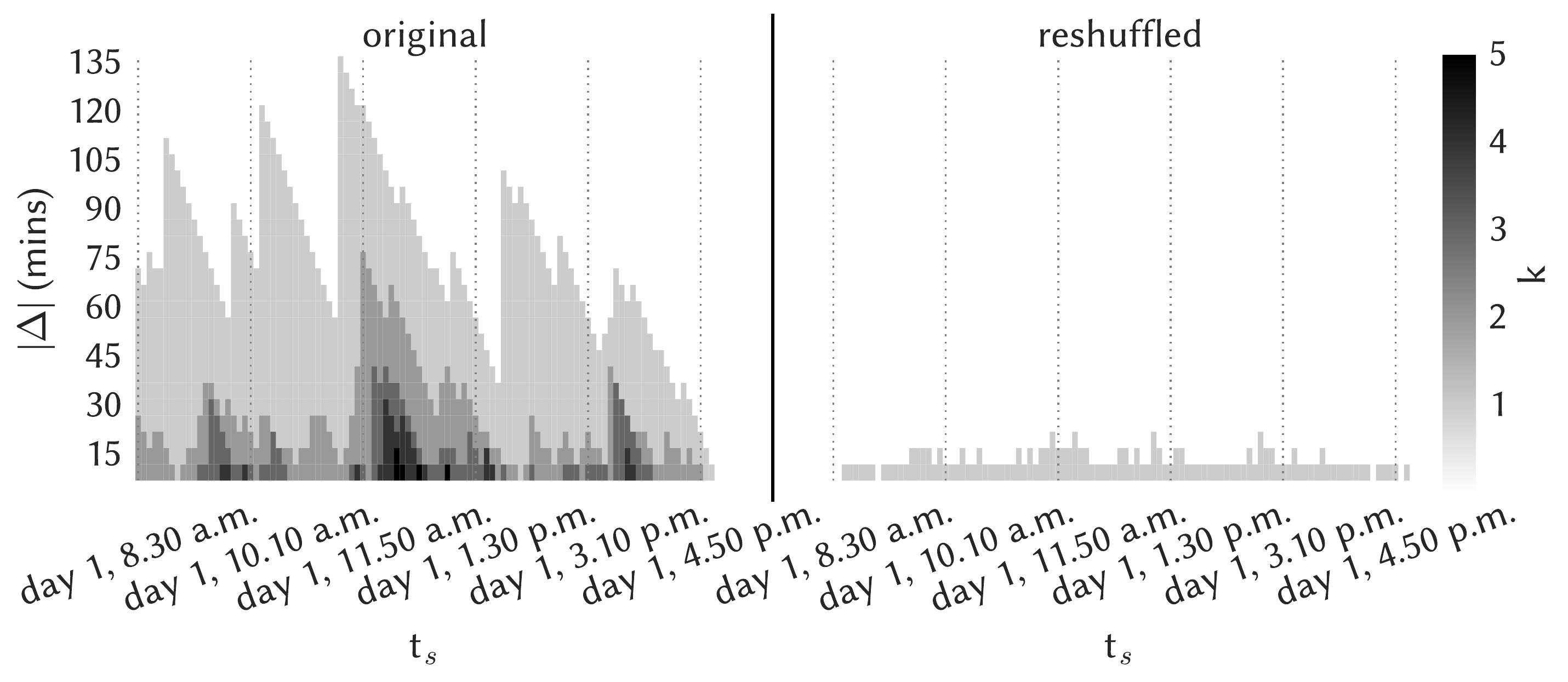}} \vspace{-2mm}\\
\footnotesize{\textsf{PrimarySchool}} \vspace{4mm}\\
\centerline{\includegraphics[width=0.7\columnwidth]{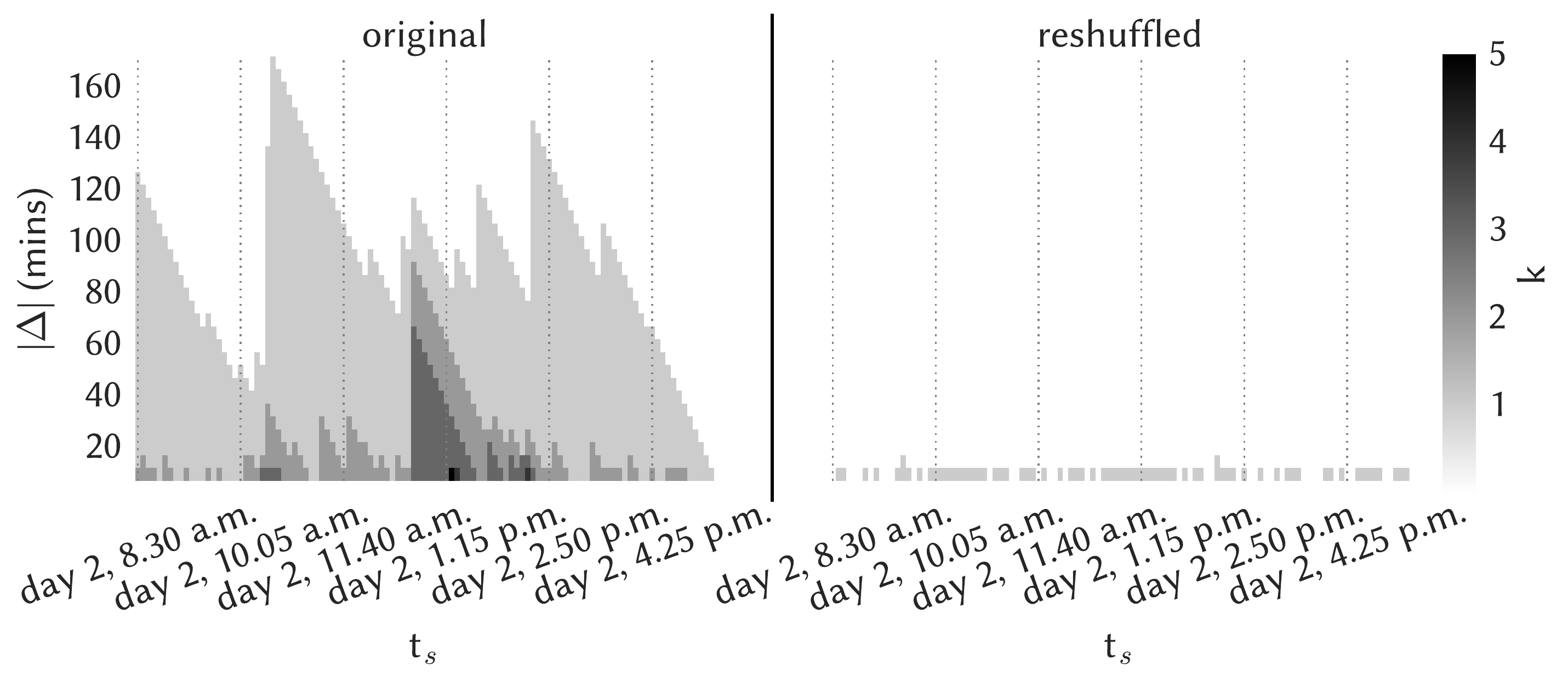}} \vspace{-2mm}\\
\footnotesize{\textsf{HighSchool}}
\end{tabular}
\caption{\label{fig:temporalactivity} Temporal activity of a school day of the \textsf{PrimarySchool} and \textsf{HighSchool} datasets: the $x$ axis reports the hour of the day at which the span of a \spancore starts, the $y$ axis specifies the size of the span (in minutes), and the color scale shows the order $k$.
At a glance, it can be observed that the temporal structure of the \spancore decomposition detects time-evolving cohesive
structures in the original datasets (left plots) that completely disappear in the reshuffled datasets (right plots).}
\end{figure}

\subsection{Temporal patterns}
\label{sec:temppatterns}
\spara{Temporal activity.}
We first show how \spancores\ afford a simple temporal analysis of social activities of groups of people within a school day.
The left side of Figure~\ref{fig:temporalactivity} reports colormaps of the order $k$ of the span-cores as a function of their starting time $t_s$ ($x$ axis) and of the size of their temporal span $|\Delta|$ ($y$ axis), for a school day of the \textsf{PrimarySchool} and \textsf{HighSchool} datasets.
Darker gray indicates \spancores of high order and slots located in the upper part of the plots refer to \spancores of long span.
It is important to notice that the linear decay in span duration is naturally due to the definition of \spancore and to the shifting of the starting time $t_s$;
therefore, it is not a distinguishing feature of the activity patterns found in the analyzed data.
In both datasets, fluctuations of $k$ and $|\Delta|$ are observed along the day, which can be related to school events.
Around $10$ a.m., the size of the span $|\Delta|$ reaches a local maximum in correspondence to the morning break, which means that students establish long-lasting interactions that hold beyond the break itself.
Moreover, when classes gather for the lunch break, the order $k$ reaches its maximum value since students tend to form larger and more cohesive groups.

In order to verify that these results are not trivially derived from the general temporal activity, as simply given by the number of interactions in each timestamp, we compare our findings to a null model.
At each timestamp of the temporal graphs, we reshuffle the edges by the Maslov-Sneppen algorithm~\cite{maslov2002specificity} which consists in repeating the following operations up to when all edges have been processed:
select at random two edges with no common vertices, e.g., $(u,v)$ and $(w,z)$, and transform them into $(u,z)$ and $(w,v)$,
if neither $(u,z)$ and $(w,v)$ existed in the original timestamp.
This reshuffling preserves the degree of each vertex in each timestamp and the global activity (i.e., the
number of contacts per timestamp), but destroys correlations between edges of successive timestamps.
In the right side of Figure~\ref{fig:temporalactivity} we show the results of the temporal analysis described above for the reshuffled datasets.
In both, the values of $|\Delta|$ and $k$ reached are much smaller than in the original datasets.
The size of the span $|\Delta|$ is always shorter than $20$ minutes, while in the original datasets it is much longer, up to $170$ minutes, and the order $k$ is always equal to $1$, compared to the original maximum of $5$.
The time-evolving cohesive structures detected by the temporal core decomposition
in the original datasets are completely lost on reshuffling,
since only span-cores of short span and low coreness are observed in the latter case.
This shows that the temporal structure exposed by the \spancore decomposition is not simply a consequence
of temporal patterns of global activity but that \spancores represent a concrete method to detect complex cohesive structures and their temporal
evolution.

\spara{Mixing patterns.}
We now show an analysis of mixing patterns of students with respect to gender and class.
Such vertex attributes are indeed available for the individuals of the \textsf{PrimarySchool} dataset.
We define as \emph{gender purity} of a \spancore the fraction of individuals of the most represented gender within the \spancore.
\emph{Class purity} is analogously defined.
The left plot of Figure~\ref{fig:purity} reports the temporal evolution of the average gender and class purity
of the maximal \spancores spanning each timestamp,
during the first school day of the \textsf{PrimarySchool} dataset.
During lessons, when students are in their own classes, class purity has
naturally very high values, very close to $1$. Gender purity is instead rather low.
On the other hand, when students are gathered together, during the morning break at $10$ a.m. and the lunch break between $12$ a.m. and $2$ p.m., the situation is overturned: gender purity reaches large values while class purity drastically decreases.
This shows that primary school students group with individuals of the same class, disregarding the gender, only when they are forced by the schedule of the lessons, but prefer on average to form cohesive groups with students of the same gender during breaks. This is in agreement
and complements a previous study of the same dataset focusing on single interactions in the static aggregated network \cite{Stehle:2013}.

The right plot of Figure~\ref{fig:purity} shows the temporal evolution of the average
gender and class purity for a null model in which
 gender and class are randomly reshuffled among individuals.
The two curves are more flat and the anti-correlation between them completely vanishes.
This testifies that the results on the original dataset are not simply due to the relative abundance
of individuals of each type interacting at each time, but reflect genuine mixing patterns and their temporal evolution.


\begin{figure}
\centerline{\includegraphics[width=0.7\columnwidth]{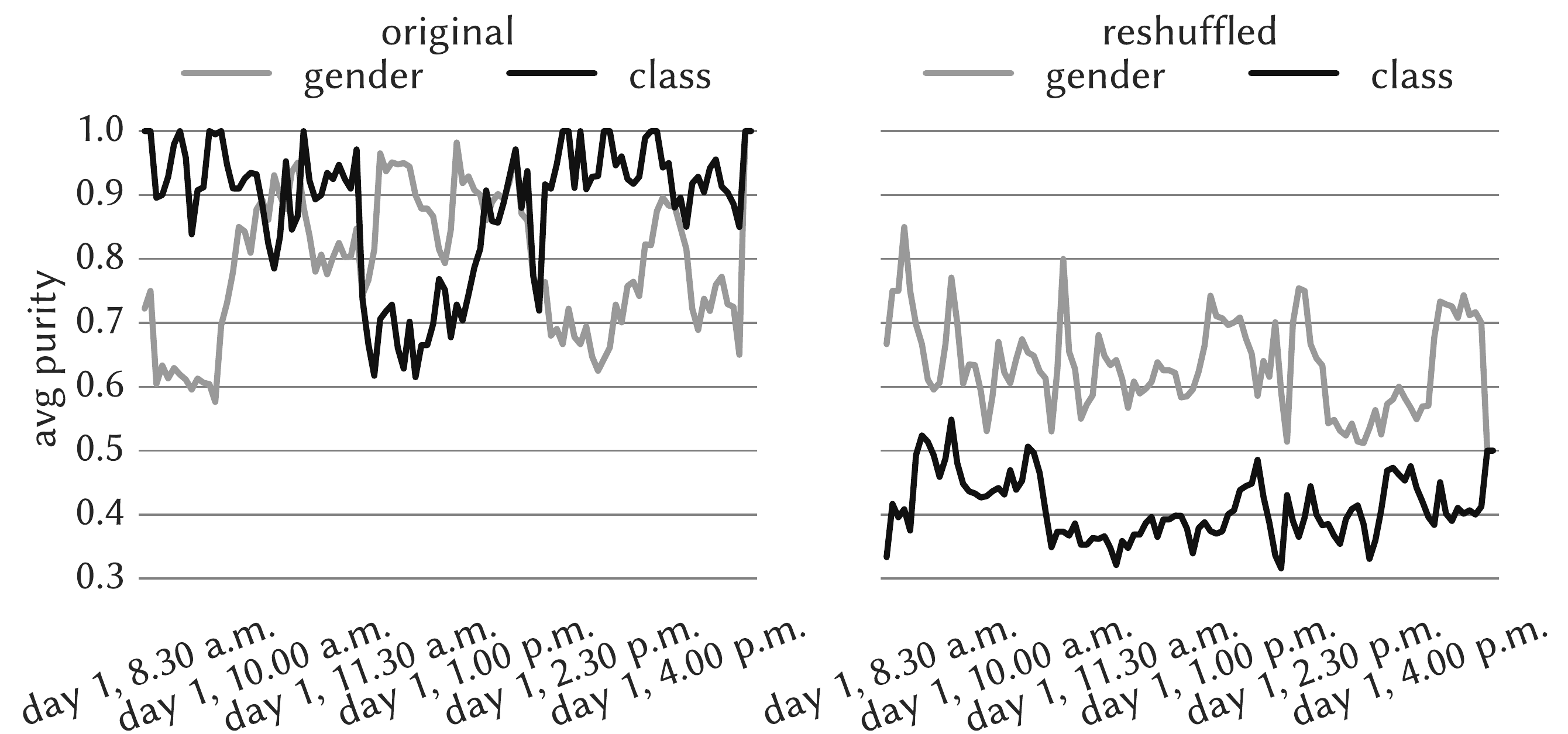}}
\caption{\label{fig:purity} Temporal evolution (time on the $x$ axis) of average gender purity and average class purity ($y$ axis) of the maximal \spancores of the \textsf{PrimarySchool} dataset.
Original data on the left, reshuffled data on the right.}
\end{figure}

\spara{Interaction length.}
Finally, we analyze the duration of interactions of social groups in schools by studying the distribution of the size of the span of the maximal \spancores of the three datasets (Figure~\ref{fig:lengthdistribution}). All distributions are extremely skewed with broad tails: most maximal span-cores have duration less than $1$ hour, but
durations much larger than the average can also be observed.
Interestingly, the three datasets at hand all exhibit the same functional shape, confirming a robust statistical behavior.
We also note that similar robust broad distributions have been observed for simpler characteristics of human interactions such as the statistics of contact durations \cite{Stehle:2011,Fournet:PLOS2015}.
Outliers appear also at very large durations, especially for the \textsf{HongKong} dataset that has maximal \spancores lasting up to $83$ hours.
Group interactions of such long span are clearly abnormal and represent outliers in the distributions.
We will show, in the following of this section, how to exploit such outliers to detect both irregular interactions and anomalous temporal intervals.

\begin{figure}
\centerline{\includegraphics[width=0.7\columnwidth]{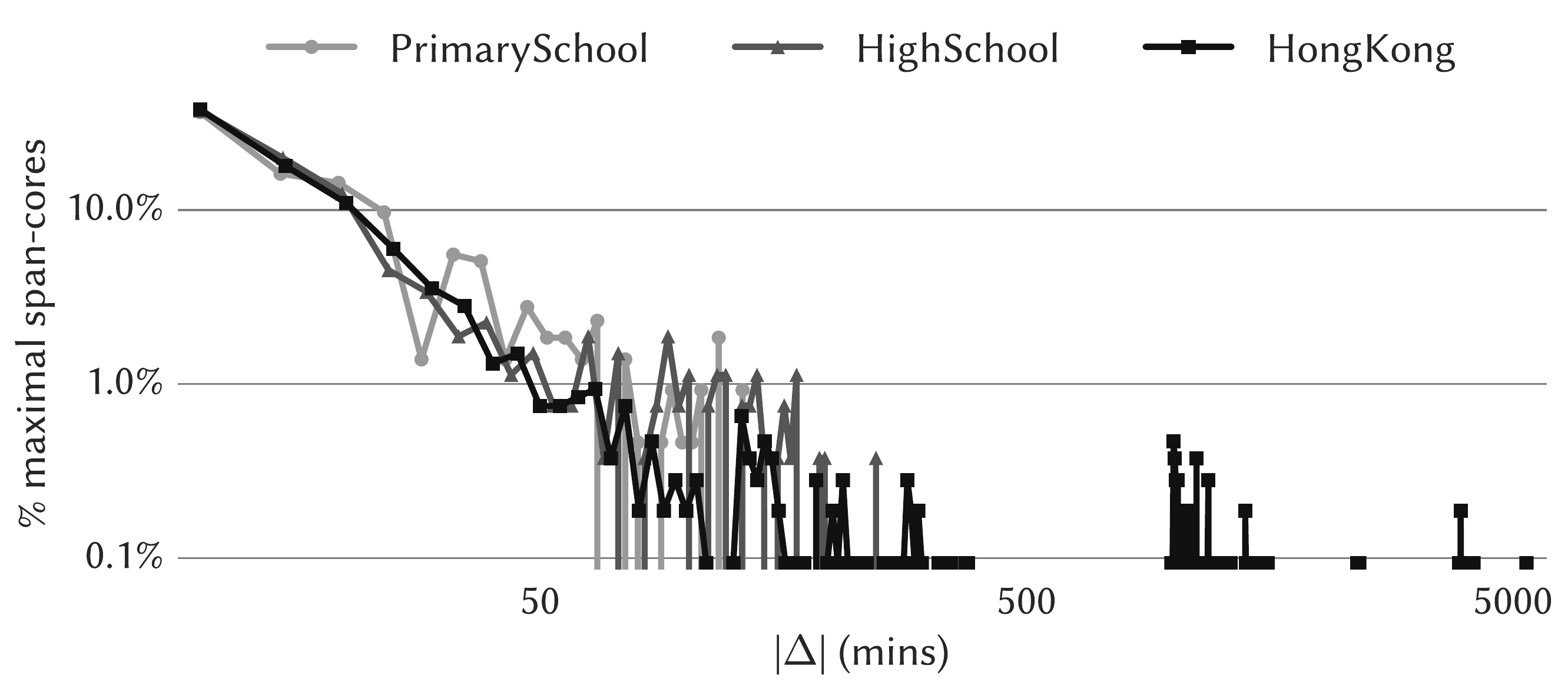}}
\caption{\label{fig:lengthdistribution} Distribution of the size of the span $|\Delta|$ of the maximal \spancores. The $x$ axis reports the size of the span (in minutes), while the $y$ axis the percentage of maximal \spancores having a given size of the span.}
\end{figure}

\subsection{Anomaly detection}
\label{sec:anomaly}
The identification of anomalous behaviors in temporal networks has been the focus of several studies in the last few years~\cite{mongiovi2013netspot, sapienza2015detecting}.
Based on the above findings, we devise a simple procedure to detect anomalous edges and intervals of the \textsf{HongKong} dataset that exploits maximal \spancores.
The topmost plot of Figure~\ref{fig:anomaly} reports the number of edges for each timestamp of the original \textsf{HongKong} dataset.
It is easy to notice that there is a lot of constant anomalous activity between school days and during the weekend, i.e., days six and seven:
unexpectedly, the number of interactions per timestamp does not drop to zero. This happened in fact
because proximity sensors were left in each class and close to each other, at the end of the lessons.
In order to automatically detect these steady activity patterns that do not correspond to any genuine social dynamics,
we apply the following procedure: $(i)$ find a set of anomalously long temporal intervals supporting maximal \spancores, $(ii)$ identify anomalous vertices, and, $(iii)$ filter out anomalous edges.

The first step of this procedure requires to find the set of temporal intervals $\mathcal{I} = \{\Delta \sqsubseteq T \mid C_{k,\Delta} \in \imcores \land |\Delta| > tr \}$ that are the span of a maximal \spancore $C_{k,\Delta}$ with size longer than a certain threshold $tr$.
Then, for each timestamp $t \in T$, select as anomalous all those vertices that appear in the \spancores $\{C_{1,\Delta} \mid \Delta \in \mathcal{I} \land t \in \Delta\}$, i.e., the \spancores of $k=1$ whose span is in $\mathcal{I}$ and contains~$t$.
Finally, at each timestamp $t \in T$, remove edges that are incident to at least a vertex that has been marked as anomalous at time $t$.
Consistently with the distribution of the span durations of the maximal \spancores, we select the threshold $tr = 22$ ($110$ minutes).
The results of this filtering procedure are shown in the middle plot of Figure~\ref{fig:anomaly}.
The number of edges during school days remains approximately unchanged, while the activity noticeably decreases in-between.
Identifying as positives the spurious interactions occurring when the school is closed and as negatives the genuine interactions observed when the school is open, this approach achieves a precision of $0.91$ and a recall of $0.64$.

We can refine this anomaly detection process by identifying, in addition to anomalous edges, also anomalous temporal intervals.
We define a timestamp $t \in T$ as anomalous if the ratio between the number of original edges (top plot of Figure~\ref{fig:anomaly}) and the number of filtered edges (middle plot of Figure~\ref{fig:anomaly}) exceeds a given threshold.
We apply this further filtering to the \textsf{HongKong} dataset with a threshold of $1.5$ and report the results in the bottommost plot of Figure~\ref{fig:anomaly}.
The number of edges when the school is closed drops to zero, while the activity during school days is not modified, except for the last one, which is affected by the proximity to the end of the time domain.
The overall procedure yields a slightly higher value of precision, $0.93$, and substantially improves the recall to $0.99$.

\begin{figure}
\centerline{\includegraphics[width=0.7\columnwidth]{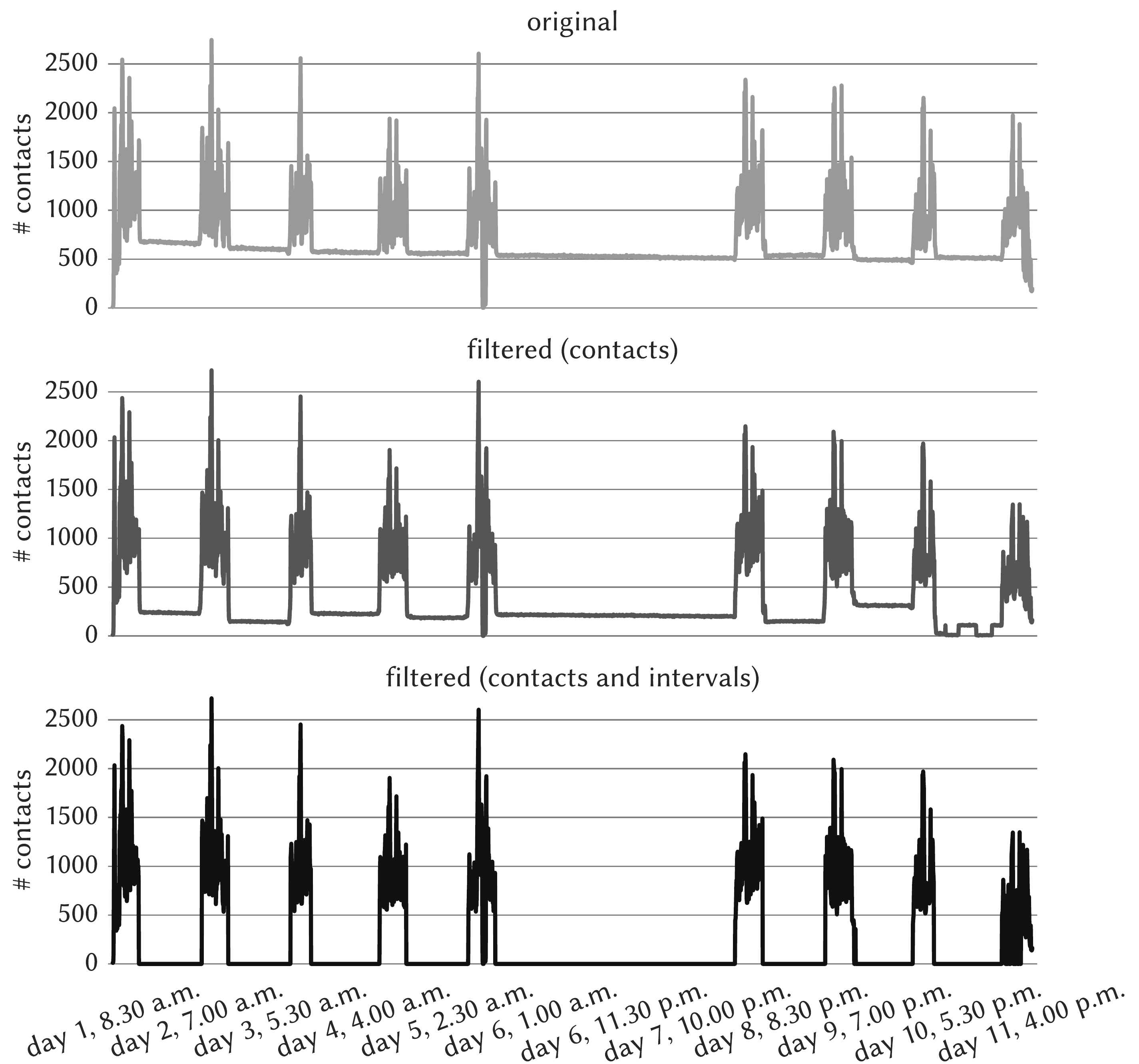}}
\caption{\label{fig:anomaly} \textsf{HongKong} dataset: number of edges  per timestamp in the original data (top), after filtering anomalous edges (middle), and after filtering anomalous edges and intervals (bottom).
Days 6 and 7 are weekend.}
\end{figure}

\subsection{Graph embedding and vertex classification}
\label{sec:classification}
In this subsection we show how \temporalcs can be profitably exploited for classifying the vertices of a temporal graph.
Specifically, the classification framework we set up is based on the paradigm of \emph{graph embedding},
which has attracted a great deal of attention in the last few years, and whose goal is to assign to
every vertex of a graph a numerical vector (i.e., an \emph{embedding}) such that structurally similar vertices are
represented by similar vectors, and vice versa~\cite{grover2016node2vec,epasto2019single,GOYAL201878}.
Here, our framework simply consists in learning suitable embeddings for the vertices of the input graph, and then give them as input to some (well-established) classifier to ultimately accomplish the desired classification task.
Thus, the main goal is to learn embeddings that are well-representative of the relationships among vertices, so as to help the classifier perform accurately.
As our main result here, we show how an embedding strategy based on a simple exploitation of the output of \temporalcs achieves results comparable to
well-established vertex-embedding methods
such as \deepwalk~\cite{perozzi2014deepwalk}, \LINE~\cite{tang2015line}, and \nodevec~\cite{grover2016node2vec}.

\spara{Method.}
For every vertex of the input temporal graph, we build an embedding as an $h$-dimensional vector conveying the information provided by a solution to the \temporalcs problem on the same graph.
Specifically, consider a vertex $u \in V $ and a solution $\{\langle S_i, \Delta_i \rangle \}_{i=1}^h$ to \temporalcs on query-vertex set $Q = \{u\}$.
We define $u$'s embedding as
$$
\mathbf{X}_u = [v^*_{Q, \Delta_1}, v^*_{Q, \Delta_2}, \ldots, v^*_{Q, \Delta_h}],
$$
which corresponds to the temporally-ordered sequence of minimum degrees of the $h$ communities identified by the temporal-community-search solution.
Below we show that this simple approach is sufficient to achieve interesting experimental results.
Clearly, more sophisticated methods are possible, e.g., by simultaneously exploiting  information from the $S_i$ communities.
However, our main goal here is to give an idea of how the \temporalcs problem can be successfully leveraged in a relevant application scenario, rather than devise the best temporal-community-search-based graph-embedding method.

\begin{figure}
	\begin{tabular}{c}
		\centerline{\includegraphics[width=0.7\columnwidth]{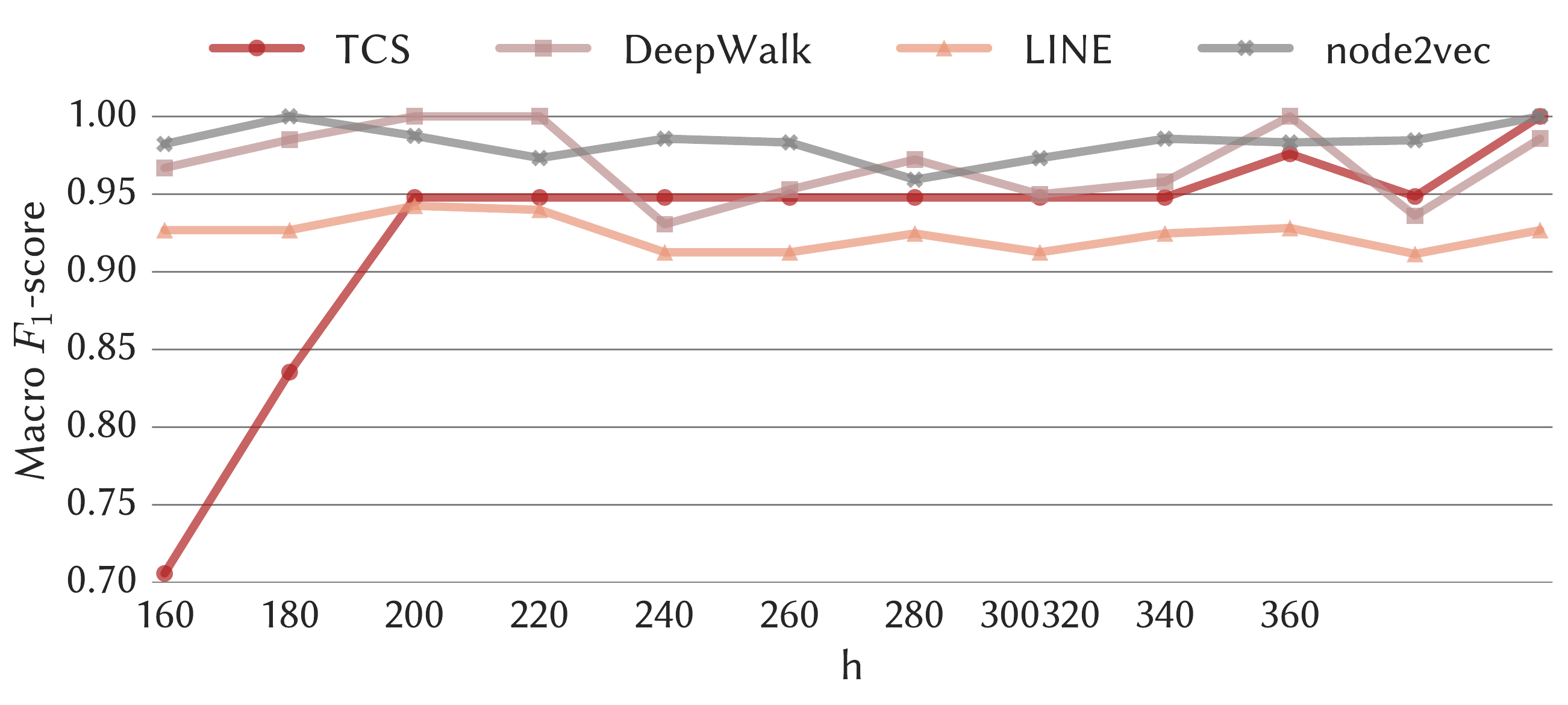}} \vspace{-2mm}\\
		\footnotesize{\textsf{PrimarySchool}} \vspace{4mm}\\
		\centerline{\includegraphics[width=0.7\columnwidth]{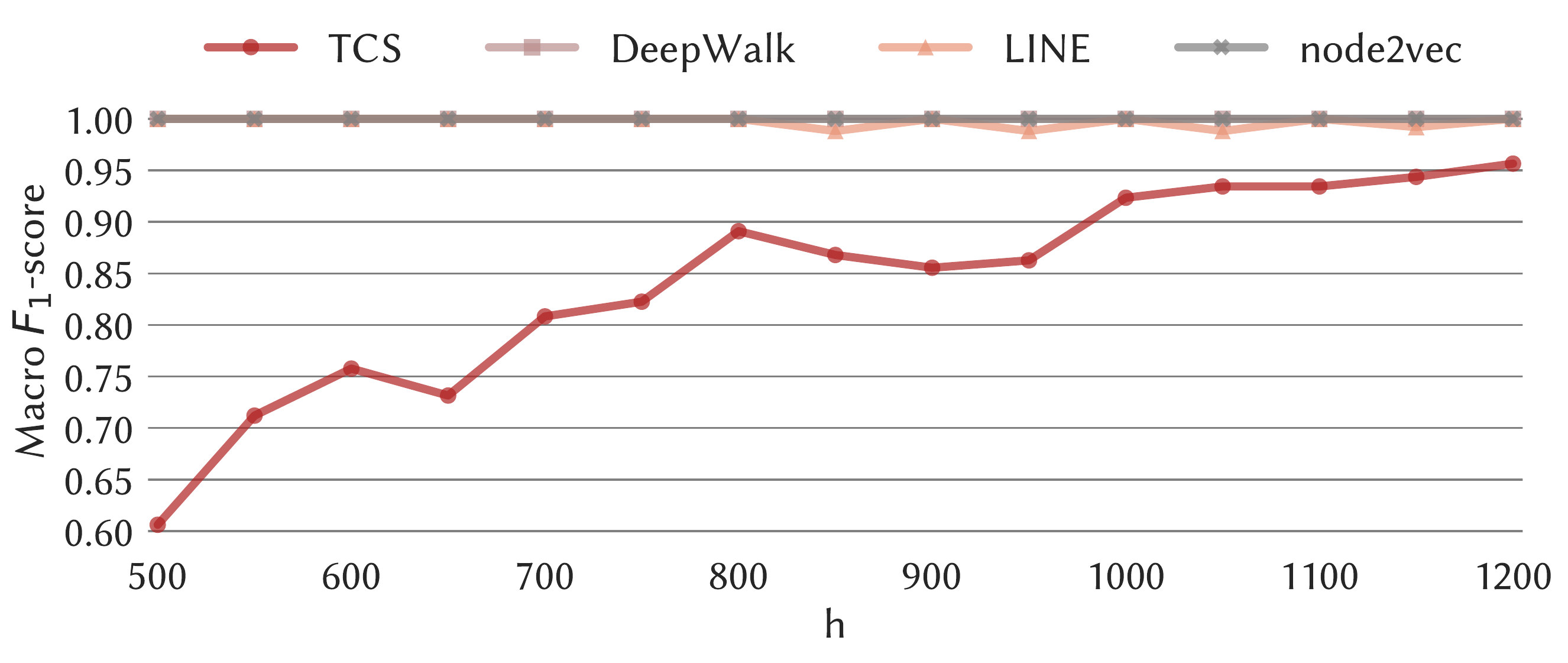}} \vspace{-2mm}\\
		\footnotesize{\textsf{HighSchool}} \\
	\end{tabular}
	\caption{ \label{fig:f1scores}
	{	
		 Vertex classification: Macro $F_1$-score of the proposed temporal-community-search-based graph-embedding method \tcs and the competing methods, with varying the dimensionality $h$ of the output embeddings, on the \textsf{PrimarySchool} and \textsf{HighSchool} datasets.}
	}
\end{figure}

\spara{Evaluation.}
We assess the performance of our method on the \textsf{PrimarySchool} and \textsf{HighSchool} datasets.
In these datasets vertices correspond to students, and vertex labels (to be predicted) are the classes that every student belongs to.
We involve in the comparison the following state-of-the-art vertex-embedding methods:
\begin{itemize}
	\item \deepwalk~\cite{perozzi2014deepwalk}, a method that preserves the proximity between vertices by running a set of random walks and maximizing the sum of the log-likelihood of a set of vertices for each walk.
	\item \LINE~\cite{tang2015line}, which optimizes a suitable objective function preserving both first-order (one-hop) and second-order (two-hop) proximities.
	Neighborhoods are not explored via random walk, but in a breadth-first fashion.
	\item \nodevec~\cite{grover2016node2vec}, which is based on the same idea underlying \deepwalk, but allowing more flexibility on how random walks
	explore and leave the neighborhood of the current vertex.
\end{itemize}
These three methods consider non-temporal graphs.
Therefore, we feed them with aggregated graphs in which every edge exists if it exists in at least one timestamp.

{
	Exhaustive grid search is carried out to tune parameters of
	 \nodevec~\cite{grover2016node2vec} and \deepwalk~\cite{perozzi2014deepwalk}.
	In particular, for both methods, we tune the neighborhood parameters of a vertex, i.e.,  number of walks $r$, and walk length $l$, while the neighborhood size $k$ is set to $10$.
	Furthermore, for \nodevec we tune the return and in-out parameters $p$ and $q$.
	For each parameter, we use the same grid of values as the one considered in the parameter sensitivity analysis reported in the original \nodevec paper~\cite{grover2016node2vec}.
	Specifically, we consider the following parameter space:
	
	\begin{itemize}
		
		 \item Number of walks, $r = \{6,8,10,12,14,16,18,20\}$;
		
		 \item Walk length, $l = \{ 30,40,50,60,70,80,90,100,110 \}$;
		
	
	 	\item Return parameter, $p =  \{ 0.25, 0.5, 1 , 2, 4 \}$;
	 	
	 	\item In-out parameter, $q =  \{ 0.25, 0.5, 1 , 2, 4 \}$.
		
	\end{itemize}

	We select the combination of parameters maximizing the Macro $F_1$-score averaged over a range of numbers of
	latent dimensions $d = \{16,32,64,128,256\}$. 
}

%
After filtering out those vertices representing the teachers, we partition the remaining vertices
(i.e., the students) into training and test sets with an 80-20 split.
A standard scaler is applied to the features extracted by each embedding method and, then, a penalized logistic-regression classifier is trained.

In Figure \ref{fig:f1scores} we report classification results in terms of Macro $F_1$-score, with varying the dimensionality $h$ of the embeddings.
On the \textsf{PrimarySchool} dataset, for $h \geq 200$, our \tcs has performance close to $1$ in terms Macro $F_1$-score, similarly to the three baselines.
It can be observed that the \tcs results are better as $h$ gets higher;
in particular, \tcs is even better than \nodevec for $h = |T|$.
This is expected and is motivated as, for higher $h$, \tcs is allowed to rely on more temporal information about the vertices.
On the \textsf{HighSchool} dataset, \tcs is outperformed by all methods for smaller $h$.
However, again, the performance of \tcs becomes competitive for larger $h$, up to achieving comparable results to the best method(s) for $h = |T|$.

\section{Conclusions}
\label{sec:conclusions}

Temporal networks are a powerful representation of how relations are established and interrupted along time among a given population of entities.
An interesting primitive for analyzing this type of networks is the extraction of relevant patterns,
such as dense subgraphs, together with their time interval of existence (or span).
Following this idea, we introduced in this paper a notion of temporal core decomposition where each core is associated with its span. Exploiting containment properties among cores we developed efficient algorithms for computing all the span-cores, and also only the maximal ones.
We then introduced the problem of temporal community search and showed how it  can be solved in polynomial time via dynamic programming.
We also proved an interesting connection between temporal community search and maximal span-cores, which made it possible
to devise a considerably more efficient algorithm than the na\"ive dynamic-programming one.
Finally, we presented  applications on empirical networks of human close-range proximity, that illustrate the relevance of the notions of (maximal) span-core and temporal community search in a variety analyses and applications.

In future work we will study the role of maximal span-cores with large coreness and/or $|\Delta|$ in spreading processes on temporal networks.
Furthermore,  span-cores represent features that can be used
for network fingerprinting and classification, as well as for model validation, and that could provide
support for new ways of visualizing large-scale time-varying graphs.
We also plan to investigate different semantics of temporal patterns and the corresponding notions of core 
decompositions extracted from temporal networks, such as, for instance, the ones that might arise by considering a temporal edge existing in a given interval if it appears in at least one of the timestamps of the interval, or in a fraction of timestamps larger than a given threshold. 
Finally, investigation in temporal community search has just started: we plan to study different notions of community search and their corresponding extraction problems. As an example, a variant of the notion adopted in this article, which is worth to be investigated further, would correspond to relaxing the requirement of covering the whole temporal domain $T$, and instead looking for a set of communities giving a good enough temporal coverage.

%

\bibliographystyle{ACM-Reference-Format}



\end{document}